\documentclass[12pt,reqno] {amsart}
\usepackage{amsmath, wasysym}
\usepackage{bbm, dsfont}
\usepackage{amssymb}
\usepackage{xypic,color}
\usepackage{graphicx}%
\usepackage{hyperref}
\usepackage{cleveref}
\usepackage{comment}
 \usepackage{mathrsfs}
 \usepackage{float}
 \usepackage{tikz, tikz-cd}
  \usetikzlibrary{arrows.meta, knots, calc}
 \usepackage{bm}
 \usepackage{mathtools} 
\usepackage{extarrows}

 \usepackage{ulem}
 \usepackage{enumerate}
  \usepackage{mathdots}
\theoremstyle{plain}

\newtheorem{question}{Question}

\numberwithin{equation}{section}

\newtheorem{theorem}{Theorem}[section]
\newtheorem{corollary}[theorem]{Corollary}
\newtheorem{lemma}[theorem]{Lemma}
\newtheorem{proposition}[theorem]{Proposition}
\newtheorem{remark}[theorem]{Remark}
\newtheorem{definition}[theorem]{Definition}
\newtheorem{example}[theorem]{Example}
\newtheorem*{ac}{Acknowledgement}

\newcommand{\bZ}{\mathbb{Z}}
\newcommand{\bN}{\mathbb{N}}
\newcommand{\bC}{\mathbb{C}}

\newcommand{\bS}{\mathbb{S}}

\newcommand{\cA}{\mathcal{A}}  
  
\newcommand{\cC}{\mathcal{C}}

\renewcommand{\hom}{\mathrm{Hom}}

\newcommand{\Irr}{\mathrm{Irr}}
\newcommand{\Hom}{\mathrm{HOM}}

\newcommand{\id}{\mathrm{id}}

\newcommand{\Cob}{\mathrm{ACob}_2^{\cC}}

\renewcommand{\k}{\mathds{k}}
\newcommand{\V}{\mathbb{V}}
\newcommand{\M}{\mathbb{M}}
\newcommand{\darkgreen}{green!60!black}
\newcommand{\bF}{\mathbb{F}}
\newcommand{\acolor}{white!70!gray}
\colorlet{FA}{brown!70!black}

\tikzset{
    partial ellipse/.style args={#1:#2:#3}{
        insert path={+ (#1:#3) arc (#1:#2:#3)}
    }
}

\usepackage{dsfont}
\renewcommand{\1}{\mathds{1}}
\renewcommand{\Vec}{\mathrm{Vec}}
\newcommand{\BM}{\mathbb{M}}
\newcommand{\MCG}{\operatorname{MCG}}
\newcommand{\SL}{\operatorname{SL}}
\newcommand{\bT}{\mathbb{T}}
\usetikzlibrary{decorations.markings}

\tikzset{->-/.style={decoration={
  markings,
  mark=at position #1 with {\arrow{>}}},postaction={decorate}}}
\tikzset{-<-/.style={decoration={
  markings,
  mark=at position #1 with {\arrow{<}}},postaction={decorate}}}

\newcommand{\KirbyZero}{\raise -4.5em
\hbox{
\begin{tikzpicture}
\begin{knot}
\begin{scope}
  \draw (-1,0.3) ellipse (0.3 and 0.1);
  \draw (-1.3,0.3)--(-1.3,-1.2); \draw (-0.7,0.3)--(-0.7,-1.2);
  \draw[blue, ->-=0.7] (-1, 0.2)--(-1,-1.2);
\end{scope}
  \draw (-2,-1) arc (180:360:1 and 0.4);
  \draw (-2,-1) arc (180:120:1 and 0.4);
  \draw (0,-1) arc (0:60:1 and 0.4);
  \draw (-2.4,-1) arc (180:360:1.4 and 0.8);
  \draw (-2.4,-1) arc (180:110:1.4 and 0.8);
  \draw (0.4,-1) arc (0:70:1.4 and 0.8);
  \draw[red] (-2.2,-1) arc (180:360:1.2 and 0.6);
  \draw[red] (-2.2,-1) arc (180:115:1.2 and 0.6);
  \draw[red] (0.2,-1) arc (0:65:1.2 and 0.6);
\begin{scope}
  \draw[dotted] (-1.3,-2.8) arc (180:0:0.3 and 0.1);
  \draw (-1.3,-2.8) arc (180:360:0.3 and 0.1);
  \draw (-1.3,-2)--(-1.3,-2.8); \draw (-0.7,-2)--(-0.7,-2.8);
  \draw[blue] (-1, -2)--(-1,-2.9);
\end{scope}
\begin{scope}[shift={(2.4,0)}]
  \draw (-2.4, -1) arc (180:0:0.2 and 0.1);
  \draw[dotted] (-2.4, -1) arc (180:360:0.2 and 0.1);
\end{scope}
\end{knot}
\end{tikzpicture}}
\quad \longrightarrow \quad
\raise -4.5em
\hbox{
\begin{tikzpicture}
\begin{knot}
  \draw (-1,0.3) ellipse (0.3 and 0.1);
  \draw (-1.3,0.3)--(-1.3,-0.7);
  \draw (-0.7,0.3)--(-0.7,-1.3);
  \draw[blue, ->-=0.7] (-1, 0.2)--(-1,-1.2);
  \strand (-1.3,-0.7) to[out=down,in=right] (-1.85,-0.9);
  \draw (-1.85,-0.9) arc (180:145:1.85 and 0.4);
  \strand (-1.85,-1.1) to[out=right,in=up] (-1.3,-1.3);
  \draw (-1.85,-1.1) arc (180:270:0.85 and 0.3);
  \draw[red] (-1.6,-0.9) arc (90:270:0.05 and 0.1);
  \draw[dotted, red] (-1.6,-0.9) arc (90:-90:0.05 and 0.1);
\begin{scope}
  \draw (0,-1) arc (0:60:1 and 0.4);
  \draw (0,-1) arc (0:-90:1 and 0.4); 
  \draw (-2.4,-1) arc (180:360:1.4 and 0.8);
  \draw (-2.4,-1) arc (180:110:1.4 and 0.8);
  \draw (0.4,-1) arc (0:70:1.4 and 0.8);
  \draw[red] (-2.2,-1) arc (180:360:1.2 and 0.6);
  \draw[red] (-2.2,-1) arc (180:115:1.2 and 0.6);
  \draw[red] (0.2,-1) arc (0:65:1.2 and 0.6);
\end{scope}
\begin{scope}
  \draw[dotted] (-1.3,-2.8) arc (180:0:0.3 and 0.1);
  \draw (-1.3,-2.8) arc (180:360:0.3 and 0.1);
  \draw (-1.3,-2)--(-1.3,-2.8); \draw (-0.7,-2)--(-0.7,-2.8);
  \draw[blue] (-1, -2)--(-1,-2.9);
\end{scope}
\begin{scope}[shift={(2.4,0)}]
  \draw (-2.4, -1) arc (180:0:0.2 and 0.1);
  \draw[dotted] (-2.4, -1) arc (180:360:0.2 and 0.1);
\end{scope}
\end{knot}
\end{tikzpicture}}
\quad \longrightarrow \quad
\raise -4.5em
\hbox{
\begin{tikzpicture}
\begin{knot}
  \draw (-1,0.3) ellipse (0.3 and 0.1);
  \draw (-1.3,0.3)--(-1.3,-0.7);
  \draw (-0.7,0.3)--(-0.7,-1.2);
  \draw[blue, ->-=0.7] (-1, 0.2)--(-1,-1.2);
  \strand (-1.3,-0.7) to[out=down,in=right] (-1.85,-0.9);
  \draw (-1.85,-0.9) arc (180:145:1.85 and 0.4);
  \strand (-1.85,-1.1) to[out=right,in=up] (-1.3,-1.3);  
  \draw (-1.85,-1.1) arc (180:270:0.85 and 0.3);
  \draw[green] (-1.75, -0.87) arc (180:150:1.75 and 0.25);
  \strand[green] (-1.75, -0.87)--(-1.5,-0.85);
  \draw[green] (-1.78, -1.15) arc (180:270:0.78 and 0.2);
  \strand[green] (-1.78,-1.15) to[out=right,in=left] (-1, -1.25);
  \draw[green] (-1,-1.25) arc (-90:50:0.65 and 0.25);
  \draw[green] (-1, -1.35) arc (-90:55:0.8 and 0.35);
\begin{scope}
  \draw (0,-1) arc (0:60:1 and 0.4);
  \draw (0,-1) arc (0:-90:1 and 0.4); 
  \draw (-2.4,-1) arc (180:360:1.4 and 0.8);
  \draw (-2.4,-1) arc (180:110:1.4 and 0.8);
  \draw (0.4,-1) arc (0:70:1.4 and 0.8);
  \draw[red] (-2.2,-1) arc (180:360:1.2 and 0.6);
  \draw[red] (-2.2,-1) arc (180:115:1.2 and 0.6);
  \draw[red] (0.2,-1) arc (0:65:1.2 and 0.6);
\end{scope}
\begin{scope}
  \draw[dotted] (-1.3,-2.8) arc (180:0:0.3 and 0.1);
  \draw (-1.3,-2.8) arc (180:360:0.3 and 0.1);
  \draw (-1.3,-2)--(-1.3,-2.8); \draw (-0.7,-2)--(-0.7,-2.8);
  \draw[blue] (-1, -2)--(-1,-2.9);
\end{scope}
\begin{scope}[shift={(2.4,0)}]
  \draw (-2.4, -1) arc (180:0:0.2 and 0.1);
  \draw[dotted] (-2.4, -1) arc (180:360:0.2 and 0.1);
\end{scope}
\end{knot}
\node[fill, white, inner sep = 1pt] at (-1.32,-1.265) {};
\fill[white] (-1.32, -1.256) circle (1pt);
\begin{scope}[shift={(-1.7,-1.15)}, xscale=0.14, yscale=0.14]
  \draw[green] (0,0)--(1,-1);
\end{scope}
\begin{scope}[shift={(-1.4,-1.2)}, xscale=0.14, yscale=0.14]
  \draw[green] (0,0)--(1,-1);
\end{scope}
\begin{scope}[shift={(-1.1,-1.25)}, xscale=0.1, yscale=0.1]
  \draw[green] (0,0)--(1,-1);
\end{scope}
\begin{scope}[shift={(-0.8,-1.24)}, xscale=0.09, yscale=0.09]
  \draw[green] (0,0)--(1,-1);
\end{scope}
\begin{scope}[shift={(-0.5,-1.16)}, xscale=0.08, yscale=0.08]
  \draw[green] (0,0)--(1,-1);
\end{scope}
\begin{scope}[shift={(-1.66,-0.79)}, xscale=0.07, yscale=0.07]
  \draw[green] (0,0)--(1,-1);
\end{scope}
\draw[green] (-0.4,-0.9)--(-0.3,-0.83);
\draw[green] (-0.37,-1.05)--(-0.215,-1.06);
\end{tikzpicture}}
\quad \longrightarrow \quad D\cdot \sum 
\raise -4.5em
\hbox{
\begin{tikzpicture}
\begin{knot}
  \node (A) [draw] at (-1, -0.8) {\tiny{$\varphi$}};
  \node (B) [draw] at (-1, -1.8) {\tiny{$\varphi$}};
  \draw (-1,0.3) ellipse (0.3 and 0.1);
  \draw (-1.3,0.3)--(-1.3,-2.8); \draw (-0.7,0.3)--(-0.7,-2.8);
  \draw[blue, ->-=0.7] (-1, 0.2)--(A.north);
  \draw[blue] (B.south)--(-1,-2.9);
\begin{scope}
  \draw[dotted] (-1.3,-2.8) arc (180:0:0.3 and 0.1);
  \draw (-1.3,-2.8) arc (180:360:0.3 and 0.1);
\end{scope}
\begin{scope}[red, shift={(0,1.5)}]
  \draw[dotted] (-1.3,-2.8) arc (180:0:0.3 and 0.1);
  \draw (-1.3,-2.8) arc (180:360:0.3 and 0.1);
\end{scope}
\end{knot}
\end{tikzpicture}
}}

\usetikzlibrary{arrows.meta}

\makeatletter

\pgfdeclarearrow{
  name = ipe _linear,
  defaults = {
    length = +1bp,
    width  = +.666bp,
    line width = +0pt 1,
  },
  setup code = {
    \pgfarrowssetbackend{0pt}
    \pgfarrowssetvisualbackend{
      \pgfarrowlength\advance\pgf@x by-.5\pgfarrowlinewidth}
    \pgfarrowssetlineend{\pgfarrowlength}
    \ifpgfarrowreversed
      \pgfarrowssetlineend{\pgfarrowlength\advance\pgf@x by-.5\pgfarrowlinewidth}
    \fi
    \pgfarrowssettipend{\pgfarrowlength}
    \pgfarrowshullpoint{\pgfarrowlength}{0pt}
    \pgfarrowsupperhullpoint{0pt}{.5\pgfarrowwidth}
    \pgfarrowssavethe\pgfarrowlinewidth
    \pgfarrowssavethe\pgfarrowlength
    \pgfarrowssavethe\pgfarrowwidth
  },
  drawing code = {
    \pgfsetdash{}{+0pt}
    \ifdim\pgfarrowlinewidth=\pgflinewidth\else\pgfsetlinewidth{+\pgfarrowlinewidth}\fi
    \pgfpathmoveto{\pgfqpoint{0pt}{.5\pgfarrowwidth}}
    \pgfpathlineto{\pgfqpoint{\pgfarrowlength}{0pt}}
    \pgfpathlineto{\pgfqpoint{0pt}{-.5\pgfarrowwidth}}
    \pgfusepathqstroke
  },
  parameters = {
    \the\pgfarrowlinewidth,%
    \the\pgfarrowlength,%
    \the\pgfarrowwidth,%
  },
}

\pgfdeclarearrow{
  name = ipe _pointed,
  defaults = {
    length = +1bp,
    width  = +.666bp,
    inset  = +.2bp,
    line width = +0pt 1,
  },
  setup code = {
    \pgfarrowssetbackend{0pt}
    \pgfarrowssetvisualbackend{\pgfarrowinset}
    \pgfarrowssetlineend{\pgfarrowinset}
    \ifpgfarrowreversed
      \pgfarrowssetlineend{\pgfarrowlength}
    \fi
    \pgfarrowssettipend{\pgfarrowlength}
    \pgfarrowshullpoint{\pgfarrowlength}{0pt}
    \pgfarrowsupperhullpoint{0pt}{.5\pgfarrowwidth}
    \pgfarrowshullpoint{\pgfarrowinset}{0pt}
    \pgfarrowssavethe\pgfarrowinset
    \pgfarrowssavethe\pgfarrowlinewidth
    \pgfarrowssavethe\pgfarrowlength
    \pgfarrowssavethe\pgfarrowwidth
  },
  drawing code = {
    \pgfsetdash{}{+0pt}
    \ifdim\pgfarrowlinewidth=\pgflinewidth\else\pgfsetlinewidth{+\pgfarrowlinewidth}\fi
    \pgfpathmoveto{\pgfqpoint{\pgfarrowlength}{0pt}}
    \pgfpathlineto{\pgfqpoint{0pt}{.5\pgfarrowwidth}}
    \pgfpathlineto{\pgfqpoint{\pgfarrowinset}{0pt}}
    \pgfpathlineto{\pgfqpoint{0pt}{-.5\pgfarrowwidth}}
    \pgfpathclose
    \ifpgfarrowopen
      \pgfusepathqstroke
    \else
      \ifdim\pgfarrowlinewidth>0pt\pgfusepathqfillstroke\else\pgfusepathqfill\fi
    \fi
  },
  parameters = {
    \the\pgfarrowlinewidth,%
    \the\pgfarrowlength,%
    \the\pgfarrowwidth,%
    \the\pgfarrowinset,%
    \ifpgfarrowopen o\fi%
  },
}

\newdimen\ipeminipagewidth

\tikzstyle{ipe import} = [
  x=1bp, y=1bp,
  ipe node stretch/.store in=\ipenodestretch,
  ipe stretch normal/.style={ipe node stretch=1},
  ipe stretch normal,
  ipe node/.style={
    anchor=base west, inner sep=0, outer sep=0, scale=\ipenodestretch
  },
  ipe mark scale/.store in=\ipemarkscale,
  ipe mark tiny/.style={ipe mark scale=1.1},
  ipe mark small/.style={ipe mark scale=2},
  ipe mark normal/.style={ipe mark scale=3},
  ipe mark large/.style={ipe mark scale=5},
  ipe mark normal, 
  ipe circle/.pic={
    \draw[line width=0.2*\ipemarkscale]
      (0,0) circle[radius=0.5*\ipemarkscale];
    \coordinate () at (0,0);
  },
  ipe disk/.pic={
    \fill (0,0) circle[radius=0.6*\ipemarkscale];
    \coordinate () at (0,0);
  },
  ipe fdisk/.pic={
    \filldraw[line width=0.2*\ipemarkscale]
      (0,0) circle[radius=0.5*\ipemarkscale];
    \coordinate () at (0,0);
  },
  ipe box/.pic={
    \draw[line width=0.2*\ipemarkscale, line join=miter]
      (-.5*\ipemarkscale,-.5*\ipemarkscale) rectangle
      ( .5*\ipemarkscale, .5*\ipemarkscale);
    \coordinate () at (0,0);
  },
  ipe square/.pic={
    \fill
      (-.6*\ipemarkscale,-.6*\ipemarkscale) rectangle
      ( .6*\ipemarkscale, .6*\ipemarkscale);
    \coordinate () at (0,0);
  },
  ipe fsquare/.pic={
    \filldraw[line width=0.2*\ipemarkscale, line join=miter]
      (-.5*\ipemarkscale,-.5*\ipemarkscale) rectangle
      ( .5*\ipemarkscale, .5*\ipemarkscale);
    \coordinate () at (0,0);
  },
  ipe cross/.pic={
    \draw[line width=0.2*\ipemarkscale, line cap=butt]
      (-.5*\ipemarkscale,-.5*\ipemarkscale) --
      ( .5*\ipemarkscale, .5*\ipemarkscale)
      (-.5*\ipemarkscale, .5*\ipemarkscale) --
      ( .5*\ipemarkscale,-.5*\ipemarkscale);
    \coordinate () at (0,0);
  },
  /pgf/arrow keys/.cd,
  ipe arrow normal/.style={scale=1},
  ipe arrow tiny/.style={scale=.4},
  ipe arrow small/.style={scale=.7},
  ipe arrow large/.style={scale=1.4},
  ipe arrow normal,
  /tikz/.cd,
  ipe arrows/.style={
    ipe normal/.tip={
      ipe _pointed[length=1bp, width=.666bp, inset=0bp,
                   quick, ipe arrow normal]},
    ipe pointed/.tip={
      ipe _pointed[length=1bp, width=.666bp, inset=0.2bp,
                   quick, ipe arrow normal]},
    ipe linear/.tip={
      ipe _linear[length = 1bp, width=.666bp,
                  ipe arrow normal, quick]},
    ipe fnormal/.tip={ipe normal[fill=white]},
    ipe fpointed/.tip={ipe pointed[fill=white]},
    ipe double/.tip={ipe normal[] ipe normal},
    ipe fdouble/.tip={ipe fnormal[] ipe fnormal},
    ipe arc/.tip={ipe normal},
    ipe farc/.tip={ipe fnormal},
    ipe ptarc/.tip={ipe pointed},
    ipe fptarc/.tip={ipe fpointed},
  },
  ipe arrows, 
]

\tikzset{
  rgb color/.code args={#1=#2}{%
    \definecolor{tempcolor-#1}{rgb}{#2}%
    \tikzset{#1=tempcolor-#1}%
  },
}

\makeatother

\tikzstyle{transition}=[rectangle,draw=black!50,fill=white,thick,
                        inner sep=5pt,minimum size=5mm]
\begin{document}

\title[ATQFT]{Alterfold Topological Quantum Field Theory}

\author{Zhengwei Liu}
\address{Z. Liu, Yau Mathematical Sciences Center and Department of Mathematics, Tsinghua University, Beijing, 100084, China}
\address{Yanqi Lake Beijing Institute of Mathematical Sciences and Applications, Huairou District, Beijing, 101408, China}
\email{liuzhengwei@mail.tsinghua.edu.cn}

\author{Shuang Ming}
\address{S. Ming, Yanqi Lake Beijing Institute of Mathematical Sciences and Applications, Beijing, 101408, China}
\email{sming@bimsa.cn}

\author{Yilong Wang}
\address{Y. Wang, Yanqi Lake Beijing Institute of Mathematical Sciences and Applications, Beijing, 101408, China}
\email{wyl@bimsa.cn}

\author{Jinsong Wu}
\address{J. Wu, Yanqi Lake Beijing Institute of Mathematical Sciences and Applications, Beijing, 101408, China}
\email{wjs@bimsa.cn}

\maketitle

\begin{abstract}
We introduce the 3-alterfold topological quantum field theory (TQFT) by extending the quantum invariant of 3-alterfolds. 
The bases of the TQFT are explicitly characterized and the Levin-Wen model is naturally interpreted in 3-alterfold TQFT bases.
By naturally considering the RT TQFT and TV TQFT as sub-TQFTs within the 3-alterfold TQFT, we establish their equivalence. The 3-alterfold TQFT is unitary when the input fusion category is unitary. Additionally, we extend the 3-alterfold TQFT to the Morita context and demonstrate that Morita equivalent fusion categories yield equivalent TV TQFTs. We also provide a simple pictorial proof of complete positivity criteria for unitary categorization when the 3-alterfold TQFT is unitary. Expanding our scope to high-genus surfaces by replacing the torus, we introduce the high genus topological indicators and proving the equivariance under the mapping class group actions.
\end{abstract}

\section{Introduction}

In \cite{LMWW23}, the authors introduce 3-alterfolds and explore their quantum invariants. 
A 3-alterfold is defined as a 3-manifold with an associated separating surface. 
Analyzing the triangulation and Dehn surgery in a 3-alterfold, we establish the natural equivalence between the Reshetikhin-Turaev (RT) invariant \cite{ResTur91} and the Turaev-Viro (TV) invariant \cite{TurVir92, Tur94, GelKaz96, BarWes99}. In this paper, we further extend the partition function of decorated $3$-alterfolds to a topological quantum field theory in the sense of Atiyah \cite{Ati88}, by applying the universal construction introduced by Blanchet, Habegger, Masbaum and Vogel \cite{BHMV2} in 1992.

To extend this concept to a topological quantum field theory (TQFT), we introduce the notion of a 3-alterfold with (time) boundary, where the separating surface is referred to as the space boundary. 
The (time) boundary forms a 2-alterfold. 
The cobordism category \cite{Ati88} is then extended to the alterfold cobordism category. 
We demonstrate the existence of a strict symmetric monoidal functor from the alterfold cobordism category to the category of finite dimensional vector spaces, essentially constituting the 3-alterfold topological quantum field theory.
The Levin-Wen model \cite{LevWen05} is naturally explained in 3-alterfold TQFT. 
More precisely, we identify the ground state degeneracy in the 3-alterfold TQFT basis for the corresponding surface and the plaquette operator is topologically interpreted by alterfold cobordism.
We obtain the TV base and the RT base for the 3-alterfold TQFT.
When the input spherical fusion category is unitary, we have that the 3-alterfold TQFT is unitary.
Consequently, we obtain the complete positivity property of the spherical fusion category pictorially \cite{HLPW23}, which provides efficient criteria for unitary categorification.

Observing that the quantum invariant of 3-alterfolds encompasses both the RT invariant and TV invariant, it becomes evident that the RT TQFT and TV TQFT naturally arise as sub-TQFTs within our proposed framework, and we establish their equivalence.

By incorporating colors into the surface, we observe that the 3-alterfold TQFT is applicable to 2-categories. 
We further note that TV invariants remain the same for Morita equivalent spherical fusion categories. 
When the 2-category is unitary, we obtain the complete positivity property of the module category.

In \cite{LMWW23}, we gave an alterfold theoretical interpretation of the generalized Frobenius-Schur indicators, and the $\operatorname{SL}_2(\mathbb{Z})$-equivariance of these numerical invariants follows from the homeomorphism property of the alterfold partition function. Inspired by the above work, in this paper, we introduce the high genus generalization of the Frobenius-Schur indicators, which are called the topological indicators. More precisely, a genus-$g$ topological indicator is the value of a partition function of two linked genus-$g$ handlebodies which are decorated by $\mathcal{C}$- and $\mathcal{Z(C)}$-colored diagrams. Again using the homeomorphism property of the alterfold partition functions, we obtain the equivariance of the topological indicators under the mapping class group actions.


The paper organize as follows. 
In section 2, we give a short review on 2-categories, topological quantum field theories and mapping class groups. 
In section 3, we recall the main results of $3$-alterfold theory developed in \cite{LMWW23}. 
In section 4, we introduce the TQFT by extending the $3$-alterfold partition function. 
We also provide sets of (unitary) basis to the state space using (truncated) triangulation and pants decomposition.
In section 5, we show that TV TQFT and RT TQFT are 3-alterfold sub-TQFT and they are equivalent.
In section 6, we study the $3$-alterfold decorated by tensor diagrams in a Morita context. We also present similar result as in the ordinary alterfold theory.
In section 7, we study the topological indicator of a spherical fusion category, which can be realized as a higher-genus-generalization of the Frobenius-Schur indicator. 

\begin{ac}
The authors are supported by grants from Yanqi Lake Beijing Institute of Mathematical Sciences and Applications. Z.~L. was supported by BMSTC and ACZSP (Grant No. Z221100002722017) and by NKPs (Grant no. 2020YFA0713000). Z.~L., Y.~W. and J.~W. are supported by Beijing Natural Science Foundation Key Program (Grant No. Z220002). S.~M. and J.~W. are supported by NSFC (Grant no. 12371124). Y.~W. is supported by NSFC (Grant no. 12301045). J.~W. is partially supproted by NSFC (Grant no. 12031004). 
\end{ac}

\section{Preliminaries}
In this section, we recall the basic notions of Morita context in \cite{Mug03a}, topological quantum field theory by Atiyah \cite{Ati88} and mapping class groups of surfaces.

\subsection{Graphical Calculus and Morita Context}
In this section, we review the of graphical calculus for Morita context, which is a bi-category that encodes a pair of spherical fusion categories that are Morita equivalent.
A 2-category is a strict bi-category.

\begin{definition}
A spherical Morita context is a bi-category $\mathcal{E}$ satisfying
\begin{itemize}
    \item The objects of $\mathcal{E}=\{\mathfrak{A}, \mathfrak{B}\}$.
    \item All $1$-morphism $X$ has two-sided duals, with a pivotal isomorphism $X\rightarrow X^{**}$ that is spherical.
    \item The idempotent $2$-morphisms in $\mathcal{E}$ split.
    \item There are mutually two-sided dual $1$-morphism $J: \mathfrak{B}\rightarrow \mathfrak{A}$ and $\overline{J}:\mathfrak{A}\rightarrow \mathfrak{B}$ such that $d_{J}=d_{\overline J}\ne 0$. 
\end{itemize}
\end{definition}

\begin{definition}
Two spherical fusion categories $\mathcal{C}$ and $\mathcal{D}$ are said to be Morita equivalent if there exists a Morita context $\mathcal{E}=\{\mathfrak{A, B}\}$ that $\mathcal{C}$ and $\mathcal{D}$ are equivalent to $\text{END}(\mathfrak{A})$ and $\text{END}(\mathfrak{B})$ respectively as spherical fusion categories.
\end{definition}
A Morita context is said to be strict if the bi-category is a $2$-category. 
A Morita context is always equivalent to a strict one (\cite[Remark 3.12]{Mug03a}). 
In order to do the graphical calculus, all Morita context in this paper is assumed to be strict.

A strict monoidal category can be realized as a $2$-category with one object. 
We extend our conventions in \cite{LMWW23} for spherical fusion categories to Morita contexts by labeling the regions with objects in the $2$-category. 
To be precise, $1$-morphisms in $\mathcal{E}$ are still depicted as vertical lines and we label the regions on the left and right by the source object and the target object respectively.
Therefore, the composition of $1$-morphisms are read from left to the right, for instance,
$$\begin{array}{cccc}
\raisebox{-0.8cm}{
\begin{tikzpicture}
\draw[blue, ->-=0.5] (0, 0) node [black, above] {\tiny $X$} --(0, -1) node [black, below] {\tiny $X$};
\node [right]  at (0, -0.5) {\tiny $\mathfrak{A}$};
\node [left] at (0, -0.5) {\tiny $\mathfrak{A}$};
\end{tikzpicture}}, &\raisebox{-0.8cm}{
\begin{tikzpicture}
\draw[blue, ->-=0.5] (0, 0) node [black, above] {\tiny $X'$} --(0, -1) node [black, below] {\tiny $X'$};
\node [right]  at (0, -0.5) {\tiny $\mathfrak{B}$};
\node [left] at (0, -0.5) {\tiny $\mathfrak{A}$};
\end{tikzpicture}}, &\raisebox{-0.8cm}{
\begin{tikzpicture}
\draw[blue, ->-=0.5] (0, 0) node [black, above] {\tiny $X''$} --(0, -1) node [black, below] {\tiny $X''$};
\node [right]  at (0, -0.5) {\tiny $\mathfrak{A}$};
\node [left] at (0, -0.5) {\tiny $\mathfrak{B}$};
\end{tikzpicture}}  &\quad \raisebox{-0.8cm}{
\begin{tikzpicture}
\draw[blue, ->-=0.5] (0, 0) node [black, above] {\tiny $X'''$} --(0, -1) node [black, below] {\tiny $X'''$};
\node [right]  at (0, -0.5) {\tiny $\mathfrak{B}$};
\node [left] at (0, -0.5) {\tiny $\mathfrak{B}$};
\end{tikzpicture}}.\\
X:\mathfrak{A}\rightarrow \mathfrak{A} & X':\mathfrak{A}\rightarrow \mathfrak{B}& X'':\mathfrak{B}\rightarrow \mathfrak{A}&
X''':\mathfrak{B}\rightarrow \mathfrak{B}
\end{array}$$
$2$-morphisms are still depicted as coupons and compositions are read from top to the bottom. 
Later, we will draw tensor diagrams over surfaces and put an arrow on the $1$-morphisms to help us keep track of the orientation.
$$
\begin{array}{cccc}
    \raisebox{-0.8cm}{
\begin{tikzpicture}
\draw [blue] (0, 0) --(0, -1.6);
\node[above] at (0,0) {\tiny $X$};
\node[below] at (0,-1.6) {\tiny $Y$};
\begin{scope}[shift={(-0.4, -1.1)}]
\draw [fill=white] (0, 0) rectangle (0.8, 0.6);
\node at (0.4, 0.3) {\tiny $f$};
\node [right]  at (0.8, 0.3) {\tiny $\mathfrak{A}$};
\node [left] at (0, 0.3) {\tiny $\mathfrak{A}$};
\end{scope}
\end{tikzpicture}} & \raisebox{-0.8cm}{
\begin{tikzpicture}
\draw [blue] (0, 0)  --(0, -1.6) ;
\node[above] at (0,0) {\tiny $X'$};
\node[below] at (0,-1.6) {\tiny $Y'$};
\begin{scope}[shift={(-0.4, -1.1)}]
\draw [fill=white] (0, 0) rectangle (0.8, 0.6);
\node at (0.4, 0.3) {\tiny $f'$};
\node [right]  at (0.8, 0.3) {\tiny $\mathfrak{B}$};
\node [left] at (0, 0.3) {\tiny $\mathfrak{A}$};
\end{scope}
\end{tikzpicture}}, & \raisebox{-0.8cm}{
\begin{tikzpicture}
\draw [blue] (0, 0) --(0, -1.6);
\node[above] at (0,0) {\tiny $X''$};
\node[below] at (0,-1.6) {\tiny $Y''$};
\begin{scope}[shift={(-0.4, -1.1)}]
\draw [fill=white] (0, 0) rectangle (0.8, 0.6);
\node at (0.4, 0.3) {\tiny $f''$};
\node [right]  at (0.8, 0.3) {\tiny $\mathfrak{A}$};
\node [left] at (0, 0.3) {\tiny $\mathfrak{B}$};
\end{scope}
\end{tikzpicture}}, & \raisebox{-0.8cm}{
\begin{tikzpicture}
\draw [blue] (0, 0)  --(0, -1.6) ;
\node[above] at (0,0) {\tiny $X'''$};
\node[below] at (0,-1.6) {\tiny $Y'''$};
\begin{scope}[shift={(-0.4, -1.1)}]
\draw [fill=white] (0, 0) rectangle (0.8, 0.6);
\node at (0.4, 0.3) {\tiny $f'''$};
\node [right]  at (0.8, 0.3) {\tiny $\mathfrak{B}$};
\node [left] at (0, 0.3) {\tiny $\mathfrak{B}$};
\end{scope}
\end{tikzpicture}}. \\
   f: X\rightarrow Y  & f':X'\rightarrow Y' & f'': X''\rightarrow Y'' & f''':X'''\rightarrow Y'''
\end{array}
$$

\begin{remark}
For tensor diagrams starting and ending over regions labeled by the same object $\mathfrak{A}$ (resp. $\mathfrak{B}$), one can suppress the $\mathfrak{B}$ (resp. $\mathfrak{A}$) colored region to realize it as a tensor diagram in $\mathcal{C}$ (resp. $\mathcal{D}$).
\end{remark}

\begin{definition}
A Morita context is said to be indecomposable if $\operatorname{HOM}(\mathfrak{A}, \mathfrak{B})$ is generated by one simple $1$-morphism.
\end{definition}

\begin{remark}
Notice that $\operatorname{HOM}(\mathfrak{A}, \mathfrak{B})$ is naturally equipped with a left $\mathcal{C}(=\text{END}(\mathfrak{A}))$-module structure. 
On the other hand, all semisimple module category $\mathcal{M}$ of $\mathcal{C}$ can be realized this way, by first realizing $\mathcal{M}$ as the right module category of a Frobenius algebra $Q$(\cite{Ost03m}), then apply M\"{u}ger's universal construction \cite{Mug03a}.

Suppose $\mathcal{E}=\{\mathfrak{A}, \mathfrak{B}\}$ is indecomposable and and $J$ is a generator in $\operatorname{HOM}(\mathfrak{A}, \mathfrak{B})$, we have the following table comparing the notation appeared in \cite{Mug03a} and \cite{Ost03m}.

$$
\begin{tabular}{|c|c|}
\hline
\rm{Frobenius Algebra} $Q=\overline{J}J$ & \rm{Internal Hom} \underline{$\operatorname{End}$}($J$) \\
\hline
$\operatorname{HOM}(\mathfrak{A}, \mathfrak{B})$ & \rm{Right} \underline{$\operatorname{End}$}$(J)$-\rm{module}\\
\hline
$\mathcal{D}=\operatorname{END}(\mathfrak{B})$& $\mathcal{C}_{\mathcal{M}}^{\star}=\operatorname{Fun}_{\mathcal{C}}(\mathcal{M}, \mathcal{M})$\\
\hline
\end{tabular}
$$
\end{remark}

In the graphical calculus of a spherical fusion category $\mathcal{C}$, a special color, which is often called $\Omega$-color or Kirby color, plays an important a role in the quantum topology, which is defined to be the pseudo-object $\displaystyle \sum_{X\in \Irr_{\mathcal{C}}}d(X)X$. 
A key property of the $\Omega$-color is the so called handle-slide property.

In the case of Morita context, we define the $\Omega$-color for each of the four $1$-categories, namely, $\operatorname{HOM}(\mathfrak{A}, \mathfrak{B})$, $\operatorname{HOM}(\mathfrak{B}, \mathfrak{A})$, $\operatorname{HOM}(\mathfrak{A}, \mathfrak{A})$, $\operatorname{HOM}(\mathfrak{B}, \mathfrak{B})$, to be the sum of all irreducible objects, weighted by their spherical dimensions. 
The following lemma reveals the relation between these $\Omega$-colors.
The proof is similar to the one of Proposition 5.7 of \cite{Mug03a}.

\begin{lemma}\label{lem:omegacolor}
Let $\mathcal{E}=\{\mathfrak{A}, \mathfrak{B}\}$ be a Morita context, and ${}_{\mathfrak{A}}\Omega_{\mathfrak{A}}$, ${}_{\mathfrak{A}}\Omega_{\mathfrak{B}}$, ${}_{\mathfrak{B}}\Omega_{\mathfrak{A}}$, ${}_{\mathfrak{B}}\Omega_{\mathfrak{B}}$ be the $\Omega$-colors in corresponding categories. 
For all irreducible object $J$ in $\operatorname{HOM}(\mathfrak{A}, \mathfrak{B})$, we have that ${}_{\mathfrak{A}}\Omega_{\mathfrak{A}}J
=d_{J}\ {}_{\mathfrak{A}}\Omega_{\mathfrak{B}}
=J\ {}_{\mathfrak{B}}\Omega_{\mathfrak{B}}$ and 
${}_{\mathfrak{B}}\Omega_{ \mathfrak{B}}\overline{J}
=d_{J}\ {}_{\mathfrak{B}}\Omega_{\mathfrak{A}}
=\overline{J}\ {}_{\mathfrak{A}}\Omega_{\mathfrak{A}}$.
\end{lemma}
\begin{proof}
It suffices to show that ${}_{\mathfrak{A}}\Omega_{\mathfrak{A}}J
=d_{J}\ {}_{\mathfrak{A}}\Omega_{\mathfrak{B}}$.
The proofs of the rest are identical. 
Let $\{X_i|i\in I\}$ be the set of simple objects in $\operatorname{HOM}(\mathfrak{A}, \mathfrak{A})$ and $\{J_k|k\in K\}$ be the set of simple objects in $\operatorname{HOM}(\mathfrak{A}, \mathfrak{B})$.

We denote the multiplicity of $J_k$ in $X_iJ$ by $N_{i}^{k}$.
By adjointness of the $1$-morphisms, we have
$$\hom(X_iJ, J_{k})=\hom(X_i, J_{k}\overline{J}).$$
Hence,
$$\sum_{i\in I}d(X_i)X_iJ
=\sum_{i\in I, k\in K}d(X_i)N_{i}^{k}J_{k}
=\sum_{k\in K}d(J_k\overline{J})J_k
=d(J)\sum_{k\in K}d(J_k)J_k,$$
i.e. ${}_{\mathfrak{A}}\Omega_{\mathfrak{A}}J
=d_{J}\ {}_{\mathfrak{A}}\Omega_{\mathfrak{B}}$.
\end{proof}

We have the following corollary for the graphical calculus.
\begin{corollary}\label{cor:omega}
The $\Omega$-color in $\operatorname{HOM}(\mathfrak{A}, \mathfrak{B})$ and $\operatorname{HOM}(\mathfrak{B}, \mathfrak{A})$ has the handle slide property on both side.
\end{corollary}

\begin{proof}
Let $J\in \Hom(\mathfrak{A}, \mathfrak{B})$ be an simple $1$-morphism with $d(J)\neq 0$. 
Then we have
$$
\vcenter{\hbox{\begin{tikzpicture}
\draw[rounded corners, red, -<-=0.15](-0.6, -1) rectangle (0.6, 1);
\draw (0, 0) node{$\mathfrak{B}$};
\draw (1, 0) node{$\mathfrak{A}$};
\draw (-0.9, 0) node{\tiny{${}_{\mathfrak{A}}\Omega_{\mathfrak{B}}$}};
\end{tikzpicture}}}
=
\frac{1}{d(J)}\vcenter{\hbox{\begin{tikzpicture}
\draw[rounded corners, blue, -<-=0.15](-0.8, -1.2) rectangle (0.8, 1.2);
\draw[rounded corners, red](-0.5, -0.8) rectangle (0.5, 0.8);
\draw (-1, 0) node{\tiny$J$};
\draw (0, 0) node{$\mathfrak{B}$};
\draw (1, 0) node{$\mathfrak{A}$};
\end{tikzpicture}}}=
\frac{1}{d(J)}
\vcenter{\hbox{\begin{tikzpicture}
\draw[rounded corners, red](-0.8, -1.2) rectangle (0.8, 1.2);
\draw[rounded corners, blue, -<-=0.15](-0.5, -0.8) rectangle (0.5, 0.8);
\draw (-0.3, 0) node{\tiny$J$};
\draw (0, 0) node{$\mathfrak{B}$};
\draw (1, 0) node{$\mathfrak{A}$};
\end{tikzpicture}}}.$$
Hence, the corollary follows from the handle slide property of $\Omega$-colors of the two spherical fusion categories.
\end{proof}

\subsection{Topological Quantum Field Theory}
A cobordism category $\text{Cob}$ is a category consisting of the following:
\begin{enumerate}
    \item objects are disjoint union of oriented closed surfaces.
    \item morphisms are oriented compact 3-manifolds whose boundaries are oriented closed surface.
    The composition of morphisms is the gluing of the boundaries.
\end{enumerate}

A (2+1)-dimensional topological quantum field theory (TQFT) is a strict symmetric monoidal functor $\tau$ from a suitable cobordism category $\text{Cob}$ to the category $\Vec$ of finite dimensional vector spaces. 

\begin{lemma}[Trace Formula]\label{lem:traceformula}
Let $\tau: \text{Cob}\rightarrow  \Vec$ be a $2+1$ TQFT. 
Then
$$\tau(F\times S^{1})=\dim \tau(F)$$
for all object $F$ in $\text{Cob}$.
\end{lemma}
\begin{proof}
We refer to \cite[Section 1.2]{BHMV2} or \cite[Chapter III, Theorem 2.1.3] {Tur94}.
\end{proof}


\subsection{Mapping Class Groups and Quantum Representations}
In this section, we recall the definition of mapping class groups of compact surfaces, and the quantum representation from a topological quantum field theory. 
For details regarding the mapping class groups, we refer the readers to \cite{FarMar12}.

\begin{definition}
Let $F$ be a compact oriented surface. Denote by $\mathrm{Homeo^{+}}(F, \partial F)$ the group of orientation-preserving homeomorphisms of $F$ that restrict to identity on the boundary $\partial F$. The mapping class group of $F$, denoted by $\MCG(F)$, is the group 
$$\MCG(F)=\pi_{0}(\mathrm{Homeo^{+}(F, \partial F)})\,.$$
\end{definition}

\begin{remark}
Alternatively, $\MCG(F)$ is the group of the orientation preserving and boundary fixing homeomorphisms of $F$ that restrict to identity on $\partial F$ modulo ones that are isotopic to the identity map, see for example \cite{FarMar12}.
\end{remark}

\begin{example}
The mapping class group of a torus $\mathbb{T}$ is isomorphic to $\SL_2(\mathbb{Z})$, with the following identification.
\begin{itemize}
\item The matrix 
$\begin{pmatrix}
0 & -1\\ 1 & 0
\end{pmatrix}$ 
is identified with the rotation of the universal cover $\mathbb{R}^{2}$ of $\mathbb{T}$ around the origin by $\displaystyle \frac{\pi}{2}$.
\item The matrix
$\begin{pmatrix}
1 & 1\\ 0 & 1
\end{pmatrix}$ is identified with a Dehn twist along the meridian of the torus.
\end{itemize}
\end{example}

There are generalizations of mapping class groups of compact surfaces with marked points. 
Let $F$ be an oriented, compact surface with marked points in the interior of $F$. 
Then $\MCG(F)$ is the group of homeomorphisms of $F$ that leave the set of marked points invariant, modulo isotopy.

\begin{example}
Let $F$ be a closed disk with $N$ marked points in the interior, then $\MCG(F)$ is isomorphic to the braid group $B_{N}$ with $N$ strands.
\end{example}

For all $\mathfrak{f}\in \mathrm{Homeo}^{+}(F)$, we denote $C_{\mathfrak{f}}$ to be the mapping cylinder associated to $\mathfrak{f}$ \cite{Hat02}. 
The following theorem shows this assignment descents to a mapping from the mapping class group $\MCG(F)$ to the morphism space $\mathrm{Mor}_{\mathrm{Cob}}(F, F)$.
\begin{theorem}[{\cite[Theorems 1.6 and 1.9]{Mil65}}]
Let $\mathfrak{f}, \mathfrak{g}\in \mathrm{Homeo}^{+}(F)$. 
If $[\mathfrak{f}]=[\mathfrak{g}]\in \MCG(F)$ then $[C_{\mathfrak{f}}] = [C_{\mathfrak{g}}]$ as cobordism classes. 
Moreover,
\begin{equation*}
C_{\mathfrak{f}} \circ C_{\mathfrak{g}} := C_{\mathfrak{g}} \sqcup_{F} C_{\mathfrak{f}} = C_{\mathfrak{f}\circ \mathfrak{g}}\,. 
\eqno\qed
\end{equation*}
\end{theorem}
Given a TQFT $\tau: \operatorname{Cob} \to \Vec$, the assignment $[\mathfrak{f}]\mapsto \tau(C_{\mathfrak{f}})\in \mathrm{End}(\tau(F))$ defines a representation of $\MCG(F)$ on $\tau(F)$. The mapping class group representations constructed this way is called the quantum representations associated to the TQFT $\tau$.

\section{Alterfolds and Tube Category}
In this section, we review the basic notions of 3-alterfolds and the properties of the partition function on them. To prepare for the following sections, we also recall the main properties of the tube category $\cA$ which is braided equivalent to the Drinfeld center of the input spherical fusion category.

\begin{definition}
By a 3-alterfold, we mean a pair $(M, \Sigma)$ where $M$ is an oriented, compact $3$-manifold without boundary and $\Sigma$ is an embedded oriented compact surface without boundary in $M$ separating $M\setminus \Sigma$ into connected components that are alternatively colored by colors $A$ and $B$. We denote the $B$-colored (resp.\ $A$-colored) region by
$R_{B}(M, \Sigma)$ (resp.\ $R_{A}(M, \Sigma)$). 
We further require the orientation of $\Sigma$ to be consistent with the boundary orientation of $R_{A}(M, \Sigma)$, which is the opposite of the boundary orientation of $R_{B}(M, \Sigma)$.
\end{definition}

\begin{definition}
Let $\mathcal{C}$ be a spherical fusion category. By a $\mathcal{C}$-decorated 3-alterfold, we mean a triple $(M, \Sigma, \Gamma)$ where $(M, \Sigma)$ is a 3-alterfold and $\Gamma$ is a $\mathcal{C}$-diagram embedded in the separating surface $\Sigma$.
\end{definition}

For $\mathcal{C}$-decorated $3$-alterfolds, the authors have defined a set of moves changing the decoration and the $A$-$B$ coloring of the underlying $3$-manifold as follows.

\begin{itemize}
\item \textbf{Planar Graphical Calculus:} Let $D\subset \Sigma$ be an embedded disk and $\Gamma_{D}=D\cap \Gamma$. Suppose $\Gamma_{D}'=\Gamma_{D}$ as morphisms in $\mathcal{C}$. We change the $\Gamma$ to $\Gamma'$ by replacing $\Gamma_{D}$ to $\Gamma_{D}'$.

\begin{align}\label{eq:calculus}
\begin{array}{ccc}
   \vcenter{\hbox{
\begin{tikzpicture}
\draw[->-=0.8, ->-=0.2, blue] (-.3, 1.2)--(-.3, -1.2);
\draw[->-=0.8, ->-=0.2, blue] (.3, 1.2)--(.3, -1.2);
\draw[blue] (0, -.75) node{$\ldots$};
\draw[blue] (0, .65) node{$\ldots$};
\draw[fill=white] (-.4, -.3) rectangle (.4, .3);
\draw (0, 0) node{$\Gamma_{D}$};
\draw[dashed] (-1.2, -1.2) rectangle (1.2, 1.2);
\draw (1, -1) node{$D$};
\end{tikzpicture}}}
&
\rightarrow
&
\vcenter{\hbox{
\begin{tikzpicture}    
\draw[->-=0.8, ->-=0.2, blue] (-.3, 1.2)--(-.3, -1.2);
\draw[->-=0.8, ->-=0.2, blue] (.3, 1.2)--(.3, -1.2);
\draw[blue] (0, -.75) node{$\ldots$};
\draw[blue] (0, .65) node{$\ldots$};
\draw[fill=white] (-.4, -.3) rectangle (.4, .3);
\draw (0, 0) node{$\Gamma_{D}'$};
\draw[dashed] (-1.2, -1.2) rectangle (1.2, 1.2);
\draw (1, -1) node{$D$};
\end{tikzpicture}}}\\
\end{array}
\end{align}

\item \textbf{Move 0:} Let $P$ be a point in the interior of $R_B$. We change the color of a tubular neighborhood $P_{\epsilon}$ of $P$ from $B$ to $A$ and add a factor of $\displaystyle \frac{1}{\mu}$ in front of it.
\begin{align}\label{eq:move0}
\begin{array}{ccc}
   \vcenter{\hbox{
\begin{tikzpicture}
\draw[dashed] (0, 0) rectangle (2, 2);
\draw[dashed] (.8, .8) rectangle (2.8, 2.8);
\draw[dashed] (0, 2)--+(0.8, 0.8);
\draw[dashed] (2, 2)--+(0.8, 0.8);
\draw[dashed] (2, 0)--+(0.8, 0.8);
\draw[dashed] (0, 0)--+(0.8, 0.8);
\draw (2, 2.3) node{\tiny{$R_{B}$}};
\draw (1.3, 1.3) node[right]{\tiny{$P$}} node{$\cdot$};
\end{tikzpicture}}}  & \rightarrow & \frac{1}{\mu}\vcenter{\hbox{
\begin{tikzpicture}
\draw[dashed] (0, 0) rectangle (2, 2);
\draw[dashed] (.8, .8) rectangle (2.8, 2.8);
\draw[dashed] (0, 2)--+(0.8, 0.8);
\draw[dashed] (2, 2)--+(0.8, 0.8);
\draw[dashed] (2, 0)--+(0.8, 0.8);
\draw[dashed] (0, 0)--+(0.8, 0.8);
\filldraw[white!70!gray] (1.3, 1.3) circle (0.3);
\draw[dashed, opacity=0.3] (1.3, 1.3) [partial ellipse=0:180:0.3 and 0.1];
\draw[opacity=0.3] (1.3, 1.3) [partial ellipse=180:360:0.3 and 0.1];
\draw (2, 2.3) node{\tiny{$R_{B}$}};
\end{tikzpicture}}}\\
\vspace{2mm}\\
(M, \Sigma, \Gamma)&\rightarrow & (M, \Sigma \sqcup \partial P_{\epsilon}, \Gamma)\\
\end{array}
\end{align}

\item \textbf{Move 1:} Let $S$ be an embedded arc in $R_{B}$ with $\partial S$ meeting $\Sigma$ transversely and not intersecting $\Gamma$. We change the color of the tubular neighborhood $S_{\epsilon}$ of $S$ from $B$ to $A$, and put an $\Omega$-colored circle $C$ on the belt of $S_{\epsilon}$.

\begin{align}\label{eq:move1}
\begin{array}{ccc}
   \vcenter{\hbox{
\begin{tikzpicture}
\path[fill=gray!50!white]
(-1.5, 0) [partial ellipse=-90:90:0.5 and 1];
\draw (-1.5, 0) [partial ellipse=-90:90:0.5 and 1];
\path[fill=gray!50!white]
(-2.5, -1) rectangle (-1.5, 1);
\begin{scope}[xscale=-1]
\path[fill=gray!50!white]
(-1.5, 0) [partial ellipse=-90:90:0.5 and 1];
\draw (-1.5, 0) [partial ellipse=-90:90:0.5 and 1];
\path[fill=gray!50!white]
(-2.5, -1) rectangle (-1.5, 1);
\end{scope}
\draw[dashed] (-1, 0)--(1, 0);
\draw (0, 0.2) node[above]{\tiny{$S$}};
\end{tikzpicture}}}  & \rightarrow & \vcenter{\hbox{
\begin{tikzpicture}
\path[fill=gray!50!white]
(-1.5, 0) [partial ellipse=-90:90:0.5 and 1];
\draw (-1.5, 0) [partial ellipse=-90:90:0.5 and 1];
\path[fill=gray!50!white]
(-2.5, -1) rectangle (-1.5, 1);
\begin{scope}[xscale=-1]
\path[fill=gray!50!white]
(-1.5, 0) [partial ellipse=-90:90:0.5 and 1];
\draw (-1.5, 0) [partial ellipse=-90:90:0.5 and 1];
\path[fill=gray!50!white]
(-2.5, -1) rectangle (-1.5, 1);
\end{scope}
\draw[dashed] (-1, 0)--(1, 0);
\draw (0, 0.2) node[above]{\tiny{$C$}};
\path[fill=white]
(-1.1, -0.25) rectangle (1.1, 0.25);
\path[fill=gray!50!white]
(-1.1, -0.25) rectangle (1.1, 0.25);
\draw (-1.02, 0.25)--(1.02, 0.25);
\draw (-1.02, -0.25)--(1.02, -0.25);
\draw[red, dashed] (0, 0) [partial ellipse=-90:90:0.125 and 0.25];
\draw[red] (0, 0) [partial ellipse=90:270:0.125 and 0.25];
\end{tikzpicture}}}\\
\vspace{2mm}\\
(M, \Sigma, \Gamma)&\rightarrow & (M, \partial(R_{B}\setminus S_{\epsilon}), \Gamma\sqcup C)\\
\end{array}
\end{align}

\item \textbf{Move 2:} Let $D$ be a disk in $R_{B}$ with $\partial D\subset \partial R_{B}$ intersect $\Gamma$ only at edges transversely. We change the color of a tubular neighborhood of $D$ from $B$ to $A$, then cut the diagram along $\partial D$ and put a pair of sum of dual basis $\phi$ and $\phi'$ on both side of $D_{\epsilon}$.

\begin{align}\label{eq:move2}
\begin{array}{ccc}
   \vcenter{\hbox{
\begin{tikzpicture}
\draw[dashed] (0, 1.5) [partial ellipse=0:360:1 and 0.3];
\draw[dashed] (0, 0) [partial ellipse=0:360:1 and 0.3];
\draw[dashed] (0, -1.5) [partial ellipse=0:360:1 and 0.3];
\draw (0.5, 0) node{\tiny{$D$}};
\draw (-1, -1.5)--(-1, 1.5);
\draw (1, -1.5)--(1, 1.5);
\path[fill=gray!50!white]
(-1, -1.5) rectangle (-2, 1.5);
\path[fill=gray!50!white]
(1, -1.5) rectangle (2, 1.5);
\draw [blue, ->-=0.7] (0, 1.2)--(0, -1.8);
\draw [blue, ->-=0.7] (-0.3, 1.2)--(-0.3, -1.8);
\draw [blue, ->-=0.7] (0.3, 1.2)--(0.3, -1.8);
\end{tikzpicture}}}  & \rightarrow & \vcenter{\hbox{
\begin{tikzpicture}
\draw[dashed] (0, 1.5) [partial ellipse=0:360:1 and 0.3];
\draw[dashed] (0, -1.5) [partial ellipse=0:360:1 and 0.3];
\draw (-1, -1.5)--(-1, -1.2);
\draw (1, 1.5)--(1, 1.2);
\draw (-1, 1.5)--(-1, 1.2);
\draw (1, -1.5)--(1, -1.2);
\path[fill=gray!50!white]
(-1, -1.5) rectangle (-2, 1.5);
\path[fill=gray!50!white]
(1, -1.5) rectangle (2, 1.5);
\draw [blue, ->-=0.5] (0, 1.2)--(0, 0.7);
\draw [blue, ->-=0.5] (-0.3, 1.2)--(-0.3, 0.7);
\draw [blue, ->-=0.5] (0.3, 1.2)--(0.3, 0.7);
\draw [blue, ->-=0.5] (0, -1)--(0, -1.8);
\draw [blue, ->-=0.5] (-0.3, -1)--(-0.3, -1.8);
\draw [blue, ->-=0.5] (0.3, -1)--(0.3, -1.8);
\path[fill=gray!50!white]
(-1, -1.2) arc (180:0:1)--(1, 1.2) arc (0:-180:1)--(-1, -1.2);
\draw (-1, -1.2) arc (180:0:1);
\draw (1, 1.2) arc (0:-180:1);
\draw [fill=white](-0.4, 0.5) rectangle (0.4, 0.8); \node at (0, 0.65) {\tiny{$\phi$}};
\draw [fill=white](-0.4, -1.2) rectangle (0.4, -0.9); \node at (0, -1.05) {\tiny{$\phi'$}};
\end{tikzpicture}}}\\
\vspace{2mm}\\
(M, \Sigma, \Gamma)&\rightarrow & (M, \partial(R_{B}\setminus D_{\epsilon}), \Gamma')\\
\end{array}
\end{align}
Recall that $\phi$ and $\phi'$ denote the dual base $\{\phi_j\}_j$ and $\{\phi_j'\}_j$ by and suppress the summation.

\item \textbf{Move 3:} Let $T$ be a $B$-colored $3$-ball with $\partial T\subset \Sigma$ and $\partial T\cap \Gamma=\emptyset$. We change the color of $T$ from $B$ to $A$.
\begin{align}\label{eq:move3}
\begin{array}{ccc}
\vcenter{\hbox{
\begin{tikzpicture}
\path[fill=gray!50!white]
(-2, -1.5) rectangle (2, 1.5);
\draw [fill=white] (0, 0) [partial ellipse=0:360:1];
\draw[opacity=0.3] (0, 0)[partial ellipse=0:-180:1 and 0.3];
\draw[dashed, opacity=0.3] (0, 0)[partial ellipse=0:180:1 and 0.3];
\draw (0.3, 0.3) node{\tiny{$T$}};
\end{tikzpicture}
}}
&\rightarrow&
\vcenter{\hbox{
\begin{tikzpicture}
\path[fill=gray!50!white]
(-2, -1.5) rectangle (2, 1.5);
\end{tikzpicture}
}}\\
\vspace{2mm}\\
(M, \Sigma, \Gamma)&\rightarrow& (M, \Sigma\setminus \partial T, \Gamma)\\
\end{array}
\end{align}
\end{itemize}

In \cite{LMWW23}, we proved the existence of a partition function on decorated 3-alterfolds that is invariant under the above moves. The precise statement of the result is given below.

\begin{theorem}\cite[Theorem 3.9]{LMWW23}\label{thm:partition function}
Let $\mathcal{C}$ be a spherical fusion category over $\k$.  There exists a unique multiplicative partition function $Z$ mapping $\mathcal{C}$-decorated 3-alterfolds $(M, \Sigma, \Gamma)$ to the ground field $\k$ that is invariant under planar graphical calculus, Move $0$-$3$, and evaluate to $1$ for all $A$-colored alterfolds. In particular, the partition function is a topological invariant of pair $(R_{B}(M), \Gamma)$. \qed
\end{theorem}

We remark that since the partition function only depend on the $B$-colored part, so changing the interior of $A$-colored region does not change the value of the partition function.

In addition, we also introduced in \cite{LMWW23} the tube category $\cA$ as a topological interpretation of the Drinfeld center in the alterfold theory, which is essential in proving the equivalence of Turaev-Viro invariant of a spherical fusion category and Reshetikhin-Turaev invariant of its Drinfeld center. In the tube category $\cA$,
\begin{itemize}
    \item objects are horizontal $A$-colored disks and
    \item morphisms are $A$-colored handlebodies with $\mathcal{C}$-decoration on boundary region of the complement of the horizontal boundary.
\end{itemize}
We showed in \cite[Theorem 4.33]{LMWW23} that $\mathcal{A}$ is equivalent to the Drinfeld center $\mathcal{Z(C)}$ of $\mathcal{C}$ as a modular tensor category.
We will describe objects and morphisms in $\mathcal{Z(C)}$ in terms of $A$-colored $\mathcal{C}$-decorated cornered handlebodies, which will be freely used for the rest of the current paper.

In the end of this section, we recall that the induction functor $I:\mathcal{C}\rightarrow \mathcal{A}$, defined by $I(Y)=O_{Y}$, denote a single disk with label $Y$ on its boundary, and on the morphism level, we have
\begin{align*}
I\left(
\vcenter{\hbox{\scalebox{0.6}{
\begin{tikzpicture}[scale=0.35]
\draw[blue, ->-=0.5] (0, 4.2) node[above, black]{\tiny{$X$}}->(0, 0.5);
\draw[blue, ->-=0.5] (0, -0.5)--(0, -4.8) node[below, black]{\tiny{$Y$}};
\node [draw, fill=white] (0, 0){\tiny $f$};
\end{tikzpicture}}}}\right)
=
\vcenter{\hbox{\scalebox{0.6}{
\begin{tikzpicture}[scale=0.35]
\draw (0,5) [partial ellipse=0:360:2 and 0.8];
\draw (-2, 5)--(-2, -4);
\draw (2, 5)--(2, -4);
\draw[dashed] (0,-4) [partial ellipse=0:180:2 and 0.8];
\draw (0,-4) [partial ellipse=180:360:2 and 0.8];
\draw[blue, ->-=0.5] (0, 4.2) node[above, black]{\tiny{${X}$}}->(0, 0.5);
\draw[blue, ->-=0.5] (0, -0.5)--(0, -4.8)node[below, black]{\tiny{${Y}$}};
\node [draw, fill=white] (0, 0){\tiny $f$};
\end{tikzpicture}}}}.
\end{align*}

\section{Alterfold Topological Quantum Field Theory}

In this section, we prove the partition function $Z$ extends to a topological quantum field theory $\V$. We first introduce the $3$-dimensional alterfold cobordism category where the functor $\V$ is defined.

\subsection{The Category $\Cob$ and the Functor $\V$}

\begin{definition}[2-Alterfold]
By a 2-alterfold, we mean a pair $(F, \gamma)$ where $F$ is an oriented, compact surface without boundary and $\gamma$ is an embedded oriented closed multicurve in $F$ separating $F\setminus \gamma$ into connected components that are alternatively colored by $A$ and $B$. 
We denote the $B$-colored (resp.~$A$-colored) region by
$R_{B}(F, \gamma)$(resp. $R_{A}(F, \gamma)$). 
We further require the orientation of $\gamma$ to be consistent with the boundary orientation of $R_{A}(F, \gamma)$ and opposite to the boundary orientation of $R_{B}(F, \gamma)$.
\end{definition}

\begin{remark}
A 2-alterfold always looks like the following locally, taking the standard orientation of the $xy$-plane.
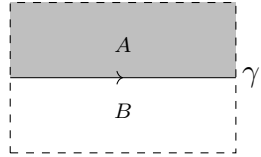
\[
\begin{tikzpicture}
\path[fill=gray!50!white] (-1.5, 0) rectangle (1.5, 1);
\path[fill=white] (-1.5, 0) rectangle (1.5, -1);
\draw[->-=0.5] (-1.5, 0)--(1.5, 0);
\draw (1.7, 0) node {$\gamma$};
\draw[dashed](-1.5, -1) rectangle (1.5, 1);
\draw (0, 0.2) node[above]{\tiny $A$};
\draw (0, -0.2) node[below]{\tiny $B$};
\end{tikzpicture}
\]
     As a counter example, a torus with a longitude and colored by $A$ or $B$ is not a 2-alterfold.
In fact, the multicurve $\gamma$ can be obtained by intersecting the separating surface $\Sigma$ in a $3$-alterfold $(M, \Sigma)$ and a generic embedded surface $F$ in $M$ transversely.
\end{remark}

\begin{definition}[Decorated 2-Alterfold]
Let $\mathcal{C}$ be a spherical fusion category. 
By a $\cC$-decorated 2-alterfold, we mean a triple $(F, \gamma, P)$ where $(F, \gamma)$ is a $2$-alterfold and $P$ is a finite set of (oriented) points on $\gamma$ colored by objects in $\mathcal{C}$.
 


\end{definition}

\begin{remark}\label{rem:orientationreversing}
Suppose $(F, \gamma, P)$ is a $\cC$-decorated 2-alterfold.
If the orientation of $F$ is reversed, the orientation of  the marked points are changed to their duals.
\end{remark}

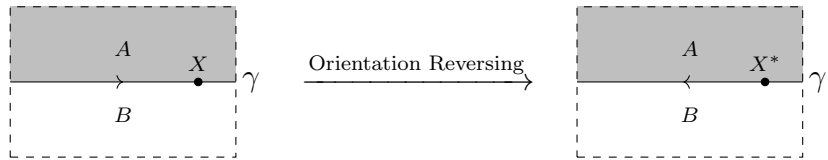
\[
\vcenter{\hbox{
\begin{tikzpicture}
\path[fill=gray!50!white] (-1.5, 0) rectangle (1.5, 1);
\path[fill=white] (-1.5, 0) rectangle (1.5, -1);
\draw[->-=0.5] (-1.5, 0)--(1.5, 0);
\draw (1.7, 0) node {$\gamma$};
\draw[dashed](-1.5, -1) rectangle (1.5, 1);
\draw (0, 0.2) node[above]{\tiny $A$};
\draw (0, -0.2) node[below]{\tiny $B$};
\draw[fill=black] (1, 0) node[above]{\tiny$X$}circle(0.05);
\end{tikzpicture}}}
\hspace{3mm}
\xrightarrow[]{\text{Orientation Reversing}}
\hspace{3mm}
\vcenter{\hbox{
\begin{tikzpicture}
\path[fill=gray!50!white] (-1.5, 0) rectangle (1.5, 1);
\path[fill=white] (-1.5, 0) rectangle (1.5, -1);
\draw[-<-=0.5] (-1.5, 0)--(1.5, 0);
\draw (1.7, 0) node {$\gamma$};
\draw[dashed](-1.5, -1) rectangle (1.5, 1);
\draw (0, 0.2) node[above]{\tiny $A$};
\draw (0, -0.2) node[below]{\tiny $B$};
\draw[fill=black] (1, 0) node[above]{\tiny$X^{*}$}circle(0.05);
\end{tikzpicture}}}
\]

\begin{definition}[$3$-Alterfold with Time Boundary]\label{def:sp-bdry}
By a $3$-alterfold with time boundary, we mean a pair $(M, \Sigma)$ where $M$ is an oriented, compact 3-manifold with boundary $\partial M$ and $\Sigma$ is an embedded oriented surface in $M$ meeting $\partial M$ transversely such that $\partial \Sigma\subset \partial M$. 
The surface $\Sigma$ separates $M\setminus \Sigma$ into connected components that are alternatively colored by $A$ and $B$ coherent with the orientation of $\Sigma$. 
We call the pair $(\partial M, \partial \Sigma)$ the time boundary of $(M, \Sigma)$.
We call the separating surface $\Sigma$ the space boundary of $(M, \Sigma)$.
\end{definition}

\begin{remark}
Note that the time boundary $(\partial M, \partial \Sigma )$ of $(M, \Sigma)$ is naturally a $2$-alterfold, equipped with the boundary orientation. 
The consistency of orientations in Definition \ref{def:sp-bdry} means the orientation of $\Sigma$ equals to the boundary orientation of $R_{A}(M, \Sigma)$ and the opposite of the boundary orientation of $R_{B}(M, \Sigma)$. 
\end{remark}

\begin{definition}[Decorated 3-Alterfold with Time Boundary]
Let $\mathcal{C}$ be a spherical fusion category. 
By a $\mathcal{C}$-decorated $3$-alterfold with time boundary, we mean a triple $(M, \Sigma, \Gamma)$ where $(M, \Sigma)$ is a $3$-alterfold with time boundary and $\Gamma$ is a $\mathcal{C}$-diagram embedded in $\Sigma$, with strands ending either at the vertices of $\Gamma$ or at $\partial \Sigma$. 
We further require $\Gamma$ meets $\partial M$ transversely.
\end{definition}

\begin{remark}
Note that the boundary of a $\mathcal{C}$-decorated $3$-alterfold $(M, \Sigma, \Gamma)$ is naturally a $\cC$-decorated $2$-alterfold $(\partial M, \Sigma\cap \partial M, \Gamma\cap \partial M)$ equipped with boundary orientation.
\end{remark}


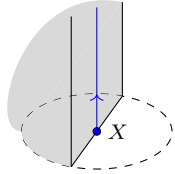
\begin{align*}
\vcenter{\hbox{
\begin{tikzpicture}
    \path [fill=white!70!gray](0.34,1.72)..controls (-1, 2) and (-1.5, 0)..(-1, 0);
    \path [fill=white!70!gray](0.34, 0.47)--(0.34,1.72)--(-1, 0)--(0, 0);
    \draw [dashed] (0, 0) [partial ellipse=0:360:1 and 0.5];
    \path [fill=white!70!gray] (0, 0) [partial ellipse=-110:-290:1 and 0.5];
    \draw (-0.34, -0.47)--(0.34, 0.47);
    \draw (-0.34, -0.47)--+(0, 2);
    \draw (0.34, 0.47)--+(0, 1.25);
    \draw[blue, ->-=0.3] (0, 0)--+(0, 1.65);
    \draw[fill=blue] (0, 0) node[right]{\tiny$X$} circle(0.05);
\end{tikzpicture}}}
\end{align*}

\begin{definition}[Alterfold Cobordism Category]
The alterfold cobordism category $\mathrm{ACob}_2$ is a symmetric strict monoidal category consisting of the following data:
\begin{enumerate}
\item Objects are disjoint unions of  2-alterfolds.
\item Morphisms between 2-alterfolds $(F, \gamma)$ and $(F', \gamma')$ are 3-alterfolds $(M, \Sigma)$ with time boundary equipped with an orientation preserving homeomorphism $\partial_{-}(M, \Sigma):=(\partial M_{-}, \Sigma\cap \partial M_{-})\rightarrow (F, \gamma)$ and an orientation reversing homeomorphism $\partial_{+}(M, \Sigma):=(\partial M_{+}, \Sigma\cap \partial M_{+})\rightarrow (F', \gamma')$, up to equivalence, where $\partial M=\partial M_- \cup \partial M_+$.
\end{enumerate}
The composition of morphisms 
$(M, \Sigma):(F, \gamma)\rightarrow (F', \gamma')$ and $(M', \Sigma'): (F', \gamma')\rightarrow (F'', \gamma'')$ is given by the gluing $(M, \Sigma)\sqcup_{(F', \gamma')}(M', \Sigma')$, and the identity morphism $\id_{(F, \gamma)}$ is given by the cylinder $(F\times [0, 1], \gamma\times [0, 1])$. The monoidal structure is defined to be the disjoint union of $\cC$-decorated alterfolds.

A $\cC$-decorated alterfold cobordism category $\Cob$ is an alterfold cobordism category such that the 2-alterfolds and 3-alterfolds with time boundary are decorated by a spherical fusion category $\cC$.
\end{definition}



\begin{remark}
The category $\mathrm{ACob}_2^{\mathcal{C}}$ is rigid, and we associate $(F, \gamma, P)$ with a canonical dual object $-(F, \gamma, P):=(-F, -\gamma, -P)$ where $-F$ is the surface homeomorphic to $F$ with opposite orientation. Notice the orientation of $\gamma$ (and $P$) also need to be reversed to make its orientation consistent with change of the orientation of $F$ (see Remark \ref{rem:orientationreversing}).
The evaluation and the coevaluation map can be established by bending the identity cobordism.
\end{remark}

We slightly extend the original definition of TQFT of Atiyah. 
Namely, we call a monoidal functor from any cobordism category (possibly with additional structures) to the category of finite dimensional vector spaces a TQFT.
Now we are able to state the main theorem of this section.
\begin{theorem}\label{thm:tqft}
The partition function $Z$ defined in Theorem \ref{thm:partition function} extends to a topological quantum field theory 
$$\mathbb{V}: \Cob \rightarrow \Vec.$$
\end{theorem}

The proof of the theorem follows from the standard approach of universal reconstruction \cite{BHMV2}. We sketch the idea as follows. 
First notice that for all 2-alterfold $(F, \gamma, P)$, there exists a pairing 
\begin{align*}
[\cdot ,\cdot ]: &\hom(\emptyset, (F, \gamma, P))\times \hom((F, \gamma, P), \emptyset) \rightarrow \k, \\
& [(M, \Sigma, 
\Gamma),(M', \Sigma', \Gamma')] := Z((M, \Sigma, 
\Gamma)\sqcup_{(F, \gamma, P)}(M', \Sigma', \Gamma')).
\end{align*}
Then we define the vector space assigned to $(F, \gamma, P)$ to be
$$\V(F, \gamma, P)=\text{span}_{\k}\{\hom(\emptyset, (F, \gamma, P))\}/\ker([\cdot ,\cdot ])\,,$$
where by $\ker([\cdot ,\cdot ])$, we mean the annihilator of $\hom((F, \gamma, P), \emptyset)$ with respect to the pairing $[\cdot ,\cdot ]$.
In order to prove $Z$ extends to a TQFT, we only need to show for all $\mathcal{C}$-decorated 2-alterfolds $(F, \gamma, P)$, the assigned vector space $\V(F, \gamma, P)$ is finite dimensional. 
We prove the latter statement by constructing a finite set of generators. Then we show arbitrary cobordism can be reduced to a sum of generators by the local moves described in Theorem \ref{thm:partition function}. 

We first show the difference of two cobordisms differed by a local move is in the kernel of $[\cdot ,\cdot ]$.

\begin{lemma}\label{lem:equivmove}
Suppose $(M, \Sigma, \Gamma), (M', \Sigma', \Gamma')\in \hom(\emptyset, (F, \gamma, P))$. 
If $(M', \Sigma', \Gamma')$ is derived from applying one of the moves described in Theorem \ref{thm:partition function}, then $(M, \Sigma, \Gamma)-(M', \Sigma', \Gamma')$ is in the kernel of $[,]$. 
In another word, $\V(M, \Sigma, \Gamma)=\V(M', \Sigma', \Gamma')$ as vectors in $\V(F, \gamma, P)$.
\end{lemma}

\begin{proof}
For all cobordism $(M'', \Sigma'', \Gamma'')$ with $\partial_{-}(M'', \Sigma'', \Gamma'')=(F,\gamma,P)$ and $\partial_{+}(M'', \Sigma'', \Gamma'')=\emptyset$, the $\mathcal{C}$-decorated $3$-alterfolds $(M'', \Sigma'', \Gamma'')\circ(M, \Sigma, \Gamma)$ and $(M'', \Sigma'', \Gamma'')\circ (M', \Sigma', \Gamma')$ are differed by one of the moves described in Theorem \ref{thm:partition function}, thus evaluate to the same scalar through the partition function. 
By Theorem \ref{thm:partition function}, $\V(M, \Sigma, \Gamma)-\V(M', \Sigma', \Gamma')$ is in the kernel of $[\ ,\ ]$.
\end{proof}

We use the above lemma to find a smaller generating set described in the theorem below.

\begin{theorem}\label{thm:collarrepn}
Let $(F, \gamma, P)$ be a $\mathcal{C}$-decorated 2-alterfold.
Then $\V(F, \gamma, P)$ is spanned by vectors $\V(M, \Sigma, \Gamma)\in \V(F, \gamma, P)$ satisfying the following conditions:
\begin{enumerate}[(1)]
    \item $M$ is homemorphic to a multi-handlebody(disjoint union of handlebodies).
    \item Each component of $\gamma$ bounds a disk in $M$.
    \item $\Sigma=F\times \{-\epsilon\}\cup \gamma\times [-\epsilon, 0]$, and the $B$-colored region is homeomorphic to $R_{B}(F)\times [-\epsilon, 0]$.
\end{enumerate}
We call such vector $\V(M, \Sigma, \Gamma)$ as a collar representative of $\V(F, \gamma, P)$.
\end{theorem}
\begin{remark}
The second condition comes from the fact that if $\gamma$ is a set of finite disjoint closed curves on a closed oriented surface $F$, then there exist a handlebody with boundary $F$ subject to the condition $(2)$. Such handlebody can be obtained by gluing disks along the components of $\gamma$ on one side of $F$(call it inner side), gluing more $2$-handles such that inner boundary homeomorphic to disjoint union of spheres, fill in $3$-handles result in a $3$-alterfold satisfies condition $(2)$.
Below is an example of a collar representative.
\begin{align}\label{eq:}
\vcenter{\hbox{
\begin{tikzpicture}
\path [fill=white!70!gray] (0, 0) [partial ellipse=0:180:2 and 0.8];
\path [fill=white!70!gray] (0, 0) [partial ellipse=180:360:2 and 2];
\begin{scope}[scale=0.9]
\path [fill=white!70!gray] (0, 0) [partial ellipse=0:180:1.8 and 1.8];
\end{scope}
\draw (0, 0) [partial ellipse=0:360:2 and 2];
\draw[dashed] (0, 0) [partial ellipse=0:180:2 and 0.8];
\draw (0, 0) [partial ellipse=180:360:2 and 0.8];
\begin{scope}[scale=0.9]
\draw[dashed] (0, 0) [partial ellipse=0:180:1.8 and 0.68];
\draw (0, 0) [partial ellipse=180:360:1.8 and 0.68];
\draw (0, 0) [partial ellipse=0:180:1.8 and 1.8];
\end{scope}
\draw[blue, ->-=0.5] (-.2, -0.8) node[below, black]{\tiny$X_2$}--(0, -0.6);
\draw[blue, ->-=0.5] (0.3, -0.78)node[below, black]{\tiny$X_3$}--(0.5, -0.58);
\draw[blue, ->-=0.5] (-0.7, -0.77)node[below, black]{\tiny$X_1$}--(-0.5, -0.57);
\draw[blue] (0, -0.6)--+(0, 0.5);
\draw[blue] (0.5, -0.58)--+(0, 0.5);
\draw[blue] (-0.5, -0.57)--+(0, 0.5);
\draw[fill=white] (-0.6, -0.1) rectangle (0.6, 0.6); \node at (0, 0.25){$\phi$};
\draw (1.8, 0.5) node[right]{$\gamma\times [-\epsilon, 0]$};
\draw (0, 1.1) node{$R_{B}\times\{-\epsilon\}$};
\end{tikzpicture}}}
\in \hom_{\Cob}\left(\emptyset, 
\vcenter{\hbox{
\begin{tikzpicture}
\path [fill=white!70!gray] (0, 0) [partial ellipse=180:360:2 and 2];
\path [fill=white] (0, 0) [partial ellipse=0:-180:2 and 0.8];
\draw (0, 0) [partial ellipse=0:360:2 and 2];
\draw[dashed] (0, 0) [partial ellipse=0:180:2 and 0.8];
\draw (0, 0) [partial ellipse=180:360:2 and 0.8];
\draw (2, 0) node[right]{$\gamma$};
 \draw[fill=blue] (-.2, -0.8) node[below]{\tiny$X_2$} circle(0.05);
 \draw[fill=blue] (0.3, -0.78)node[below]{\tiny$X_3$} circle(0.05);
 \draw[fill=blue] (-0.7, -0.77)node[below]{\tiny$X_1$} circle(0.05);
\end{tikzpicture}}}
\right)
\end{align}
\end{remark}

\begin{proof}
We only need to show for all cobordism $(M, \Sigma, \Gamma)$ with $\partial (M, \Sigma, \Gamma)=-(F, \gamma, P)$, the vector $\V(M, \Sigma, \Gamma)$ can be written as a linear combination of collar representatives by applying $3$-alterfold moves.

We choose a Morse function $f: M\rightarrow (-\infty, 0]$ such that $f^{-1}(\{0\})=\partial M=F$, and restrict to a Morse function $f|_{\Sigma}$ and $f|_{\Gamma}$ near $\gamma\subset \Sigma$ and $P\subset \Gamma$. Then
there exists an $\epsilon > 0$ such that $(f^{-1}([-\epsilon,0]), f|_{\Sigma}^{-1}([-\epsilon,0]), f|_{\Gamma}^{-1}([-\epsilon,0]))$ is homeomorphic to $(F\times [-\epsilon, 0], \gamma\times [-\epsilon, 0], P\times [-\epsilon, 0])$. 

We first choose a handle decomposition of $R_{B}(f^{-1}(-\infty, -\epsilon])$. One can apply Moves $0$-$3$ to change the color of $R_{B}(f^{-1}(-\infty, -\epsilon])$  to $A$-color with respect to this handle decomposition.
Then we remove the interior of $f^{-1}(-\infty, -\epsilon]$ and glue back an $A$-colored handlebody to make the resulting $3$-alterfold cobordism satisfy the condition $(1)$ and $(2)$.
This operation does not change the image under the functor $\V$ due to the Homeomorphism property in Theorem \ref{thm:partition function}.
By Lemma \ref{lem:equivmove}, the moves we applied to change the underlying alterfold cobordism preserves the image under the functor $\V$, the theorem follows.



\end{proof}


Next, we show modulo the Planar Graphical Calculus, the vector space spanned by collar representatives is finite dimensional.

\begin{proposition}
For all $\cC$-decorated 2-alterfold $(F, \gamma, P)$, the associated vector space $\V(F, \gamma, P)$ is finite dimensional.
\end{proposition}

\begin{proof}
Without loss of generality, we assume $F$ is connected. 
According to Theorem \ref{thm:collarrepn}, we only need to show for all $\mathcal{C}$-decorated 2-alterfold $(F, \gamma, P)$, the space spanned by collar representatives is finite dimensional modulo the Planar Graphical Calculus.

Let $(H, \Sigma)$ be a $3$-alterfold with time boundary $(F, \gamma)$ satisfying the conditions of Theorem \ref{thm:collarrepn}.
We choose a finite set of $g(=\text{genus of } R_{B})$ closed curves $C=\{C_1,\ldots, C_{g}\}$ in $R_{B}(F)$ such that 
\begin{itemize}
\item[(a)] the genus (defined to be the sum of the genus of each connected components) of $R_{B}(F)\setminus C$ equals to $0$ and 
\item [(b)] each component of $C$ bounds a disk in the handlebody $H$.
\end{itemize}
Applying Planar Graphical Calculus, we can further assume that $\Gamma$ meets $C\times \{-\epsilon\}$ transversely at one point, colored by some simple object.
$$
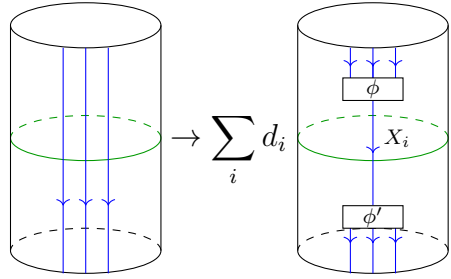
\vcenter{\hbox{
\begin{tikzpicture}
\draw (0, 1.5) [partial ellipse=0:360:1 and 0.3];
\draw[\darkgreen, dashed] (0, 0) [partial ellipse=0:180:1 and 0.3];
\draw[\darkgreen] (0, 0) [partial ellipse=180:360:1 and 0.3];
\draw[dashed] (0, -1.5) [partial ellipse=0:180:1 and 0.3];
\draw[] (0, -1.5) [partial ellipse=180:360:1 and 0.3];
\draw (-1, -1.5)--(-1, 1.5);
\draw (1, -1.5)--(1, 1.5);
(-1, -1.5) rectangle (-2, 1.5);
(1, -1.5) rectangle (2, 1.5);
\draw [blue, ->-=0.7] (0, 1.2)--(0, -1.8);
\draw [blue, ->-=0.7] (-0.3, 1.2)--(-0.3, -1.8);
\draw [blue, ->-=0.7] (0.3, 1.2)--(0.3, -1.8);
\end{tikzpicture}}}
\rightarrow
\sum_{i}d_i
\vcenter{\hbox{
\begin{tikzpicture}
\draw (0, 1.5) [partial ellipse=0:360:1 and 0.3];
\draw[\darkgreen, dashed] (0, 0) [partial ellipse=0:180:1 and 0.3];
\draw[\darkgreen] (0, 0) [partial ellipse=180:360:1 and 0.3];
\draw[dashed] (0, -1.5) [partial ellipse=0:180:1 and 0.3];
\draw[] (0, -1.5) [partial ellipse=180:360:1 and 0.3];
\draw (-1, -1.5)--(-1, 1.5);
\draw (1, -1.5)--(1, 1.5);
(-1, -1.5) rectangle (-2, 1.5);
(1, -1.5) rectangle (2, 1.5);
\draw [blue, ->-=0.5] (0, 1.2)--(0, 0.7);
\draw [blue, ->-=0.5] (-0.3, 1.2)--(-0.3, 0.7);
\draw [blue, ->-=0.5] (0.3, 1.2)--(0.3, 0.7);
\draw [blue, ->-=0.5] (0, -1)--(0, -1.8);
\draw [blue, ->-=0.5] (-0.3, -1)--(-0.3, -1.8);
\draw [blue, ->-=0.5] (0.3, -1)--(0.3, -1.8);
\draw [blue, ->-=0.6] (0, 1)--(0, -1);
\draw (0, 0) node[right]{\tiny{$X_i$}}; 
\draw [fill=white](-0.4, 0.5) rectangle (0.4, 0.8); \node at (0, 0.65) {\tiny{$\phi$}};
\draw [fill=white](-0.4, -1.2) rectangle (0.4, -0.9); \node at (0, -1.05) {\tiny{$\phi'$}};
\end{tikzpicture}}}
$$

Therefore, the vector space $\V(F, \gamma, P)$ can be spanned by $\mathcal{C}$-diagrams $\Gamma$ over $R_{B}(F)\setminus C$ with the following boundary constraints:
\begin{itemize}
    \item At boundary components $\gamma_i\subset \gamma$, $\Gamma\cap \gamma_i=P\cap \gamma_i$.
    \item At boundary components $C_{i-}$ and $C_{i+}$ created by cutting along $C_i$, $\Gamma\cap C_{i-}$ and $\Gamma\cap C_{i+}$ are simple objects which are dual to each other.
\end{itemize}
The space of $\mathcal{C}$-diagrams (modulo planar graphical calculus) on $R_{B}(F)\setminus {C}$ can be identified as a morphism space $\mathcal{A}=\mathcal{Z(C)}$. More precisely, we can assume this morphism space has source $\1_{\mathcal{A}}$ by rigidity of $\mathcal{A}$, where the target space are union of circles with marked points given by the prescribed label $P\cap \gamma_i$ on $\gamma_i$ and a pair of dual simple object on $C_{j+}$ and $C_{j-}$, summing over all admissible choices. Let $I$ be the induced functor from $\mathcal{C}\rightarrow \mathcal{Z(C)}\cong \mathcal{A}$ described in Remark 4.25 in \cite{LMWW23}(which is the adjoint functor of the forgetful functor $F:\mathcal{Z(C)}\rightarrow \mathcal{C}$), this space can be written as: 
$$\hom_{\cA}(\1, I(P_1)\otimes I(P_2)\otimes \cdots \otimes I(P_k)\otimes (\textbf{O}^{\otimes g})),$$
where $P_i$ equals to the tensor product(in anticlockwise order) of markings on $\gamma_i$, and $\textbf{O}=\bigoplus_{i}I(X_i)\otimes I(X_i^{*})$.  By the adjunction between the functors $I$ and $F$, once gluing along the cutting circles $C_i$, one gets a surjective map from the morphism space of $\mathcal{A}$ to $\V(F, \gamma, P)$.
\end{proof}

\begin{proof}[Proof of Theorem \ref{thm:tqft}]
Consider the natural transformation
\begin{align*}
 \iota: \V(F, \gamma, P)\otimes \V(F', \gamma', P')&\rightarrow \V((F, \gamma, P)\sqcup(F', \gamma', P'))\\
 \V(M, \Sigma, \Gamma)\otimes \V(M', \Sigma', \Gamma)&\mapsto \V((M, \Sigma, \Gamma')\sqcup(M', \Sigma', \Gamma)).
\end{align*}
We show $\iota$ is a natural isomorphism.

We first show $\iota$ is injective. Suppose $\kappa$ is an element in the kernel of $\iota$. Since $\V(F, \gamma, P)$ is of finite dimensional, we may assume $$\kappa=\sum_{i=1}^{N}\V(M_i, \Sigma_i, \Gamma_i)\otimes \V(M_i', \Sigma_i', \Gamma_i').$$
Note that $\kappa\in \ker(\iota)$ implies that 
$$\sum_{i=1}^{N}Z(((M_i, \Sigma_i, \Gamma_i)\sqcup(M_i', \Sigma_i', \Gamma_i'))\sqcup_{(F\sqcup F')}(M'', \Sigma'', \Gamma''))=0$$
for all $(M'', \Sigma'', \Gamma'')$ with boundary $(F, \gamma, P)\sqcup(F', \gamma', P')$. 
In particular,
the equality hold when $(M'', \Sigma'', \Gamma'')$ is homeomorphic to a disjoint union of two alterfold cobordisms with their boundary equal to $(F, \gamma, P)$ and $(F', \gamma', P')$. 
On the other hand, the images of such alterfold cobordisms under the functor $\V$ span $\V(F, \gamma, P)^*\otimes \V(F', \gamma', P')^*$. 
Thus $\kappa=0$, and injectivity of $\iota$ is proved.

Then we show $\iota$ is surjective. By Theorem \ref{thm:collarrepn}, for all $3$-alterfold $(M, \Sigma, \Gamma)$ with boundary $(-F, -\gamma, -P)\sqcup(-F', -\gamma', -P')$. 
There exist a collar representative assigned to the same vector by $\V$. 
The surjectivity of $\iota$ follows since collar representative is in the image of $\iota$.

By the definition of $\iota$, we have the following commutative diagram:
$$
\begin{tikzcd}
\V(\bF_1)\otimes \V(\bF_2)\otimes \V(\bF_3)\arrow[r, "\iota\otimes \id"]\arrow[d, "\id\otimes \iota"']& \V(\bF_1\sqcup \bF_2)\otimes \V(\bF_3)\arrow[d, "\iota"]\\
\V(\bF_1)\otimes \V(\bF_2\sqcup \bF_3)\arrow[r, "\iota"]& \V(\bF_1\sqcup \bF_2\sqcup \bF_3)\\
\end{tikzcd}
$$
where $\bF_i=(F_i, \gamma_i, P_i)$ is a $\cC$-decorated $2$-alterfold for $i=1, 2, 3$. 
Therefore $(\V, \iota)$ is a monoidal functor, which is symmetric.
\end{proof}

\subsection{TQFT Basis via Polygon Decomposition}

In this section, we show that given a truncated polygon decomposition of $B$-colored region of the underlying $2$-alterfold $(F, \gamma, P)$, we can identify the space $\V(F, \gamma, P)$ as (direct sum of tensor product of) morphism spaces of $\mathcal{C}$, In particular, we have an explicit description of the basis $\V(F, \gamma, P)$ for such alterfolds.

Let $(F, \gamma, P)$ be a decorated $2$-alterfold. Without loss of generality, we may assume the surface $F$ is connected, otherwise, the basis can be chosen to be the tensor product of the basis of each connected components.

An \textit{truncated $n$-gon} $\Delta$ is a $2n$-gon obtained from truncating all the vertices of an $n$-gon. We depict a truncated $n$-gon by a $2n$-gon whose edges alternatively colored by black and green, where the black edges are created by the truncation.

\begin{definition}
An truncated polygon decomposition $\mathcal{P}$ of a surface $S$ with nonempty boundary consists of a disjoint union of truncated polygon $X=\bigsqcup\Delta_i$ and a collection of homeomorphisms $\Phi$ identifying pairs of edges in $X$ such that the quotient space $X/\Phi$ is homeomorphic to $S$. We say $\mathcal{P}$ is a truncated polygon decomposition of a $\mathcal{C}$-decorated $2$-alterfold $(F, \gamma, P)$ if $\mathcal{P}$ is a truncated polygon decomposition of $R_{B}(F)$ and the green edges of $\mathcal{P}$ do not meet the marked points $P$. 
\end{definition}

\begin{example}
A truncated polygon decomposition of a $3$-holed sphere.
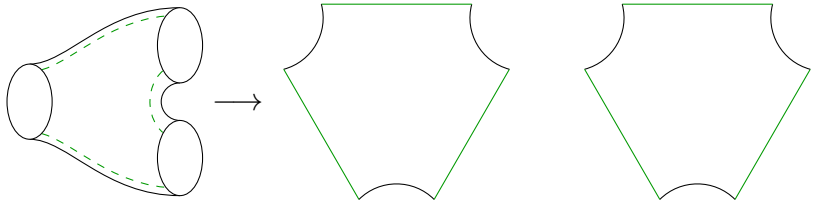
\[
\vcenter{\hbox{
\begin{tikzpicture}
\draw (0, 0.75) [partial ellipse=90:270: 0.25 and 0.25];
\draw (0, 2)..controls +(-1, 0) and (-1.5, 1.25)..(-2, 1.25);
\draw (0, -0.5)..controls +(-1, 0) and (-1.5, 0.25)..(-2, 0.25);
\draw[\darkgreen, dashed](0, 1.9)..controls +(-1, 0) and (-1.5, 1.15)..(-2, 1.15);
\draw[\darkgreen, dashed](0, -0.4)..controls +(-1, 0) and (-1.5, 0.35)..(-2, 0.35);
\draw[\darkgreen, dashed](0, 0.75)[partial ellipse=125:235:0.4 and 0.5];
\draw[fill=white] (0, 0) [partial ellipse=0:360:0.3 and 0.5];
\draw[fill=white] (0, 1.5) [partial ellipse=0:360:0.3 and 0.5];
\begin{scope}[xshift=-2cm, yshift=0.75cm]
\draw[fill=white] (0, 0) [partial ellipse=0:360:0.3 and 0.5];
\end{scope}
\end{tikzpicture}}}
\longrightarrow
\vcenter{\hbox{
\begin{tikzpicture}
\draw (-0.5, -1.732) to[bend left=45] (0.5, -1.732);
\draw[\darkgreen] (0.5, -1.732) to (1.5, 0);
\draw (1.5, 0) to[bend left=45] (1, 0.866);
\draw[\darkgreen] (1, 0.866) to (-1, 0.866);
\draw (-1, 0.866) to[bend left=45] (-1.5, 0);
\draw[\darkgreen] (-1.5, 0) to (-0.5, -1.732);
\begin{scope}[xshift=4cm]
\draw (-0.5, -1.732) to[bend left=45] (0.5, -1.732);
\draw[\darkgreen] (0.5, -1.732) to (1.5, 0);
\draw (1.5, 0) to[bend left=45] (1, 0.866);
\draw[\darkgreen] (1, 0.866) to (-1, 0.866);
\draw (-1, 0.866) to[bend left=45] (-1.5, 0);
\draw[\darkgreen] (-1.5, 0) to (-0.5, -1.732);
\end{scope}
\end{tikzpicture}}}
\]
\end{example}

\begin{remark}
Notice that a $2$-alterfold $(F, \gamma, P)$ admits a truncated polygon decomposition if all connected components of $R_{B}(F)$ has nonempty boundary. If we further assume that no component is not homeomorphic to a cylinder or disk, then it admit a truncated triangulation, which is equivalent to the ideal triangulation (or polygon decomposition) for punctured surface as in \cite{Pen87}.
\end{remark}

Let $(F, \gamma, P)$ be a connected $\mathcal{C}$-decorated $2$-alterfold with $R_{A}(F)\ne \emptyset$ and $\mathcal{P}=\{\Delta_1, \cdots \Delta_m\}$ be a truncated polygon decomposition of $(F, \gamma, P)$, we associate it with a finite set of vectors $B_{\mathcal{P}}\subset \V(F, \gamma, P)$ as follows. 
Choose a coloring $c: E_G\rightarrow \Irr_{\cC}$ of the green edges $E_G$ of $\mathcal{P}$. 
To each $e \in E_G$, add a marked point labeled by $c(e)$ to the middle of $e$, the orientation of the marked point is chosen to be consistent with the orientation of the edge $e$ and the orientation of $\Delta_i$. 
Denote by $\Delta_{i}^{c}$ the $i$-th truncated polygon with these additional marked points. 

We associate to each $\Delta_i^{c}$ a morphism space $V_{i, c}=\hom(\1, X(\Delta_i^{c}))$, where $X(\Delta_{i}^{c})$ denotes the tensor product of objects corresponding to the marked points on the boundary of $\Delta_{i}^{c}$, with the order chosen to be the boundary orientation of $\Delta_i$. Let $B_{\Delta_{i}^{c}}$ be a set of basis of $V_{i, c}$. Note that as $\mathcal{C}$ is spherical, up to canonical isomorphisms,  $X(\Delta_{i}^{c})$ (and hence $B_{\Delta_{i}^{c}}$) only depends on the cyclic ordering of the marked points on the boundary of $\Delta_i$.

We define the set of the vectors
\begin{equation}\label{eq:basis}
    B_{\mathcal{P}}=\left\{v_1^{c}\otimes\cdots\otimes v_{m}^c \left| v_i^{c}\in B_{\Delta_{i}^{c}}; \  c: E_{G}\rightarrow \Irr_{\cC} \right. \right\}.
\end{equation}
We define a mapping $\Phi: B_{\mathcal{P}}\rightarrow \V(F, \gamma, P)$ as follows. We define $\Phi(v_1^{c}\otimes\cdots\otimes v_{m}^c)$ to be the following alterfold cobordism.
First put the coupon with morphism $v_i^{c}$ in the middle of $\Delta_{i}$, then glue up the edges to get a tensor diagram embedded in the $R_{B}(F, \gamma, P)$, which can be further identified as a collar representative in $\V(F, \gamma, P)$.

\begin{example} Gluing up a $3$-holed sphere. 
    $$
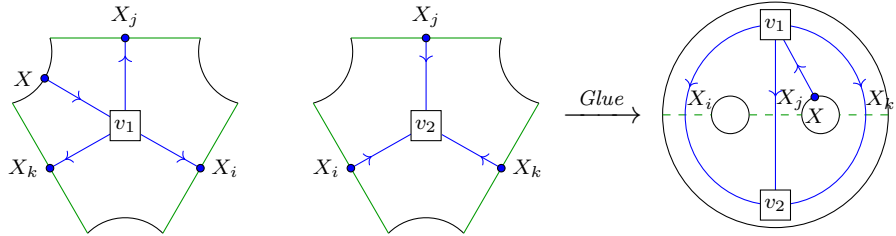
\vcenter{\hbox{
\begin{tikzpicture}
\draw (-0.5, -1.732) to[bend left=45] (0.5, -1.732);
\draw[\darkgreen] (0.5, -1.732) -- (1.5, 0) node(A)[midway]{};
\draw (1.5, 0) to[bend left=45] (1, 0.866);
\draw[\darkgreen] (1, 0.866) -- (-1, 0.866) node(B)[midway]{};
\draw (-1, 0.866) to[bend left=45]node(G)[midway]{} (-1.5, 0);
\draw[blue, ->-=0.45] (G.center)--(0, -0.3);
\draw[\darkgreen] (-1.5, 0) -- (-0.5, -1.732) node(C)[midway]{};
\begin{scope}[xshift=4cm]
\draw (-0.5, -1.732) to[bend left=45] (0.5, -1.732);
\draw[\darkgreen] (0.5, -1.732) -- (1.5, 0) node(D)[midway]{};
\draw (1.5, 0) to[bend left=45] (1, 0.866);
\draw[\darkgreen] (1, 0.866) -- (-1, 0.866) node(E)[midway]{};
\draw (-1, 0.866) to[bend left=45] (-1.5, 0);
\draw[\darkgreen] (-1.5, 0) -- (-0.5, -1.732) node(F)[midway]{};
\draw [blue, -<-=0.8] (0, -0.3)--(D.center);
\draw [blue, -<-=0.8] (0, -0.3)--(E.center);
\draw [blue, -<-=0.8] (0, -0.3)--(F.center);
\draw[fill=white] (-0.2, -0.5) rectangle (0.2, -0.1);
\draw (0, -0.3) node{\tiny$v_2$};
\draw[fill=blue] (F) node[left]{\tiny$X_i$} circle(0.05);
\draw[fill=blue] (E) node[above]{\tiny$X_j$} circle(0.05);
\draw[fill=blue] (D) node[right]{\tiny$X_k$} circle(0.05);
\end{scope}
\draw [blue, ->-=0.8] (0, -0.3)--(A.center);
\draw [blue, ->-=0.8] (0, -0.3)--(B.center);
\draw [blue, ->-=0.8] (0, -0.3)--(C.center);
\draw[fill=white] (-0.2, -0.5) rectangle (0.2, -0.1);
\draw (0, -0.3) node{\tiny$v_1$};
\draw[fill=blue] (A) node[right]{\tiny$X_i$} circle(0.05);
\draw[fill=blue] (B) node[above]{\tiny$X_j$} circle(0.05);
\draw[fill=blue] (C) node[left]{\tiny$X_k$} circle(0.05);
\draw[fill=blue] (G) node[left]{\tiny$X$} circle(0.05);
\end{tikzpicture}}}
\xrightarrow{\text{     Glue     }}
\vcenter{\hbox{
\begin{tikzpicture}
\draw (0, 0) [partial ellipse=0:360:1.5 and 1.5];
\draw (-0.6, 0) [partial ellipse=0:360:0.25 and 0.25];
\draw (0.6, 0) [partial ellipse=0:360:0.25 and 0.25] node(A)[pos=0.3]{};
\draw[blue, ->-=0.45] (A.center)--(0, 1.2);
\draw[fill=blue] (A) node[below]{\tiny$X$} circle(0.05);
\draw[\darkgreen, dashed] (-1.5, 0)--(-0.85, 0);
\draw[\darkgreen, dashed] (1.5, 0)--(0.85, 0);
\draw[\darkgreen, dashed] (-0.35, 0)--(0.35, 0);
\draw[blue,->-=0.4] (0, 0)[partial ellipse=90:270:1.2]; 
\node[black] at (-1,0.2) {\tiny$X_i$};
\draw[blue,->-=0.4] (0, 0)[partial ellipse=90:-90:1.2];
\node[black] at (0.2, 0.2){\tiny$X_j$};
\draw[blue,->-=0.4] (0, 1.2)--(0, -1.2);
\node[black] at (1.4, 0.2){\tiny$X_k$};
\draw[fill=white] (-0.2, 1) rectangle (0.2, 1.4);
\draw (0, 1.2) node{\tiny$v_1$};
\draw[fill=white] (-0.2, -1) rectangle (0.2, -1.4);
\draw (0, -1.2) node{\tiny$v_2$};
\end{tikzpicture}
}}
$$
\end{example}


Now we show $B_{\mathcal{P}}$ forms a basis of $\V(F, \gamma, P)$. We first show
\begin{lemma}\label{lem:surjectivity}
The linearization of $\Phi$ from $\text{span}(B_{\mathcal{P}})$ to $ \V(F, \gamma, P)$ is surjective.
We denote the linearization of $\Phi$ by $\Phi$ again.
\end{lemma}

\begin{proof}
We only need to show all vectors in $\V(F, \gamma, P)$ can be written as a linear combination of vectors in $\Phi(B_{\mathcal{P}})$. Since collar representatives form a spanning set of $V(F, \gamma, P)$, we only need to verify for all collar representative $\V(M, \Sigma, \Gamma)\in \V(F, \gamma, P)$, it can be written as linear combination of vectors in $\Phi(B_{\mathcal{P}})$.

We may assume $\Gamma$ and $\mathcal{P}$ are in generic position, in the sense that all coupons of $\Gamma$ lies in the interior of truncated polyhedrons and all edges of $\Gamma$ meet green edges transversely, otherwise we can isotope $\Gamma$. 
Apply the Planar Graphical Calculus if necessary, we may assume $\Gamma$ intersects each green edge at one point, colored by a simple object in $\mathcal{C}$. 
Further apply the Planar Graphical Calculus, we can simplify tensor diagrams in each truncated polygon such that there is only one coupon. 
$$
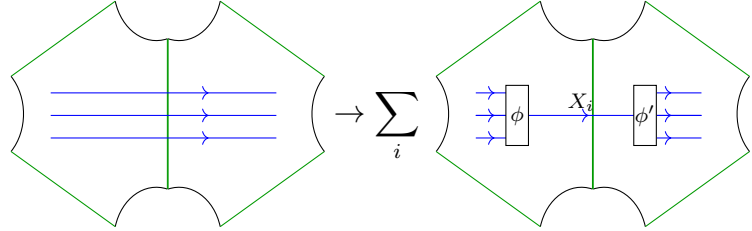
\vcenter{\hbox{
\begin{tikzpicture}
\begin{scope}[xscale=0.8]
\begin{scope}[xshift=-2.61cm]
\begin{scope}[rotate=30]
\draw (-0.5, -1.732) to[bend left=45] (0.5, -1.732);
\draw[\darkgreen] (0.5, -1.732) to (1.5, 0);
\draw (1.5, 0) to[bend left=45] (1, 0.866);
\draw[\darkgreen] (1, 0.866) to (-1, 0.866);
\draw (-1, 0.866) to[bend left=45] (-1.5, 0);
\draw[\darkgreen] (-1.5, 0) to (-0.5, -1.732);
\end{scope}
\end{scope}
\begin{scope}[rotate=-30]
\draw (-0.5, -1.732) to[bend left=45] (0.5, -1.732);
\draw[\darkgreen] (0.5, -1.732) to (1.5, 0);
\draw (1.5, 0) to[bend left=45] (1, 0.866);
\draw[\darkgreen] (1, 0.866) to (-1, 0.866);
\draw (-1, 0.866) to[bend left=45] (-1.5, 0);
\draw[\darkgreen] (-1.5, 0) to (-0.5, -1.732);
\end{scope}
\end{scope}
\begin{scope}[rotate=90, yshift=1.4cm, xshift=-0.27cm]
\draw [blue, ->-=0.7] (0, 1.2)--(0, -1.8);
\draw [blue, ->-=0.7] (-0.3, 1.2)--(-0.3, -1.8);
\draw [blue, ->-=0.7] (0.3, 1.2)--(0.3, -1.8);
\end{scope}
\end{tikzpicture}}}
\rightarrow \sum_{i}
\vcenter{\hbox{
\begin{tikzpicture}
\begin{scope}[xscale=0.8]
\begin{scope}[xshift=-2.61cm]
\begin{scope}[rotate=30]
\draw (-0.5, -1.732) to[bend left=45] (0.5, -1.732);
\draw[\darkgreen] (0.5, -1.732) to (1.5, 0);
\draw (1.5, 0) to[bend left=45] (1, 0.866);
\draw[\darkgreen] (1, 0.866) to (-1, 0.866);
\draw (-1, 0.866) to[bend left=45] (-1.5, 0);
\draw[\darkgreen] (-1.5, 0) to (-0.5, -1.732);
\end{scope}
\end{scope}
\begin{scope}[rotate=-30]
\draw (-0.5, -1.732) to[bend left=45] (0.5, -1.732);
\draw[\darkgreen] (0.5, -1.732) to (1.5, 0);
\draw (1.5, 0) to[bend left=45] (1, 0.866);
\draw[\darkgreen] (1, 0.866) to (-1, 0.866);
\draw (-1, 0.866) to[bend left=45] (-1.5, 0);
\draw[\darkgreen] (-1.5, 0) to (-0.5, -1.732);
\end{scope}
\end{scope}
\begin{scope}[rotate=90, yshift=1.4cm, xshift=-0.27cm]
\draw [blue, ->-=0.5] (0, 1.2)--(0, 0.7);
\draw [blue, ->-=0.5] (-0.3, 1.2)--(-0.3, 0.7);
\draw [blue, ->-=0.5] (0.3, 1.2)--(0.3, 0.7);
\draw [blue, ->-=0.5] (0, -1)--(0, -1.8);
\draw [blue, ->-=0.5] (-0.3, -1)--(-0.3, -1.8);
\draw [blue, ->-=0.5] (0.3, -1)--(0.3, -1.8);
\draw [blue, ->-=0.5] (0, 1.2)--(0, -1.8);
\draw [fill=white](-0.4, 0.5) rectangle (0.4, 0.8); \node at (0, 0.65) {\tiny{$\phi$}};
\draw [fill=white](-0.4, -1.2) rectangle (0.4, -0.9); \node at (0, -1.05) {\tiny{$\phi'$}};
\end{scope}
\draw (-1.2,-0.1) node{\tiny$X_i$};
\end{tikzpicture}}}
$$
To each summand, there is a natural coloring $c: E_G\rightarrow \Irr_{\mathcal{C}}$ determined by the label of the edge in $\Gamma$ intersecting the green edge. 
Moreover, the morphism in the coupon can be written as a linear combination of vectors in $B_{\Delta_{i}^{c}}$ since the latter set is chosen to be a basis. 
Therefore, we prove the surjectivity.

\end{proof}

\begin{theorem}\label{thm:tvbasis}
Let $(F, \gamma, P)$ be a connected $\mathcal{C}$-decorated $2$-alterfold and $\mathcal{P}=\{\Delta_1, \cdots \Delta_m\}$ be a truncated polygon decomposition of $(F, \gamma, P)$. Then
$\Phi(B_{\mathcal{P}})$ is a basis of $\V(F, \gamma, P)$.
\end{theorem}

\begin{proof}
The surjectivity is proved in Lemma \ref{lem:surjectivity}. 
Thus it remains to show vectors in $\Phi(B_{\mathcal{P}})$ are linearly independent. 
We prove this by showing there exists a set of dual vectors under the pairing $[\ ,\ ]$ in $V(-F, -\gamma, -P)$.

Let $\mathcal{\overline{P}}$ be the truncated triangulation of $(-F, -\gamma, -P)$ that equals to the mirror image triangulation of $\mathcal{P}$. We denote the mirror of a green edge $e$(resp. truncated triangle $\Delta_i$) of $\mathcal{P}$ by $\overline{e}$(resp. $\overline{\Delta_i}$). Without ambiguity, we still denote the mirror of the coloring of edges of $c$ by $c$.

Now we define a set of vectors 
$$
\overline{B_{\mathcal{P}}}
=\{\overline{v_{1}^{c}}\otimes \cdots \otimes \overline{{v}_{m}^{c}}|\overline{v_i^{c}}\subset B_{\overline{\Delta_{i}^{c}}};\  c:E_{G}\rightarrow \Irr_{\cC}\}
$$
where $B_{\overline{\Delta_{i}^{c}}}$ is the basis of $\hom(\1, X(\overline{\Delta_{i}^{c}}))\cong \hom(X(\Delta_{i}^{c}), \1)$ that is dual to $B_{\Delta_{i}^{c}}=\hom(\1, X(\Delta_{i}^{c}))$. In a similar way as in the $B_\mathcal{P}$ case, we can define a mapping $\Psi: \overline{B_\mathcal{P}}\rightarrow \V(-F, -\gamma, -P)$. 
We claim $\Psi(\overline{B_{\mathcal{P}}})$ is dual to $\Phi(B_{\mathcal{P}})$ under the pairing $[\ ,\ ]$.

To prove the claim, we first show if the colorings $c\ne c'$, then $[\Phi(v_1^{c}\otimes\cdots\otimes v_m^{c}), \Psi(\overline{v_1^{c'}}\otimes\ldots\otimes\overline{v_m^{c'}})]=0$. If $c\ne c'$, there exist edges $e$ and $e'$ colored differently. 
On the other hand, $e\cup e'$ bound a $B$-colored disk in the $\mathcal{C}$-decorated $3$-alterfold obtained from gluing the collar representatives. 
Applying Move $2$ to this disk, one ends up with a summation of $\mathcal{C}$-decorated $3$-alterfold over an orthogonal basis of $\hom(\1, c(e)\otimes c'(\overline{e}))\otimes \hom(c(e)\otimes c'(\overline{e}), \1)$. 
However, this is an empty set if $c(e)$ and $c'(\overline{e})$ are not dual to each other.

Fix the coloring $c$, we compute the pairing $[\Phi(v_1^{c}\otimes \cdots\otimes v_m^{c}), \Psi(\overline{w_1^{c}}\otimes\cdots\otimes\overline{w_m^{c}})]$. 
Notice for all green edge $e$, $e\cup \overline{e}$ bound a disk in $R_{B}(\Phi(v_1^{c}\otimes\cdots\otimes v_m^{c})\circ \Psi(\overline{w_1^{c}}\otimes\cdots\otimes\overline{w_m^{c}}))$. 
Applying Move $2$ to each of such disks, one ends up with a $3$-alterfold with $B$-colored region equal to the union of $3$-balls $B_i$ with their boundary $\partial B_i=\Delta_{i}\cup \overline{\Delta_{i}}$.
Further evaluating the closed $\mathcal{C}$ diagram on these spheres. 
One have
$$[\Phi(v_1^{c}\otimes\cdots\otimes v_m^{c}), \Psi(\overline{w_1^{c}}\otimes\cdots\otimes\overline{w_m^{c}})]=\prod_{i} \overline{w_{i}^{c}}\circ v_i^{c}.$$
Since we choose $B_{\overline{\Delta_{i}^{c}}}$ to consist a set of dual basis of $B_{\Delta_{i}^{c}}$, 
$B_{\mathcal{P}}$ and $B_{\overline{\mathcal{P}}}$ form a dual basis on each block associated to $c$. 
Thus we proved the claim and the theorem follows.

\end{proof}

Notice the condition $R_{A}(F)\ne \emptyset$ is essential for the existence of a truncated triangulation. 
Therefore, when $R_{A}(F)=\emptyset$, we are not able to identify $\V(F, \emptyset_B, \emptyset)$ as a morphism space in $\cC$, where $\emptyset_B$ means $R_A(F)=\emptyset$. However, we can identify it as a summand of a morphism space. 
We need the following lemma embedding a $B$-colored closed surface to a surface admit a truncated polygon decomposition.

Let $(F, \gamma, P)$ be a $\mathcal{C}$-decorated $2$-alterfold, and let $c$ be the boundary of a disk $\mathbb{D} \subset R_B(F)$.  
We denote by $(F, \gamma\sqcup c, P)$ the $\mathcal{C}$-decorated $2$-alterfold obtained from $(F, \gamma, P)$ by changing the color of $\mathbb{D}$ from $B$ to $A$. 
Consider the 3-alterfold cobordism $\mathbb{M}=(F\times [0,1], \gamma\times [0, 1]\sqcup \tau, P\times [0,1])$ depicted below with $\partial_{-}M=(F, \gamma, P)$ and $\partial_{+}M=(F, \gamma\sqcup c, P)$, where $\tau$ is a hemisphere bounding an $A$-colored half-ball with $\partial \tau=c$. 

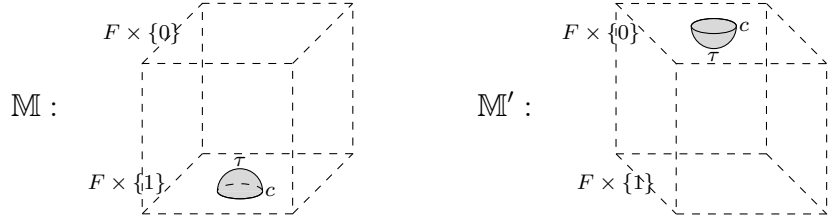
$$
\BM: \vcenter{\hbox{
\begin{tikzpicture}
\draw[dashed] (0, 0) rectangle (2, 2);
\draw[dashed] (.8, .8) rectangle (2.8, 2.8);
\draw[dashed] (0, 2)--+(0.8, 0.8);
\draw[dashed] (2, 2)--+(0.8, 0.8);
\draw[dashed] (2, 0)--+(0.8, 0.8);
\draw[dashed] (0, 0)--+(0.8, 0.8);
\draw (0, 2.4) node{\tiny$F\times \{0\}$};
\draw (-0.2, 0.4) node{\tiny$F\times \{1\}$};
\begin{scope}[yshift=-1cm]
\filldraw[fill=white!70!gray] (1.3, 1.3) [partial ellipse=180:360:0.3 and 0.1];
\filldraw[fill=white!70!gray] (1.3, 1.3) [partial ellipse=0:180:0.3 and 0.3];
\draw (1.7, 1.3) node{\tiny$c$};
\draw (1.3, 1.7) node{\tiny$\tau$};
\draw[dashed] (1.3, 1.3) [partial ellipse=0:180:0.3 and 0.1];
\draw (1.3, 1.3) [partial ellipse=180:360:0.3 and 0.1];
\end{scope}
\end{tikzpicture}}}
\qquad\qquad
\BM': \vcenter{\hbox{
\begin{tikzpicture}[yscale=-1]
\draw[dashed] (0, 0) rectangle (2, 2);
\draw[dashed] (.8, .8) rectangle (2.8, 2.8);
\draw[dashed] (0, 2)--+(0.8, 0.8);
\draw[dashed] (2, 2)--+(0.8, 0.8);
\draw[dashed] (2, 0)--+(0.8, 0.8);
\draw[dashed] (0, 0)--+(0.8, 0.8);
\draw (0, 2.4) node{\tiny$F\times \{1\}$};
\draw (-0.2, 0.4) node{\tiny$F\times \{0\}$};
\begin{scope}[yshift=-1cm]
\filldraw[fill=white!70!gray] (1.3, 1.3) [partial ellipse=180:360:0.3 and 0.1];
\filldraw[fill=white!70!gray] (1.3, 1.3) [partial ellipse=0:180:0.3 and 0.3];
\draw (1.7, 1.3) node{\tiny$c$};
\draw (1.3, 1.7) node{\tiny$\tau$};
\draw (1.3, 1.3) [partial ellipse=0:180:0.3 and 0.1];
\draw (1.3, 1.3) [partial ellipse=180:360:0.3 and 0.1];
\end{scope}
\end{tikzpicture}}}
$$

\begin{lemma}\label{lem:embed}
The linear map $\V(\mathbb{M}): \V(F, \gamma, P) \to \V(F, \gamma\sqcup c, P)$ is injective. 
\end{lemma}
\begin{proof}
Let $\mathbb{M'}$ be the 3-alterfold depicted above as the mirror image of $\BM$, and consider the morphism $\V(\mathbb{M'}): \V(F, \gamma\sqcup c, P)\rightarrow \V(F, \gamma, P)$. 
Notice by definition, 
$$\V(\mathbb{M'})\circ \V(\mathbb{M})=\V(\mathbb{M}\sqcup_{(F, \gamma\sqcup c, P)}\mathbb{M'})=\mu\V(\id_{(F,\gamma, P)})=\id_{\V(F, \gamma, P)},$$
where the middle equality comes from applying Move $0$ to remove the $A$-colored ball in the middle. Therefore, $\V(\mathbb{M})$ is an injective map and $\V(\mathbb{M'})$ is surjective.
\end{proof}

Lemma \ref{lem:embed} implies that $\V(F, \emptyset_B, \emptyset)$ can be identified as the image of $\V(\mathbb{M})\circ\V(\mathbb{M'})$, which is a subspace of $\V(F, c, \emptyset)$. On the other hand,
$$
\V(\mathbb{M})\circ\V(\mathbb{M'})=
\V\left(
\vcenter{\hbox{
\begin{tikzpicture}
\draw[dashed] (0, 0) rectangle (2, 2);
\draw[dashed] (.8, .8) rectangle (2.8, 2.8);
\draw[dashed] (0, 2)--+(0.8, 0.8);
\draw[dashed] (2, 2)--+(0.8, 0.8);
\draw[dashed] (2, 0)--+(0.8, 0.8);
\draw[dashed] (0, 0)--+(0.8, 0.8);
\draw (0, 2.4) node{\tiny$F\times \{0\}$};
\draw (-0.2, 0.4) node{\tiny$F\times \{1\}$};
\begin{scope}[yshift=-1cm]
\filldraw[fill=white!70!gray] (1.3, 1.3) [partial ellipse=180:360:0.3 and 0.1];
\filldraw[fill=white!70!gray] (1.3, 1.3) [partial ellipse=0:180:0.3 and 0.3];
\draw (1.7, 1.3) node{\tiny$c$};
\draw (1.3, 1.7) node{};
\draw[dashed] (1.3, 1.3) [partial ellipse=0:180:0.3 and 0.1];
\draw (1.3, 1.3) [partial ellipse=180:360:0.3 and 0.1];
\end{scope}
\begin{scope}[yscale=-1, yshift=-3.7cm]
\filldraw[fill=white!70!gray] (1.3, 1.3) [partial ellipse=180:360:0.3 and 0.1];
\filldraw[fill=white!70!gray] (1.3, 1.3) [partial ellipse=0:180:0.3 and 0.3];
\draw (1.7, 1.3) node{\tiny$c$};
\draw (1.3, 1.7) node{};
\draw (1.3, 1.3) [partial ellipse=0:180:0.3 and 0.1];
\draw (1.3, 1.3) [partial ellipse=180:360:0.3 and 0.1];
\end{scope}
\end{tikzpicture}
}}
\space\right)
=\V\left(
\vcenter{\hbox{
\begin{tikzpicture}
\draw[dashed] (.8, .8) rectangle (2.8, 2.8);
\draw[dashed] (0, 2)--+(0.8, 0.8);
\draw[dashed] (2, 2)--+(0.8, 0.8);
\draw[dashed] (2, 0)--+(0.8, 0.8);
\draw[dashed] (0, 0)--+(0.8, 0.8);
\draw (0, 2.4) node{\tiny$F\times \{0\}$};
\draw (-0.2, 0.4) node{\tiny$F\times \{1\}$};
\begin{scope}[yshift=-1cm]
\filldraw[white!70!gray] (1.3, 1.3) [partial ellipse=180:360:0.3 and 0.1];
\draw (1.7, 1.3) node{\tiny$c$};
\draw[dashed] (1.3, 1.3) [partial ellipse=0:180:0.3 and 0.1];
\draw (1.3, 1.3) [partial ellipse=180:360:0.3 and 0.1];
\filldraw[white!70!gray] (1, 1.31) rectangle (1.6, 3.39);
\draw[red, dashed](1.3, 2.35) [partial ellipse=0:180:0.3 and 0.1];
\draw[red](1.3, 2.35) [partial ellipse=180:360:0.3 and 0.1];
\end{scope}
\begin{scope}[yscale=-1, yshift=-3.7cm]
\filldraw[white!70!gray] (1.3, 1.3) [partial ellipse=180:360:0.3 and 0.1];
\draw (1.7, 1.3) node{\tiny$c$};
\draw (1.3, 1.3) [partial ellipse=0:180:0.3 and 0.1];
\draw (1.3, 1.3) [partial ellipse=0:360:0.3 and 0.1];
\end{scope}
\draw[dashed] (0, 0) rectangle (2, 2);
\end{tikzpicture}
}}
\right)\,,
$$
where the last equality follows by applying Move $1$. Combining these facts, we have the following proposition.

\begin{proposition}
Let $(F, \emptyset_B, \emptyset)$ be a $\mathcal{C}$-decorated 2-alterfold with $R_B(F) = F$, and $c \subset F$ a circle bounding a disk. Then the vector space $\V(F, \emptyset_B, \emptyset)$, as a subspace of $\V(F, c, \emptyset)$, is spanned by collar representatives with an $\Omega$-colored circle at $c\times \left\{\displaystyle \frac{\epsilon}{2}\right\}$.\qed
\end{proposition}

\begin{theorem}\label{thm:utqft}
Suppose $\mathcal{C}$ is a unitary fusion category. 
Then the functor $\V$ is a unitary topological quantum field theory.
\end{theorem}
\begin{proof}
Let $\mathbb{M}=(M, \Sigma, \Gamma)$ be a 3-alterfold with time boundary.
We define its complex conjugation $\mathbb{M}^{\dagger}=(M^{\dagger}, \Sigma^{\dagger}, \Gamma^{\dagger})$ to be its mirror image. 
To be specific, the 3-alterfold $(M^{\dagger}, \Sigma^{\dagger})$ is obtain by reversing the orientation of the 3-alterfold $(M, \Sigma)$, and $\Gamma^{\dagger}$ is a tensor diagram over $\Sigma^{\dagger}$ with its edges identical to $\Gamma$ with orientation reversed and with the morphisms in the coupons equal to the complex conjugation of the corresponding ones in $\Gamma$.

We first claim the dagger operation on the decorated cobordism category defined above descends to a complex conjugation
$$\dagger: \V(F, \gamma, P)\rightarrow \V(-F, -\gamma, -P).$$ This is because the dagger operation intertwines isotopy, Moves $0$, $1$, $3$ and Planar Graphical Calculus in an obvious way. It intertwines Move $2$ because of the following commutative diagram.

$$
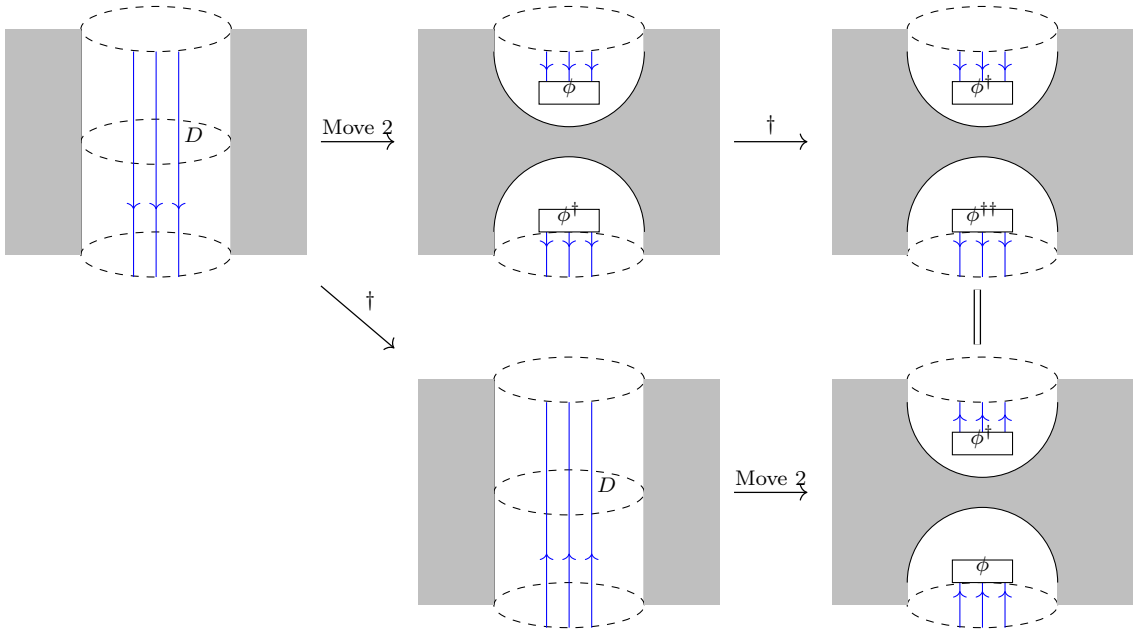
\begin{tikzcd}
{\vcenter{\hbox{
\begin{tikzpicture}
\draw[dashed] (0, 1.5) [partial ellipse=0:360:1 and 0.3];
\draw[dashed] (0, 0) [partial ellipse=0:360:1 and 0.3];
\draw[dashed] (0, -1.5) [partial ellipse=0:360:1 and 0.3];
\draw (0.5, 0) node{\tiny{$D$}};
\draw (-1, -1.5)--(-1, 1.5);
\draw (1, -1.5)--(1, 1.5);
\path[fill=gray!50!white]
(-1, -1.5) rectangle (-2, 1.5);
\path[fill=gray!50!white]
(1, -1.5) rectangle (2, 1.5);
\draw [blue, ->-=0.7] (0, 1.2)--(0, -1.8);
\draw [blue, ->-=0.7] (-0.3, 1.2)--(-0.3, -1.8);
\draw [blue, ->-=0.7] (0.3, 1.2)--(0.3, -1.8);
\end{tikzpicture}}}}\arrow[r, "\text{Move $2$}"]\arrow[dr, "\dagger"]&\vcenter{\hbox{
\begin{tikzpicture}
\draw[dashed] (0, 1.5) [partial ellipse=0:360:1 and 0.3];
\draw[dashed] (0, -1.5) [partial ellipse=0:360:1 and 0.3];
\draw (-1, -1.5)--(-1, -1.2);
\draw (1, 1.5)--(1, 1.2);
\draw (-1, 1.5)--(-1, 1.2);
\draw (1, -1.5)--(1, -1.2);
\path[fill=gray!50!white]
(-1, -1.5) rectangle (-2, 1.5);
\path[fill=gray!50!white]
(1, -1.5) rectangle (2, 1.5);
\draw [blue, ->-=0.5] (0, 1.2)--(0, 0.7);
\draw [blue, ->-=0.5] (-0.3, 1.2)--(-0.3, 0.7);
\draw [blue, ->-=0.5] (0.3, 1.2)--(0.3, 0.7);
\draw [blue, ->-=0.5] (0, -1)--(0, -1.8);
\draw [blue, ->-=0.5] (-0.3, -1)--(-0.3, -1.8);
\draw [blue, ->-=0.5] (0.3, -1)--(0.3, -1.8);
\path[fill=gray!50!white]
(-1, -1.2) arc (180:0:1)--(1, 1.2) arc (0:-180:1)--(-1, -1.2);
\draw (-1, -1.2) arc (180:0:1);
\draw (1, 1.2) arc (0:-180:1);
\draw [fill=white](-0.4, 0.5) rectangle (0.4, 0.8); \node at (0, 0.65) {\tiny{$\phi$}};
\draw [fill=white](-0.4, -1.2) rectangle (0.4, -0.9); \node at (0, -1.05) {\tiny{$\phi^{\dagger}$}};
\end{tikzpicture}}} \arrow[r, "\dagger"]&\vcenter{\hbox{
\begin{tikzpicture}
\draw[dashed] (0, 1.5) [partial ellipse=0:360:1 and 0.3];
\draw[dashed] (0, -1.5) [partial ellipse=0:360:1 and 0.3];
\draw (-1, -1.5)--(-1, -1.2);
\draw (1, 1.5)--(1, 1.2);
\draw (-1, 1.5)--(-1, 1.2);
\draw (1, -1.5)--(1, -1.2);
\path[fill=gray!50!white]
(-1, -1.5) rectangle (-2, 1.5);
\path[fill=gray!50!white]
(1, -1.5) rectangle (2, 1.5);
\draw [blue, ->-=0.5] (0, 1.2)--(0, 0.7);
\draw [blue, ->-=0.5] (-0.3, 1.2)--(-0.3, 0.7);
\draw [blue, ->-=0.5] (0.3, 1.2)--(0.3, 0.7);
\draw [blue, ->-=0.5] (0, -1)--(0, -1.8);
\draw [blue, ->-=0.5] (-0.3, -1)--(-0.3, -1.8);
\draw [blue, ->-=0.5] (0.3, -1)--(0.3, -1.8);
\path[fill=gray!50!white]
(-1, -1.2) arc (180:0:1)--(1, 1.2) arc (0:-180:1)--(-1, -1.2);
\draw (-1, -1.2) arc (180:0:1);
\draw (1, 1.2) arc (0:-180:1);
\draw [fill=white](-0.4, 0.5) rectangle (0.4, 0.8); \node at (0, 0.65) {\tiny{$\phi^{\dagger}$}};
\draw [fill=white](-0.4, -1.2) rectangle (0.4, -0.9); \node at (0, -1.05) {\tiny{$\phi^{\dagger\dagger}$}};
\end{tikzpicture}}}
\ar[equal]{d}\\
{}&\vcenter{\hbox{
\begin{tikzpicture}
\draw[dashed] (0, 1.5) [partial ellipse=0:360:1 and 0.3];
\draw[dashed] (0, 0) [partial ellipse=0:360:1 and 0.3];
\draw[dashed] (0, -1.5) [partial ellipse=0:360:1 and 0.3];
\draw (0.5, 0) node{\tiny{$D$}};
\draw (-1, -1.5)--(-1, 1.5);
\draw (1, -1.5)--(1, 1.5);
\path[fill=gray!50!white]
(-1, -1.5) rectangle (-2, 1.5);
\path[fill=gray!50!white]
(1, -1.5) rectangle (2, 1.5);
\draw [blue, -<-=0.7] (0, 1.2)--(0, -1.8);
\draw [blue, -<-=0.7] (-0.3, 1.2)--(-0.3, -1.8);
\draw [blue, -<-=0.7] (0.3, 1.2)--(0.3, -1.8);
\end{tikzpicture}}}\arrow[r, "\text{Move $2$}"]&\vcenter{\hbox{
\begin{tikzpicture}
\draw[dashed] (0, 1.5) [partial ellipse=0:360:1 and 0.3];
\draw[dashed] (0, -1.5) [partial ellipse=0:360:1 and 0.3];
\draw (-1, -1.5)--(-1, -1.2);
\draw (1, 1.5)--(1, 1.2);
\draw (-1, 1.5)--(-1, 1.2);
\draw (1, -1.5)--(1, -1.2);
\path[fill=gray!50!white]
(-1, -1.5) rectangle (-2, 1.5);
\path[fill=gray!50!white]
(1, -1.5) rectangle (2, 1.5);
\draw [blue, -<-=0.5] (0, 1.2)--(0, 0.7);
\draw [blue, -<-=0.5] (-0.3, 1.2)--(-0.3, 0.7);
\draw [blue, -<-=0.5] (0.3, 1.2)--(0.3, 0.7);
\draw [blue, -<-=0.5] (0, -1)--(0, -1.8);
\draw [blue, -<-=0.5] (-0.3, -1)--(-0.3, -1.8);
\draw [blue, -<-=0.5] (0.3, -1)--(0.3, -1.8);
\path[fill=gray!50!white]
(-1, -1.2) arc (180:0:1)--(1, 1.2) arc (0:-180:1)--(-1, -1.2);
\draw (-1, -1.2) arc (180:0:1);
\draw (1, 1.2) arc (0:-180:1);
\draw [fill=white](-0.4, 0.5) rectangle (0.4, 0.8); \node at (0, 0.65) {\tiny{$\phi^{\dagger}$}};
\draw [fill=white](-0.4, -1.2) rectangle (0.4, -0.9); \node at (0, -1.05) {\tiny{$\phi$}};
\end{tikzpicture}}}
\end{tikzcd}
$$

To show $\V$ defines a unitary TQFT with complex conjugation $\dagger$, we only need to show vector space $\V(\mathbb{F})$ equipped with the sesquilinear form 
\begin{align*}
\langle \cdot, \cdot\rangle: \V(\mathbb{F})\times \V(\mathbb{F})&\rightarrow \mathbb{C}\\
(v, w)&\mapsto w^{\dagger} \circ v
\end{align*}
is a Hilbert space.

Suppose $R_{A}(\mathbb{F})\ne \emptyset$, we can choose a truncated triangulation $\mathcal{P}$ of $\mathbb{F}$ and a basis $B_{\mathcal{P}}$ as in Equation \eqref{eq:basis}, we further require 
$$B_{\Delta_{i}^{c}}\subset \hom(\1, X(\Delta_i^{c}))$$
is an orthonormal basis for all $c$ and $i$. By Theorem \ref{thm:tvbasis}, 
$$[\Phi(v_1^{c}\otimes\cdots\otimes v_m^{c}), \Psi(({w_1^{c}})^{\dagger}\otimes\cdots\otimes({w_m^{c}})^{\dagger})]=\prod_{i} ({w_{i}^{c}})^{\dagger}\circ v_i^{c},$$
which equals $1$ if $v_i^{c}=w_{i}^c$ for all $i$ and $c$ and equals $0$ otherwise. Thus $\V(\mathbb{F})$ is a Hilbert space.

If $R_{A}(\mathbb{F})=\emptyset$, by Lemma \ref{lem:surjectivity}, there exist a surjective map from $\V(\mathbb{F'}):=\V(F, c, \emptyset)$ to $\V(\mathbb{F})=\V(F, \emptyset_B, \emptyset)$. By an identical argument  as in the proof of Lemma \ref{lem:surjectivity} with $\mathbb{M'}$ replace by $\mathbb{M}^{\dagger}$, one can show this is an orthogonal projection. Therefore, $\V(\mathbb{F})=\V(F, \emptyset_B, \emptyset)$ is a Hilbert subspace of $\V(\mathbb{F'})=\V(F, c, \emptyset)$.
\end{proof}

\begin{remark}
In the Levin-Wen lattice model, the basis of ground states are given by the simple objects of the Drinfeld center of a spherical fusion category.
The $3$-alterfold TQFT provide an explanation to this relation.

We associate an admissible state of the Levin-Wen lattice model associated to a spherical fusion category as the following $3$-alterfold with time boundary, it is a collar representative of a surface with $A$-colored regions are disjoint union of disks.
$$
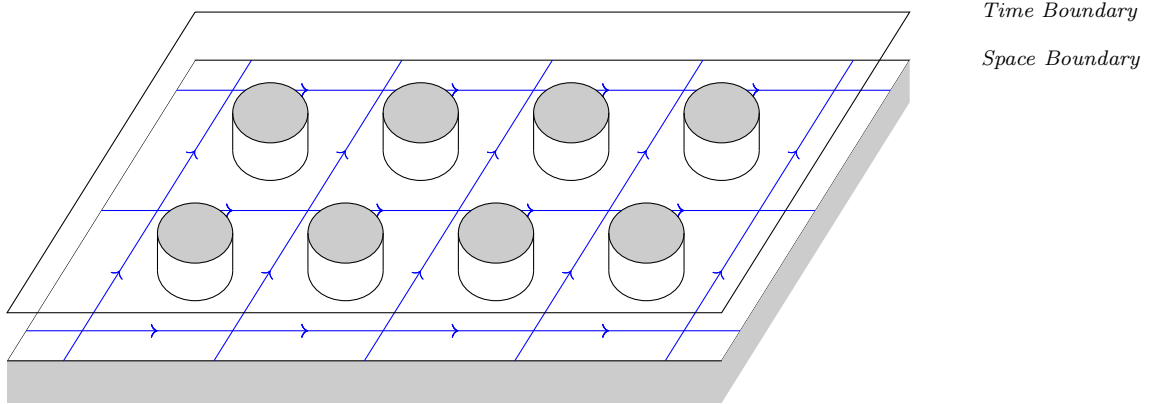
\vcenter{\hbox{
\begin{tikzpicture}
\begin{scope}[yscale=0.8]
\path [fill=black!20!white] (-7, -0.5)--(2.5, -0.5)--(2.5, -1.2)--(-7,-1.2);
\draw (-7, -0.5)--(-4.5, 4.5)--(5, 4.5) --(2.5, -0.5)--(-7, -0.5);
\node at (5.8, 4.5) [right] {\tiny Space Boundary};
\node at (5.8, 5.3) [right] {\tiny Time Boundary};
\draw[blue, ->-=0.2, ->-=0.4, ->-=0.6, ->-=0.8] (-7, 0)--(3, 0);
\draw[blue, ->-=0.2, ->-=0.4, ->-=0.6, ->-=0.8] (-6, 2)--(4, 2);
\draw[blue, ->-=0.2, ->-=0.4, ->-=0.6, ->-=0.8] (-5, 4)--(5, 4);
\draw[blue, ->-=0.3, ->-=0.7] (-6.25, -0.5)--(-3.75, 4.5);
\draw[blue, ->-=0.3, ->-=0.7] (-4.25, -0.5)--(-1.75, 4.5);
\draw[blue, ->-=0.3, ->-=0.7] (-2.25, -0.5)--(0.25, 4.5);
\draw[blue, ->-=0.3, ->-=0.7] (-0.25, -0.5)--(2.25, 4.5);
\draw[blue, ->-=0.3, ->-=0.7] (1.75, -0.5)--(4.25, 4.5);
\draw (-4.5, 1) node(A){};
\draw (-2.5, 1) node(B){};
\draw (-0.5, 1) node(C){};
\draw (1.5, 1) node(D){};
\draw (-3.5, 3) node(E){};
\draw (-1.5, 3) node(F){};
\draw (0.5, 3) node(G){};
\draw (2.5, 3) node(H){};
\path[fill=white] (-7, -0.5)--(-4.5, 4.5)--(-5.5, 4.5)--(-8, -0.5);
\path[fill=white] (2.5, -0.5)--(5, 4.5)--(6, 4.5)--(3.5, -0.5);
\path [fill=black!20!white] (2.5, -0.5)--(5, 4.5)--(5, 3.8)--(2.5, -1.2);
\draw (-4.5, 5.3) -- (5,5.3)--(2.5, 0.3)--(-7, 0.3)--(-4.5, 5.3);
\end{scope}
\begin{scope}[shift={(-4.5, 1)}]
\draw[fill=black!20!white] (0, 0.3) [partial ellipse=0:360:0.5 and 0.4];
\draw (-0.5, 0.3)--(-0.5, -0.2);
\draw (0.5, 0.3)--(0.5, -0.2);
\draw (0, -0.2) [partial ellipse=0:-180:0.5 and 0.4];
\end{scope}
\begin{scope}[shift={(-2.5, 1)}]
\draw[fill=black!20!white] (0, 0.3) [partial ellipse=0:360:0.5 and 0.4];
\draw (-0.5, 0.3)--(-0.5, -0.2);
\draw (0.5, 0.3)--(0.5, -0.2);
\draw (0, -0.2) [partial ellipse=0:-180:0.5 and 0.4];
\end{scope}
\begin{scope}[shift={(-0.5, 1)}]
\draw[fill=black!20!white] (0, 0.3) [partial ellipse=0:360:0.5 and 0.4];
\draw (-0.5, 0.3)--(-0.5, -0.2);
\draw (0.5, 0.3)--(0.5, -0.2);
\draw (0, -0.2) [partial ellipse=0:-180:0.5 and 0.4];
\end{scope}
\begin{scope}[shift={(1.5, 1)}]
\draw[fill=black!20!white] (0, 0.3) [partial ellipse=0:360:0.5 and 0.4];
\draw (-0.5, 0.3)--(-0.5, -0.2);
\draw (0.5, 0.3)--(0.5, -0.2);
\draw (0, -0.2) [partial ellipse=0:-180:0.5 and 0.4];
\end{scope}
\begin{scope}[shift={(-3.5, 2.6)}]
\draw[fill=black!20!white] (0, 0.3) [partial ellipse=0:360:0.5 and 0.4];
\draw (-0.5, 0.3)--(-0.5, -0.2);
\draw (0.5, 0.3)--(0.5, -0.2);
\draw (0, -0.2) [partial ellipse=0:-180:0.5 and 0.4];
\end{scope}
\begin{scope}[shift={(-1.5, 2.6)}]
\draw[fill=black!20!white] (0, 0.3) [partial ellipse=0:360:0.5 and 0.4];
\draw (-0.5, 0.3)--(-0.5, -0.2);
\draw (0.5, 0.3)--(0.5, -0.2);
\draw (0, -0.2) [partial ellipse=0:-180:0.5 and 0.4];
\end{scope}
\begin{scope}[shift={(0.5, 2.6)}]
\draw[fill=black!20!white] (0, 0.3) [partial ellipse=0:360:0.5 and 0.4];
\draw (-0.5, 0.3)--(-0.5, -0.2);
\draw (0.5, 0.3)--(0.5, -0.2);
\draw (0, -0.2) [partial ellipse=0:-180:0.5 and 0.4];
\end{scope}
\begin{scope}[shift={(2.5, 2.6)}]
\draw[fill=black!20!white] (0, 0.3) [partial ellipse=0:360:0.5 and 0.4];
\draw (-0.5, 0.3)--(-0.5, -0.2);
\draw (0.5, 0.3)--(0.5, -0.2);
\draw (0, -0.2) [partial ellipse=0:-180:0.5 and 0.4];
\end{scope}
\end{tikzpicture}
}}$$
For each blue edge, we label it by a simple object of $\mathcal{C}$.
For each degree-$4$ vertex, we label it by a morphism in the invariant space of the tensor product of the objects over the adjacent edges.

A plaquette operator is given by the $\Omega$-color circle acting around the $A$-color disc on the time boundary.   
The Hamiltonian is given by the sum of all plaquette operators. 
As plaquette operators are mutually commuting projections, the ground state of the Hamiltonian can be expressed as the common image of the plaquette operator, which equals to the subspace generated by the following diagram:
$$
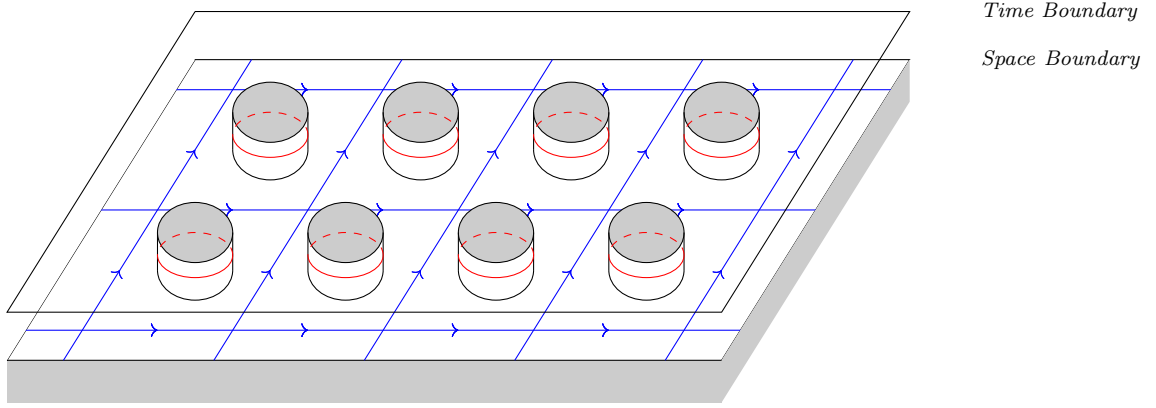
\vcenter{\hbox{
\begin{tikzpicture}
\begin{scope}[yscale=0.8]
\path [fill=black!20!white] (-7, -0.5)--(2.5, -0.5)--(2.5, -1.2)--(-7,-1.2);
\draw (-7, -0.5)--(-4.5, 4.5)--(5, 4.5) --(2.5, -0.5)--(-7, -0.5);
\node at (5.8, 4.5) [right] {\tiny Space Boundary};
\node at (5.8, 5.3) [right] {\tiny Time Boundary};
\draw[blue, ->-=0.2, ->-=0.4, ->-=0.6, ->-=0.8] (-7, 0)--(3, 0);
\draw[blue, ->-=0.2, ->-=0.4, ->-=0.6, ->-=0.8] (-6, 2)--(4, 2);
\draw[blue, ->-=0.2, ->-=0.4, ->-=0.6, ->-=0.8] (-5, 4)--(5, 4);
\draw[blue, ->-=0.3, ->-=0.7] (-6.25, -0.5)--(-3.75, 4.5);
\draw[blue, ->-=0.3, ->-=0.7] (-4.25, -0.5)--(-1.75, 4.5);
\draw[blue, ->-=0.3, ->-=0.7] (-2.25, -0.5)--(0.25, 4.5);
\draw[blue, ->-=0.3, ->-=0.7] (-0.25, -0.5)--(2.25, 4.5);
\draw[blue, ->-=0.3, ->-=0.7] (1.75, -0.5)--(4.25, 4.5);
\draw (-4.5, 1) node(A){};
\draw (-2.5, 1) node(B){};
\draw (-0.5, 1) node(C){};
\draw (1.5, 1) node(D){};
\draw (-3.5, 3) node(E){};
\draw (-1.5, 3) node(F){};
\draw (0.5, 3) node(G){};
\draw (2.5, 3) node(H){};
\path[fill=white] (-7, -0.5)--(-4.5, 4.5)--(-5.5, 4.5)--(-8, -0.5);
\path[fill=white] (2.5, -0.5)--(5, 4.5)--(6, 4.5)--(3.5, -0.5);
\path [fill=black!20!white] (2.5, -0.5)--(5, 4.5)--(5, 3.8)--(2.5, -1.2);
\draw (-4.5, 5.3) -- (5,5.3)--(2.5, 0.3)--(-7, 0.3)--(-4.5, 5.3);
\end{scope}
\begin{scope}[shift={(-4.5, 1)}]
\draw[fill=black!20!white] (0, 0.3) [partial ellipse=0:360:0.5 and 0.4];
\draw (-0.5, 0.3)--(-0.5, -0.2);
\draw (0.5, 0.3)--(0.5, -0.2);
\draw [red] (0, 0) [partial ellipse=0:-180:0.5 and 0.3];
\draw [red,dashed] (0, 0) [partial ellipse=0:180:0.5 and 0.3];
\draw (0, -0.2) [partial ellipse=0:-180:0.5 and 0.4];
\end{scope}
\begin{scope}[shift={(-2.5, 1)}]
\draw[fill=black!20!white] (0, 0.3) [partial ellipse=0:360:0.5 and 0.4];
\draw (-0.5, 0.3)--(-0.5, -0.2);
\draw (0.5, 0.3)--(0.5, -0.2);
\draw [red] (0, 0) [partial ellipse=0:-180:0.5 and 0.3];
\draw [red,dashed] (0, 0) [partial ellipse=0:180:0.5 and 0.3];
\draw (0, -0.2) [partial ellipse=0:-180:0.5 and 0.4];
\end{scope}
\begin{scope}[shift={(-0.5, 1)}]
\draw[fill=black!20!white] (0, 0.3) [partial ellipse=0:360:0.5 and 0.4];
\draw (-0.5, 0.3)--(-0.5, -0.2);
\draw (0.5, 0.3)--(0.5, -0.2);
\draw [red] (0, 0) [partial ellipse=0:-180:0.5 and 0.3];
\draw [red,dashed] (0, 0) [partial ellipse=0:180:0.5 and 0.3];
\draw (0, -0.2) [partial ellipse=0:-180:0.5 and 0.4];
\end{scope}
\begin{scope}[shift={(1.5, 1)}]
\draw[fill=black!20!white] (0, 0.3) [partial ellipse=0:360:0.5 and 0.4];
\draw (-0.5, 0.3)--(-0.5, -0.2);
\draw (0.5, 0.3)--(0.5, -0.2);
\draw [red] (0, 0) [partial ellipse=0:-180:0.5 and 0.3];
\draw [red,dashed] (0, 0) [partial ellipse=0:180:0.5 and 0.3];
\draw (0, -0.2) [partial ellipse=0:-180:0.5 and 0.4];
\end{scope}
\begin{scope}[shift={(-3.5, 2.6)}]
\draw[fill=black!20!white] (0, 0.3) [partial ellipse=0:360:0.5 and 0.4];
\draw (-0.5, 0.3)--(-0.5, -0.2);
\draw (0.5, 0.3)--(0.5, -0.2);
\draw [red] (0, 0) [partial ellipse=0:-180:0.5 and 0.3];
\draw [red,dashed] (0, 0) [partial ellipse=0:180:0.5 and 0.3];
\draw (0, -0.2) [partial ellipse=0:-180:0.5 and 0.4];
\end{scope}
\begin{scope}[shift={(-1.5, 2.6)}]
\draw[fill=black!20!white] (0, 0.3) [partial ellipse=0:360:0.5 and 0.4];
\draw (-0.5, 0.3)--(-0.5, -0.2);
\draw (0.5, 0.3)--(0.5, -0.2);
\draw [red] (0, 0) [partial ellipse=0:-180:0.5 and 0.3];
\draw [red,dashed] (0, 0) [partial ellipse=0:180:0.5 and 0.3];
\draw (0, -0.2) [partial ellipse=0:-180:0.5 and 0.4];
\end{scope}
\begin{scope}[shift={(0.5, 2.6)}]
\draw[fill=black!20!white] (0, 0.3) [partial ellipse=0:360:0.5 and 0.4];
\draw (-0.5, 0.3)--(-0.5, -0.2);
\draw (0.5, 0.3)--(0.5, -0.2);
\draw [red] (0, 0) [partial ellipse=0:-180:0.5 and 0.3];
\draw [red,dashed] (0, 0) [partial ellipse=0:180:0.5 and 0.3];
\draw (0, -0.2) [partial ellipse=0:-180:0.5 and 0.4];
\end{scope}
\begin{scope}[shift={(2.5, 2.6)}]
\draw[fill=black!20!white] (0, 0.3) [partial ellipse=0:360:0.5 and 0.4];
\draw (-0.5, 0.3)--(-0.5, -0.2);
\draw (0.5, 0.3)--(0.5, -0.2);
\draw [red] (0, 0) [partial ellipse=0:-180:0.5 and 0.3];
\draw [red,dashed] (0, 0) [partial ellipse=0:180:0.5 and 0.3];
\draw (0, -0.2) [partial ellipse=0:-180:0.5 and 0.4];
\end{scope}
\end{tikzpicture}
}}$$

By moving each red curves over the $A$-colored vertical pillars and applying Move $1$ to cut the pillars, the ground states are exactly the state lying in the subspace $\V(F, \emptyset_B, \emptyset)$ whose time boundary is a $B$-color surface $F$.
In particular, if $F$ is homeomorphic to a torus, the vector space $\V(F, \emptyset_B, \emptyset)$ equals to the Grothendieck group of Drinfeld center $\mathcal{Z(C)}$.

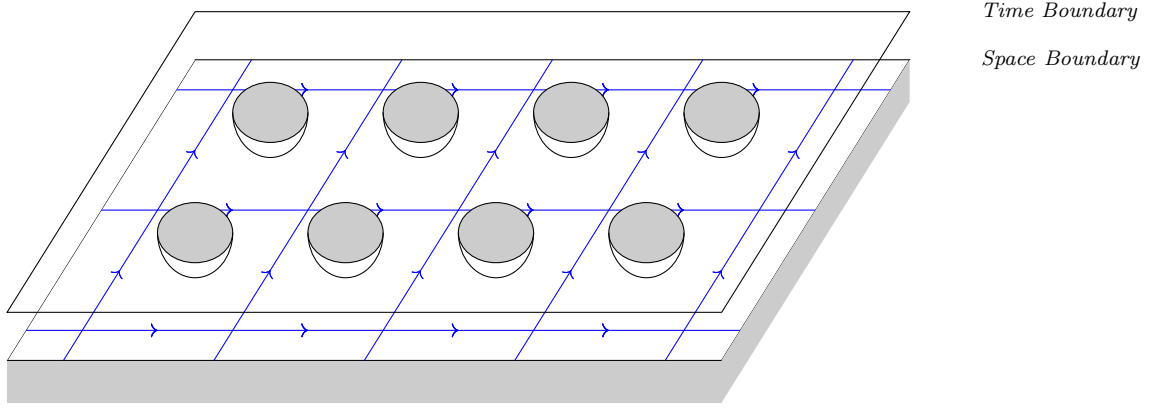
$$
\vcenter{\hbox{
\begin{tikzpicture}
\begin{scope}[yscale=0.8]
\path [fill=black!20!white] (-7, -0.5)--(2.5, -0.5)--(2.5, -1.2)--(-7,-1.2);
\draw (-7, -0.5)--(-4.5, 4.5)--(5, 4.5) --(2.5, -0.5)--(-7, -0.5);
\node at (5.8, 4.5) [right] {\tiny Space Boundary};
\node at (5.8, 5.3) [right] {\tiny Time Boundary};
\draw[blue, ->-=0.2, ->-=0.4, ->-=0.6, ->-=0.8] (-7, 0)--(3, 0);
\draw[blue, ->-=0.2, ->-=0.4, ->-=0.6, ->-=0.8] (-6, 2)--(4, 2);
\draw[blue, ->-=0.2, ->-=0.4, ->-=0.6, ->-=0.8] (-5, 4)--(5, 4);
\draw[blue, ->-=0.3, ->-=0.7] (-6.25, -0.5)--(-3.75, 4.5);
\draw[blue, ->-=0.3, ->-=0.7] (-4.25, -0.5)--(-1.75, 4.5);
\draw[blue, ->-=0.3, ->-=0.7] (-2.25, -0.5)--(0.25, 4.5);
\draw[blue, ->-=0.3, ->-=0.7] (-0.25, -0.5)--(2.25, 4.5);
\draw[blue, ->-=0.3, ->-=0.7] (1.75, -0.5)--(4.25, 4.5);
\draw (-4.5, 1) node(A){};
\draw (-2.5, 1) node(B){};
\draw (-0.5, 1) node(C){};
\draw (1.5, 1) node(D){};
\draw (-3.5, 3) node(E){};
\draw (-1.5, 3) node(F){};
\draw (0.5, 3) node(G){};
\draw (2.5, 3) node(H){};
\path[fill=white] (-7, -0.5)--(-4.5, 4.5)--(-5.5, 4.5)--(-8, -0.5);
\path[fill=white] (2.5, -0.5)--(5, 4.5)--(6, 4.5)--(3.5, -0.5);
\path [fill=black!20!white] (2.5, -0.5)--(5, 4.5)--(5, 3.8)--(2.5, -1.2);
\draw (-4.5, 5.3) -- (5,5.3)--(2.5, 0.3)--(-7, 0.3)--(-4.5, 5.3);
\end{scope}
\begin{scope}[shift={(-4.5, 1)}]
\draw[fill=black!20!white] (0, 0.3) [partial ellipse=0:360:0.5 and 0.4];
\draw (0, 0.3) [partial ellipse=0:-180:0.5 and 0.6];
\end{scope}
\begin{scope}[shift={(-2.5, 1)}]
\draw[fill=black!20!white] (0, 0.3) [partial ellipse=0:360:0.5 and 0.4];
\draw (0, 0.3) [partial ellipse=0:-180:0.5 and 0.6];
\end{scope}
\begin{scope}[shift={(-0.5, 1)}]
\draw[fill=black!20!white] (0, 0.3) [partial ellipse=0:360:0.5 and 0.4];
\draw (0, 0.3) [partial ellipse=0:-180:0.5 and 0.6];
\end{scope}
\begin{scope}[shift={(1.5, 1)}]
\draw[fill=black!20!white] (0, 0.3) [partial ellipse=0:360:0.5 and 0.4];
\draw (0, 0.3) [partial ellipse=0:-180:0.5 and 0.6];
\end{scope}
\begin{scope}[shift={(-3.5, 2.6)}]
\draw[fill=black!20!white] (0, 0.3) [partial ellipse=0:360:0.5 and 0.4];
\draw (0, 0.3) [partial ellipse=0:-180:0.5 and 0.6];
\end{scope}
\begin{scope}[shift={(-1.5, 2.6)}]
\draw[fill=black!20!white] (0, 0.3) [partial ellipse=0:360:0.5 and 0.4];
\draw (0, 0.3) [partial ellipse=0:-180:0.5 and 0.6];
\end{scope}
\begin{scope}[shift={(0.5, 2.6)}]
\draw[fill=black!20!white] (0, 0.3) [partial ellipse=0:360:0.5 and 0.4];
\draw (0, 0.3) [partial ellipse=0:-180:0.5 and 0.6];
\end{scope}
\begin{scope}[shift={(2.5, 2.6)}]
\draw[fill=black!20!white] (0, 0.3) [partial ellipse=0:360:0.5 and 0.4];
\draw (0, 0.3) [partial ellipse=0:-180:0.5 and 0.6];
\end{scope}
\end{tikzpicture}
}}$$
\end{remark}

We leave the proof of the following corollary to readers.

\begin{corollary}
Suppose $\mathcal{C}$ is a unitary fusion category.
Then $\cA$ is a unitary modular tensor category.
In particular, the Drinfeld center $\mathcal{Z(C)}$ is a unitary modular tensor category. \qed
\end{corollary}

Recall that the mapping class group $\MCG(F)$ of a compact closed oriented surface $F$, consists of isotopy classes of orientation preserving homeomorphisms, and that a TQFT gives rise to representations of the mapping class groups of surfaces. 
Another consequence of the unitarity of $\V$ is given as follows.

\begin{corollary}
Suppose $\mathcal{C}$ is a unitary fusion category. 
The action of the mapping class group $\MCG(R_{B}(F))$ on $\V(\mathbb{F})$ is unitary.
\end{corollary}
\begin{proof}
For all $\mathbb{M}\in \hom_{\Cob}(\1, \mathbb{F})$ a $\mathcal{C}$-decorated alterfold cobordism, and $C_f$ a mapping cylinder associated to the mapping class $[f]\in \MCG(R_{B}(F))$. We compute:
\begin{align*}
\langle \V(C_{f})\V(\mathbb{M}), \V(C_{f})\V(\mathbb{M})\rangle&=\langle \V(\mathbb{M}\sqcup_{\mathbb{F}}C_{f}), \V(\mathbb{M}\sqcup_{\mathbb{F}}C_{f})\rangle\\
&=\V(\mathbb{M}^{\dagger}\sqcup_{\mathbb{F}}C_{f}^{\dagger}\sqcup_{\mathbb{F}}C_{f} \sqcup_{\mathbb{F}} \mathbb{M})\\
&=\V(\mathbb{M}^{\dagger}\sqcup_{\mathbb{F}}\mathbb{M})\\
&=\langle \V(\mathbb{M}), \V(\mathbb{M})\rangle
\end{align*}
This show that $\V([C_{f}])$ is unitary. Here, the third equality follows from the commutative diagram
\[\begin{tikzcd}
&{C_{f}^{\dagger}}\ar[dd, "f^{-1}\times \id"]&\\
F\times \{0\}\ar[ru, "f"]\ar[dr, "\id"'] & & F\times\{1\} \ar[ul, "\id"'] \ar[dl, "f^{-1}"]\\
&{C_{f^{-1}}}&
\end{tikzcd}\,,\]
which implies that $C_{f}^{\dagger}=(\mathbb{F}\times I, f: -\mathbb{F}\rightarrow \mathbb{F}, \id:\mathbb{F}\rightarrow \mathbb{F})$ is equivalent to $(\mathbb{F}\times I, \id: -\mathbb{F}\rightarrow \mathbb{F}, f^{-1}:\mathbb{F}\rightarrow \mathbb{F})$ as cobordisms via homeomorphism $f^{-1}\times I$.
\end{proof}

In \cite{LPW21, HLPW23}, Huang, the first author, Palcoux and the fourth author discovered efficient criteria for unitary categorification.
This is an interesting direction in quantum Fourier analysis \cite{JJLRW20}.
Here we present a 3D interpretation of the primary criteria for unitary categorification, which leads fruitful quantum inequalities and we will discuss them in the future.
\begin{proposition}\label{prop:positivitycriterion}
Suppose $\mathcal{C}$ is a unitary fusion category.
Then for all $n\in\bN$, we have
    \begin{align*}
\sum_{j=0}^r d_j^{-n+2} N_j^{\otimes n} \geq 0,
\end{align*}
where $N_{j}$ is the fusion matrix for the simple object $x_j$ and $d_j$ is the quantum dimension of $x_j$.
\end{proposition}
\begin{proof}
We only present a pictorial proof for $n=3$.
By the dagger operation and Lemma 4.21 in \cite{LMWW23}, we have that
    $$0\leq 
\vcenter{\hbox{\scalebox{0.7}{\begin{tikzpicture}
\begin{scope}
\draw[line width=10pt] (-1, 0) [partial ellipse=-90:90:0.6 and 2];
\draw[line width=10pt] (1, 0) [partial ellipse=-90:90:0.6 and 2];
\draw[line width=9pt, \acolor] (-1, 0) [partial ellipse=-90:90:0.6 and 2];
\draw[line width=9pt, \acolor] (1, 0) [partial ellipse=-90:90:0.6 and 2];
\draw[red] (-1, 0) [partial ellipse=-90:90:0.6 and 2];
\draw[red] (1, 0) [partial ellipse=-90:90:0.6 and 2];
\end{scope}
\begin{scope}[yshift=0.3cm]
\draw[line width=16pt, white] (-2, -2.5)..controls (-2, -1) and (-1.2, -1) .. (0, -1);
\draw[line width=16pt, white] (2, -2.5)..controls (2, -1) and (1.2, -1) .. (0, -1);
\draw[line width=16pt, white] (0, -2.5)--(0, -1.175);
\draw[line width=11pt, ] (-2, -2.5)..controls (-2, -1) and (-1.2, -1) .. (0, -1);
\draw[line width=11pt, ] (2, -2.5)..controls (2, -1) and (1.2, -1) .. (0, -1);
\draw[line width=11pt, ] (0, -2.5)--(0, -1.175);
\draw[line width=10pt, \acolor] (-2, -2.5)..controls (-2, -1) and (-1.2, -1) .. (0, -1);
\draw[line width=10pt, \acolor] (2, -2.5)..controls (2, -1) and (1.2, -1) .. (0, -1);
\filldraw[line width=10pt, \acolor] (0, -2.5)--(0, -1.175);
\path[fill=\acolor] (-2, -2.5) [partial ellipse=0:360:0.18 and 0.1];
\path[fill=\acolor] (0, -2.5) [partial ellipse=0:360:0.18 and 0.1];
\path[fill=\acolor] (2, -2.5) [partial ellipse=0:360:0.18 and 0.1];
\draw (-2, -2.5) [partial ellipse=180:360:0.18 and 0.1];
\draw (0, -2.5) [partial ellipse=180:360:0.18 and 0.1];
\draw (2, -2.5) [partial ellipse=180:360:0.18 and 0.1];
\draw[dashed] (-2, -2.5) [partial ellipse=0:180:0.18 and 0.1];
\draw[dashed](0, -2.5) [partial ellipse=0:180:0.18 and 0.1];
\draw[dashed](2, -2.5) [partial ellipse=0:180:0.18 and 0.1];
\end{scope}
\begin{scope}[yscale=-1, yshift=0.3cm]
\draw[line width=16pt, white] (-2, -2.5)..controls (-2, -1) and (-1.2, -1) .. (0, -1);
\draw[line width=16pt, white] (2, -2.5)..controls (2, -1) and (1.2, -1) .. (0, -1);
\draw[line width=16pt, white] (0, -2.5)--(0, -1.175);
\draw[line width=11pt, ] (-2, -2.5)..controls (-2, -1) and (-1.2, -1) .. (0, -1);
\draw[line width=11pt, ] (2, -2.5)..controls (2, -1) and (1.2, -1) .. (0, -1);
\draw[line width=11pt, ] (0, -2.5)--(0, -1.175);
\draw[line width=10pt, \acolor] (-2, -2.5)..controls (-2, -1) and (-1.2, -1) .. (0, -1);
\draw[line width=10pt, \acolor] (2, -2.5)..controls (2, -1) and (1.2, -1) .. (0, -1);
\filldraw[line width=10pt, \acolor] (0, -2.5)--(0, -1.175);
\path[fill=\acolor] (-2, -2.5) [partial ellipse=0:360:0.18 and 0.1];
\path[fill=\acolor] (0, -2.5) [partial ellipse=0:360:0.18 and 0.1];
\path[fill=\acolor] (2, -2.5) [partial ellipse=0:360:0.18 and 0.1];
\draw (-2, -2.5) [partial ellipse=0:360:0.18 and 0.1];
\draw (0, -2.5) [partial ellipse=0:360:0.18 and 0.1];
\draw (2, -2.5) [partial ellipse=0:360:0.18 and 0.1];
\end{scope}
\begin{scope}[xscale=-1]
\draw[line width=16pt, white] (-1, 0) [partial ellipse=-90:90:0.6 and 2];
\draw[line width=16pt, white] (1, 0) [partial ellipse=-90:90:0.6 and 2];
\draw[line width=10pt] (-1, 0) [partial ellipse=-90:90:0.6 and 2];
\draw[line width=10pt] (1, 0) [partial ellipse=-90:90:0.6 and 2];
\draw[line width=9pt, \acolor] (-1, 0) [partial ellipse=-90:90:0.6 and 2];
\draw[line width=9pt, \acolor] (1, 0) [partial ellipse=-90:90:0.6 and 2];
\draw[red] (-1, 0) [partial ellipse=-90:90:0.6 and 2];
\draw[red] (1, 0) [partial ellipse=-90:90:0.6 and 2];
\end{scope}
\end{tikzpicture}}}}
=\sum_{j} d_j^{-1}\vcenter{\hbox{\scalebox{0.7}{
\begin{tikzpicture}
\draw[line width=11pt] (-2, -2.2)--(-2, 2.2);
\draw[line width=10pt, \acolor] (-2, -2.2)--(-2, 2.2);
\path[fill=\acolor] (-2, -2.2) [partial ellipse=0:360:0.18 and 0.1];
\path[fill=\acolor] (-2, 2.2) [partial ellipse=0:360:0.18 and 0.1];
\draw[dashed] (-2, -2.2) [partial ellipse=0:180:0.18 and 0.1];
\draw (-2, -2.2) [partial ellipse=180:360:0.18 and 0.1];
\draw (-2, 2.2) [partial ellipse=0:360:0.18 and 0.1];
\draw[blue, dashed] (-2, 0) [partial ellipse=0:180:0.18 and 0.1];
\draw[blue] (-2, 0) [partial ellipse=180:360:0.18 and 0.1];
\draw (-2.15, 0) node[left]{\tiny$X_{j}$};
\begin{scope}[xshift=1.5cm]
\draw[line width=11pt] (-2, -2.2)--(-2, 2.2);
\draw[line width=10pt, \acolor] (-2, -2.2)--(-2, 2.2);
\path[fill=\acolor] (-2, -2.2) [partial ellipse=0:360:0.18 and 0.1];
\path[fill=\acolor] (-2, 2.2) [partial ellipse=0:360:0.18 and 0.1];
\draw[dashed] (-2, -2.2) [partial ellipse=0:180:0.18 and 0.1];
\draw (-2, -2.2) [partial ellipse=180:360:0.18 and 0.1];
\draw (-2, 2.2) [partial ellipse=0:360:0.18 and 0.1];
\draw[blue, dashed] (-2, 0) [partial ellipse=0:180:0.18 and 0.1];
\draw[blue] (-2, 0) [partial ellipse=180:360:0.18 and 0.1];
\draw (-2.15, 0) node[left]{\tiny$X_{j}$};
\end{scope}
\begin{scope}[xshift=3cm]
\draw[line width=11pt] (-2, -2.2)--(-2, 2.2);
\draw[line width=10pt, \acolor] (-2, -2.2)--(-2, 2.2);
\path[fill=\acolor] (-2, -2.2) [partial ellipse=0:360:0.18 and 0.1];
\path[fill=\acolor] (-2, 2.2) [partial ellipse=0:360:0.18 and 0.1];
\draw[dashed] (-2, -2.2) [partial ellipse=0:180:0.18 and 0.1];
\draw (-2, -2.2) [partial ellipse=180:360:0.18 and 0.1];
\draw (-2, 2.2) [partial ellipse=0:360:0.18 and 0.1];
\draw[blue, dashed] (-2, 0) [partial ellipse=0:180:0.18 and 0.1];
\draw[blue] (-2, 0) [partial ellipse=180:360:0.18 and 0.1];
\draw (-2.15, 0) node[left]{\tiny$X_{j}$};
\end{scope}
\end{tikzpicture}
}}}
$$
i.e.
\begin{align*}
\sum_{j=0}^r d_j^{-1} N_j^{\otimes 3} \geq 0.
\end{align*}
This completes the proof of the proposition.
\end{proof}

\subsection{TQFT Basis via Pants Decomposition}
In this subsection, we give explicit description of a choice of bases of $\V(F, \emptyset_B, \emptyset)$ using pants decomposition of $F$.


Let $\bF = (F, \emptyset_B, \emptyset)$ be a connected 2-alterfold such that $F$ is of genus $g$. 
We define a pants decomposition $\mathbf{C}$ of $\bF$ to be the topological pants decomposition of $F$ \cite[Section 8.3.1]{FarMar12}.
Namely, if $g \ge 2$, we define a pants decomposition of $\bF$ to be a set of $3(g-1)$ disjoint circles $\{C_i\}$ on $F$ such that (the closure of) the complement of $\bigcup_i C_i$ in $F$ is homeomorphic to a disjoint union of 3-holed spheres, or pairs of pants. 
To complete the definition, we define a pants decomposition of $\bS^2$ to be $\emptyset$, and that of $\bT^2$ to be $\{C\}$ where $C$ is a meridian circle that meets $\gamma$ transversely while missing $P$.

Assume for now that $F$ is connected of genus $\geq 2$. 
A pants decomposition $\{C_i\}$ of $ (F, \gamma)$ determines a unique $3$-alterfold $\mathbb{M}=(M, \Sigma)$
that is a collar representative such that each $C_i$ bounds a disk $D_i$ in $\mathbb{M}$. Moreover, $D_i$ cuts $\M$ into pieces homeomorphic to the following:
$$
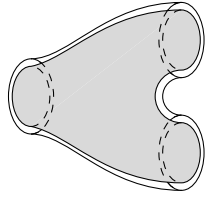
\vcenter{\hbox{
\begin{tikzpicture}
\draw (0, 0.75) [partial ellipse=90:270: 0.25 and 0.25];
\draw (0, 2)..controls +(-1, 0) and (-1.5, 1.25)..(-2, 1.25);
\draw (0, -0.5)..controls +(-1, 0) and (-1.5, 0.25)..(-2, 0.25);
\draw [fill=white] (0, 0) [partial ellipse=0:360:0.3 and 0.5];
\draw [fill=white] (0, 1.5) [partial ellipse=0:360:0.3 and 0.5];
\begin{scope}[xshift=-2cm, yshift=0.75cm]
\draw [fill=white] (0, 0) [partial ellipse=0:360:0.3 and 0.5];
\end{scope}
\path [fill=white!70!gray] (0, 1.9).. controls (-1, 1.9) and (-1.5, 1.25)..(-2, 1.15) arc (90:270:0.25 and 0.4);
\path [fill=white!70!gray] (0, 1.9) arc (90:-90:0.25 and 0.4) arc(90:270:0.35 and 0.35) arc (90:-90: 0.25 and 0.4) ..controls +(-1, 0) and (-1.5, 0.25)..(-2, 0.35);
\draw  (0, 1.9).. controls (-1, 1.9) and (-1.5, 1.25)..(-2, 1.15) arc (90:270:0.25 and 0.4);
\draw  (0, 1.9) arc (90:-90:0.25 and 0.4) arc(90:270:0.35 and 0.35) arc (90:-90: 0.25 and 0.4) ..controls +(-1, 0) and (-1.5, 0.25)..(-2, 0.35);
\draw [dashed] (0, 0) [partial ellipse=90:270:0.3 and 0.5];
\draw [dashed] (0, 1.5) [partial ellipse=90:270:0.3 and 0.5];
\draw [dashed] (-2, 0.75) [partial ellipse=-90:90:0.3 and 0.5];
\draw [dashed] (0, 0) [partial ellipse=90:270:0.25 and 0.4];
\draw [dashed] (0, 1.5) [partial ellipse=90:270:0.25 and 0.4];
\draw [dashed] (-2, 0.75) [partial ellipse=-90:90:0.25 and 0.4];
\end{tikzpicture}}}
$$

Given three objects $Y_i$, $Y_j$, $Y_{k}$ in $\mathcal{Z(C)}$, and a morphism $f\in \hom_{\mathcal{Z(C)}}(Y_i\otimes Y_j\otimes Y_k, \1)$, we associate the following morphism in $\mathcal{A}$ as building blocks.

$$
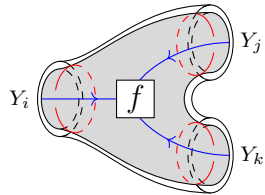
\vcenter{\hbox{
\begin{tikzpicture}
\draw (0, 0.75) [partial ellipse=90:270: 0.25 and 0.25];
\draw (0, 2)..controls +(-1, 0) and (-1.5, 1.25)..(-2, 1.25);
\draw (0, -0.5)..controls +(-1, 0) and (-1.5, 0.25)..(-2, 0.25);
\draw [fill=white] (0, 0) [partial ellipse=0:360:0.3 and 0.5];
\draw [fill=white] (0, 1.5) [partial ellipse=0:360:0.3 and 0.5];
\begin{scope}[xshift=-2cm, yshift=0.75cm]
\draw [fill=white] (0, 0) [partial ellipse=0:360:0.3 and 0.5];
\end{scope}
\path [fill=white!70!gray] (0, 1.9).. controls (-1, 1.9) and (-1.5, 1.25)..(-2, 1.15) arc (90:270:0.25 and 0.4);
\path [fill=white!70!gray] (0, 1.9) arc (90:-90:0.25 and 0.4) arc(90:270:0.35 and 0.35) arc (90:-90: 0.25 and 0.4) ..controls +(-1, 0) and (-1.5, 0.25)..(-2, 0.35);
\draw  (0, 1.9).. controls (-1, 1.9) and (-1.5, 1.25)..(-2, 1.15) arc (90:270:0.25 and 0.4);
\draw  (0, 1.9) arc (90:-90:0.25 and 0.4) arc(90:270:0.35 and 0.35) arc (90:-90: 0.25 and 0.4) ..controls +(-1, 0) and (-1.5, 0.25)..(-2, 0.35);
\draw [dashed] (0, 0) [partial ellipse=90:270:0.3 and 0.5];
\draw [dashed] (0, 1.5) [partial ellipse=90:270:0.3 and 0.5];
\draw [dashed] (-2, 0.75) [partial ellipse=-90:90:0.3 and 0.5];
\draw [dashed] (0, 0) [partial ellipse=90:270:0.25 and 0.4];
\draw [dashed] (0, 1.5) [partial ellipse=90:270:0.25 and 0.4];
\draw [dashed] (-2, 0.75) [partial ellipse=-90:90:0.25 and 0.4];
\draw [red](-1.8, 0.75)[partial ellipse=90:165:0.27 and 0.45];
\draw [red](-1.8, 0.75)[partial ellipse=-90:-165:0.27 and 0.45];
\draw [red, dashed](-1.8, 0.75)[partial ellipse=-90:90:0.27 and 0.45];
\draw [red, dashed](-0.2, 0.027) [partial ellipse=90:270:0.25 and 0.41];
\draw [red](-0.2, 0.027) [partial ellipse=15:90:0.25 and 0.41];
\draw [red](-0.2, 0.027) [partial ellipse=-15:-90:0.25 and 0.41];
\begin{scope}[yscale=-1, yshift=-1.5cm]
\draw [red, dashed](-0.2, 0.027) [partial ellipse=90:270:0.25 and 0.41];
\draw [red](-0.2, 0.027) [partial ellipse=15:90:0.25 and 0.41];
\draw [red](-0.2, 0.027) [partial ellipse=-15:-90:0.25 and 0.41];
\draw[blue, ->-=0.6] (0.25, 0)node[black, right]{\tiny$Y_j$}.. controls (-0.5, 0) and (-1, 0.5)..(-1, 0.6);
\end{scope}
\draw[blue, ->-=0.6] (-2.25, 0.75) node[black, left]{\tiny$Y_i$} to (-1, 0.75);
\draw[blue, ->-=0.6] (0.25, 0)node[black, right]{\tiny$Y_k$}.. controls (-0.5, 0) and (-1, 0.5)..(-1, 0.6);
\draw[fill=white] (-1.25, 0.5) rectangle (-0.75, 1);
\draw (-1, 0.75) node{$f$};
\end{tikzpicture}}}
$$

Let $\{Y_{0}=\1_{\mathcal{Z(C)}}, Y_1 ,\ldots\}$ be a set of isomorphism classes of irreducible objects in $\mathcal{Z(C)}$.

We define a coloring $c$ of a pants decomposition $\mathbf{C}$ is an mapping from the set of cutting circles to the set of irreducible objects of $\mathcal{Z(C)}$. Given a coloring, we associate pants $\mathbb{P}$ with boundary colored by $Y_i, Y_j, Y_k$ a morphism space $V_{\mathbb{P}, c}=\hom_{\mathcal{Z(C)}}(Y_i^{?}\otimes Y_j^{?}\otimes Y_k^{?}, \1)$, where the question mark is chosen to be $1$ if the orientation of the cutting circle equals to the boundary orientation of the pants $\mathbb{P}$. Otherwise, we take the question mark to be $*$. Let $B_{\mathbb{P}, c}$ be a basis of $V_{\mathbb{P}, c}$, we define a set of vectors
$$B_{\mathbf{C}}:=\bigsqcup_{c}\left(\bigotimes_{\mathbb{P}}B_{\mathbb{P}, c}\right)$$
and a mapping $\Phi: B_{\mathbf{C}}\rightarrow \V(F, \emptyset_B, \emptyset)$ by gluing up the building blocks described above along the disks bounded by the cutting disks.

\begin{theorem}\label{thm:rtbasis}
The image of $\Phi$ is a basis of $\V(F, \emptyset_B, \emptyset)$.
\end{theorem}
\begin{proof}
We first show all vectors in $\V(F, \emptyset_B, \emptyset)$ can be written as linear combinations of $\Phi(B)$. 
Let $\mathbb{M}=(M, \Sigma, \Gamma)$ be an alterfold cobordism with $\partial_{-}\mathbb{M}=\emptyset$ and $\partial_{+}\mathbb{M}=\mathbb{F}=(F, \emptyset_B, \emptyset)$. 
By Theorem \ref{thm:collarrepn}, we may assume $M$ is a handlebody such that each curve in $\mathcal{C}$ bounds a disk in $M$ and $\Sigma$ is a surface parallel to the boundary $\partial  M=F$. 
Therefore, we can push the cutting circles $C_i$ to the space boundary 
By applying the Planar Graphical Calculus if necessary, we may assume $\Gamma$ intersects $C_i$ transversely at most one point. 
In the other word, we have the following picture near each separating curve $C_i$.
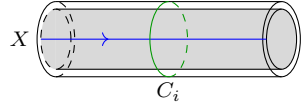
$$\begin{tikzpicture}
\path[fill=white!70!gray] (-1.5, 0) [partial ellipse=90:270:0.2 and 0.4];
\path[fill=white!70!gray](1.5, 0) [partial ellipse=-90:90:0.20 and 0.4];
\path[fill=white!70!gray] (-1.5, -0.4) rectangle (1.5, 0.4);
\draw[] (-1.5, 0) [partial ellipse=90:270:0.25 and 0.5];
\draw[dashed] (-1.5, 0) [partial ellipse=-90:90:0.25 and 0.5];
\draw[] (1.5, 0) [partial ellipse=0:360:0.25 and 0.5];
\draw (-1.5, -0.5)--(1.5, -0.5);
\draw (-1.5, 0.5)--(1.5, 0.5);
\begin{scope}[yscale=0.8]
\draw (-1.5, -0.5) --(1.5, -0.5);
\draw (-1.5, 0.5)--(1.5, 0.5);
\end{scope}
\draw[dashed] (-1.5, 0) [partial ellipse=0:360:0.20 and 0.4];
\draw (1.5, 0) [partial ellipse=0:360:0.20 and 0.4];
\draw[blue, ->-=0.3] (-1.7, 0)node[left, black]{\tiny$X$}--(1.3, 0);
\draw[\darkgreen] (0, 0) [partial ellipse=90:270:0.25 and 0.5];
\draw[\darkgreen, dashed] (0, 0) [partial ellipse=-90:90:0.25 and 0.5];
\draw (0, -0.7) node{\tiny$C_i$};
\end{tikzpicture}$$
The diagram can be regarded as the identity morphism of $I(X)$ in the Drinfeld center $\mathcal{Z(C)}$. 
Therefore, the diagram can be written as
$$\sum_{i}\vcenter
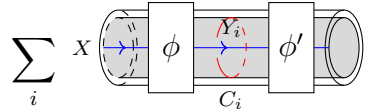
{\hbox{\begin{tikzpicture}
\path[fill=white!70!gray] (-1.5, 0) [partial ellipse=90:270:0.2 and 0.4];
\path[fill=white!70!gray](1.5, 0) [partial ellipse=-90:90:0.20 and 0.4];
\path[fill=white!70!gray] (-1.5, -0.4) rectangle (1.5, 0.4);
\draw[] (-1.5, 0) [partial ellipse=90:270:0.25 and 0.5];
\draw[dashed] (-1.5, 0) [partial ellipse=-90:90:0.25 and 0.5];
\draw[] (1.5, 0) [partial ellipse=0:360:0.25 and 0.5];
\draw (-1.5, -0.5)--(1.5, -0.5);
\draw (-1.5, 0.5)--(1.5, 0.5);
\begin{scope}[yscale=0.8]
\draw (-1.5, -0.5) --(1.5, -0.5);
\draw (-1.5, 0.5)--(1.5, 0.5);
\end{scope}
\draw[dashed] (-1.5, 0) [partial ellipse=0:360:0.20 and 0.4];
\draw (1.5, 0) [partial ellipse=0:360:0.20 and 0.4];
\draw[blue, ->-=0.1, ->-=0.55] (-1.7, 0)node[left, black]{\tiny$X$}--(1.3, 0);
\draw (0, -0.7) node{\tiny$C_i$};
\draw[fill=white] (-1.1, -0.6) rectangle (-0.5, 0.6);
\draw[fill=white] (1.1, -0.6) rectangle (0.5, 0.6);
\draw[red] (0, 0) [partial ellipse=90:160:0.2 and 0.4];
\draw[red] (0, 0) [partial ellipse=200:270:0.2 and 0.4];
\draw[red, dashed] (0, 0) [partial ellipse=90:-90:0.2 and 0.4];
\draw (-0.8, 0) node{$\phi$};
\draw (0.8, 0) node{$\phi'$};
\draw (0, 0) node[above]{\tiny$Y_i$};
\end{tikzpicture}}}$$
where $\phi$ and $\phi'$ are a pair of dual basis in $\mathcal{Z(C)}$. The morphism in the middle are idempotent, we can further write the diagram as 
$$
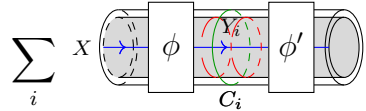
\sum_{i}\vcenter{\hbox{\begin{tikzpicture}
\path[fill=white!70!gray] (-1.5, 0) [partial ellipse=90:270:0.2 and 0.4];
\path[fill=white!70!gray](1.5, 0) [partial ellipse=-90:90:0.20 and 0.4];
\path[fill=white!70!gray] (-1.5, -0.4) rectangle (1.5, 0.4);
\draw[] (-1.5, 0) [partial ellipse=90:270:0.25 and 0.5];
\draw[dashed] (-1.5, 0) [partial ellipse=-90:90:0.25 and 0.5];
\draw[] (1.5, 0) [partial ellipse=0:360:0.25 and 0.5];
\draw (-1.5, -0.5)--(1.5, -0.5);
\draw (-1.5, 0.5)--(1.5, 0.5);
\begin{scope}[yscale=0.8]
\draw (-1.5, -0.5) --(1.5, -0.5);
\draw (-1.5, 0.5)--(1.5, 0.5);
\end{scope}
\draw[dashed] (-1.5, 0) [partial ellipse=0:360:0.20 and 0.4];
\draw (1.5, 0) [partial ellipse=0:360:0.20 and 0.4];
\draw[blue, ->-=0.1, ->-=0.55] (-1.7, 0)node[left, black]{\tiny$X$}--(1.3, 0);
\draw (0, -0.7) node{\tiny$C_i$};
\draw[fill=white] (-1.1, -0.6) rectangle (-0.5, 0.6);
\draw[fill=white] (1.1, -0.6) rectangle (0.5, 0.6);
\begin{scope}[xshift=0.2cm]
\draw[red] (0, 0) [partial ellipse=90:160:0.2 and 0.4];
\draw[red] (0, 0) [partial ellipse=200:270:0.2 and 0.4];
\draw[red, dashed] (0, 0) [partial ellipse=90:-90:0.2 and 0.4];
\end{scope}
\begin{scope}[xshift=-0.2cm]
\draw[red] (0, 0) [partial ellipse=90:160:0.2 and 0.4];
\draw[red] (0, 0) [partial ellipse=200:270:0.2 and 0.4];
\draw[red, dashed] (0, 0) [partial ellipse=90:-90:0.2 and 0.4];
\end{scope}
\draw (-0.8, 0) node{$\phi$};
\draw (0.8, 0) node{$\phi'$};
\draw (0, 0) node[above]{\tiny$Y_i$};
\draw[\darkgreen] (0, 0) [partial ellipse=90:270:0.25 and 0.5];
\draw[\darkgreen, dashed] (0, 0) [partial ellipse=-90:90:0.25 and 0.5];
\draw (0, -0.7) node{\tiny$C_i$};
\end{tikzpicture}}}$$
Applying the above graphical calculus near each $C_i$, then the diagram on each pants $\mathbb{P}$ can be written as linear combination of vectors in $B_{\mathbb{P}, c}$.

To show vectors in $B_{\mathbf{C}}$ are linearly independent, we define a dual basis of $\V(-F, \emptyset_B, \emptyset)$ under the canonical pairing. 
Consider the pants decomposition $\mathbf{C}$ of $-F$, we define the dual coloring $\overline{c}$ of $c$ by $\overline{c}(C_i):=c(C_i)^{\star}$. 
Let $B_{\mathbb{P}, c}^{'}$ be the the basis of $V_{\mathbb{P}, \overline{c}}$ that dual to $B_{\mathbb{P}, c}$. 
It is easy to verify that $\displaystyle B_{\mathbf{C}}':=\bigsqcup\left(\bigotimes_{\mathbb{P}}B_{\mathbb{P}, \overline{c}}'\right)$ is dual to $B_{\mathbf{C}}$.

\end{proof}

\begin{remark}
According to Theorem \ref{thm:rtbasis}, the space $\V(F, \emptyset_B, \emptyset)$ can be identified as the following hom-space
$$\hom_{\mathcal{Z(C)}}\left(\1_{\mathcal{Z(C)}}, \bigoplus_{c}\bigotimes_{i=1}^{g} Y_{c(i)}\otimes Y_{c(i)}^{*}\right),$$
where the direct sum is over all colorings $c:\{1,\ldots, g\}\rightarrow  \Irr_{\mathcal{Z(C)}}$. 
For example, when genus is 2, such a basis element can be depicted as follows:
$$
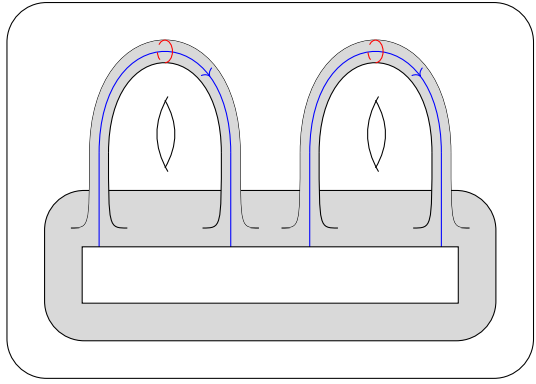
\begin{tikzpicture}
\path[fill=white!70!gray, rounded corners=15pt] (-3, -2) rectangle (3, 0);
\draw [rounded corners=15pt] (-3, -2) rectangle (3, 0);
\begin{scope}[xshift=-1.4cm, yshift=0.5cm]
\draw (-1.25, -1)..controls(-1, -1) and (-1, -1)..(-1, 0)..controls (-1, 2) and (1, 2)..(1, 0).. controls (1, -1) and (1, -1)..(1.25, -1);
\path[fill=white!70!gray](-1.25, -1)..controls(-1, -1) and (-1, -1)..(-1, 0)..controls (-1, 2) and (1, 2)..(1, 0).. controls (1, -1) and (1, -1)..(1.25, -1)--(0.5, -1).. controls (.75, -1) and (0.75, -1).. (0.75, 0)..controls (0.75, 1.6) and (-0.75, 1.6)..(-0.75, 0)..controls (-0.75, -1) and (-0.75, -1)..(-0.5, -1);
\draw (-0.5, -1).. controls (-0.75, -1) and (-0.75, -1).. (-0.75, 0);
\draw (-0.75, 0)..controls (-0.75, 1.6) and (0.75, 1.6)..(0.75, 0);
\draw (0.5, -1).. controls (.75, -1) and (0.75, -1).. (0.75, 0);
\draw[blue, ->-=0.7] (-0.875, -0).. controls (-0.875,1.8) and (0.875, 1.8).. (0.875, 0);
\draw[blue](-0.875, -1.5)--(-0.875, 0);
\draw[blue](0.875, -1.5)--(0.875, 0);
\draw[red](0, 1.35)[partial ellipse=-90:90:0.1 and 0.15];
\draw[red, dashed](0, 1.35)[partial ellipse=90:270:0.1 and 0.15];
\end{scope}
\begin{scope}[xshift=1.4cm, yshift=0.5cm]
\draw (-1.25, -1).. controls (-1, -1) and (-1, -1).. (-1, 0);
\draw (-1, 0)..controls (-1, 2) and (1, 2)..(1, 0);
\draw (1.25, -1).. controls (1, -1) and (1, -1).. (1, 0);
\path[fill=white!70!gray](-1.25, -1)..controls(-1, -1) and (-1, -1)..(-1, 0)..controls (-1, 2) and (1, 2)..(1, 0).. controls (1, -1) and (1, -1)..(1.25, -1)--(0.5, -1).. controls (.75, -1) and (0.75, -1).. (0.75, 0)..controls (0.75, 1.6) and (-0.75, 1.6)..(-0.75, 0)..controls (-0.75, -1) and (-0.75, -1)..(-0.5, -1);
\draw (-0.5, -1).. controls (-0.75, -1) and (-0.75, -1).. (-0.75, 0);
\draw (-0.75, 0)..controls (-0.75, 1.6) and (0.75, 1.6)..(0.75, 0);
\draw (0.5, -1).. controls (.75, -1) and (0.75, -1).. (0.75, 0);
\draw[blue, ->-=0.7] (-0.875, -0).. controls (-0.875,1.8) and (0.875, 1.8).. (0.875, 0);
\draw[blue](-0.875, -1.5)--(-0.875, 0);
\draw[blue](0.875, -1.5)--(0.875, 0);
\draw[red](0, 1.35)[partial ellipse=-90:90:0.1 and 0.15];
\draw[red, dashed](0, 1.35)[partial ellipse=90:270:0.1 and 0.15];
\end{scope}
\draw[fill=white] (-2.5, -1.5) rectangle (2.5, -0.75);
\draw (-1.36, 0.25)to[bend left](-1.36, 1.25);
\draw (-1.4, 0.3) to[bend right](-1.4, 1.2);
\draw (1.44, 0.25)to[bend left](1.44, 1.25);
\draw (1.4, 0.3) to[bend right](1.4, 1.2);
\draw [rounded corners=15pt] (-3.5, -2.5) rectangle (3.5, 2.5);
\end{tikzpicture}
$$
\end{remark}

\section{Relation between TV TQFT and RT TQFT}

The 3-alterfold topological quantum field theory contains Reshetikhin-Turaev topological quantum field theory and Turaev-Viro topological quantum field theory as sub-TQFTs, in the sense that the following diagram commutes, where the horizontal embeddings of $\mathrm{Cob}$ in $\mathrm{ACob}_{2}^{\mathcal{C}}$ is defined to be the mapping that taking ordinary cobordisms and surfaces to $B$-colored alterfold cobordisms and surfaces without decoration.
$$\begin{tikzcd}
{\mathrm{Cob}}\arrow[r]\arrow[dr, swap, "\operatorname{RT}_{\mathcal{Z(C)}}"]&{\Cob}\arrow[d, "{\V}"]&{\mathrm{Cob}}\arrow[l]\arrow[dl, "\operatorname{TV}_{\mathcal{C}}"]\\
{}&{\Vec}&{}
\end{tikzcd}$$

\begin{theorem}\label{thm:embedding}
The restriction of the functor $\V$ over the subcategory consisting $B$-colored surfaces and cobordisms equals to $\operatorname{RT}_{\mathcal{Z(C)}}$ and $\operatorname{TV}_{\mathcal{C}}$ respectively. 
In particular, $\operatorname{RT}_{\mathcal{Z(C)}}=\operatorname{TV}_{\mathcal{C}}$ as symmetric monoidal functors.
\end{theorem}

\begin{proof}
We first show the commutativity of the left triangle. 
Let $F$ be an oriented closed surface, the Reshetikhin-Turaev TQFT assigns to it a vector space
$$\text{span}((M, \partial_{-}M=\emptyset, \partial_{+}M=F)\slash \ker(\operatorname{RT}_{Z(\mathcal{C})}(\cdot\sqcup_{F} \cdot )).$$
Since $\operatorname{RT}_{\mathcal{Z(C)}}$ and $\V$ are equal to each other as partition functions defined on closed oriented $3$-manifold(\cite[Theorem 4.36]{LMWW23}).
To show the inclusion of the ordinary cobordism into $\Cob$ induces a well-defined linear transformation $\operatorname{RT}_{\mathcal{Z(C)}}(F)\rightarrow \V(F)$, we need to show 
$$\ker (\operatorname{RT}_{\mathcal{Z(C)}}(\cdot \sqcup_{F} \cdot))\subset \ker(\V(\cdot \sqcup_{\mathbb{F}} \cdot)).$$
 Suppose $[\mathbb{M}]\in \text{span}((M, \partial_{-}M=\emptyset, \partial_{+}M=F))$, we shall show that $\V(\mathbb{M}\sqcup_{\mathbb{F}} \cdot)=0$.
 By Theorem \ref{thm:rtbasis}, we only need to show $\V([\mathbb{M}]\sqcup_{\mathbb{F}} \Phi(\cdot))=0$. 
 Since $\V$ and $\operatorname{RT}_{\mathcal{Z(C)}}$ equals to each other as partition functions on closed $3$-manifold, this value equals $\operatorname{RT}([\mathbb{M}]\sqcup_{F} (H, \cdot))$, where $H$ is a handlebody with boundary $-F$ and the argument is a tensor diagram in $\mathcal{Z(C)}$. 
 On the other hand, this linear map is an injection by definition, in order to show it is an isomorphism, we argue their dimensions are equal, which is a consequence of the trace formula.

Next we show the commutativity of the right triangle. 
Let $F$ be an oriented closed surface with triangulation $\Delta$.
The vector space $\operatorname{TV}_{\mathcal{C}}(F)$ is generated by the decorations of $\Delta$.
To be precise, an irreducible object is assigned to each edges and an admissible morphism is assigned to each face. 
Such assignment naturally defines a tensor diagram in $F$ in the follow manner. 
$$
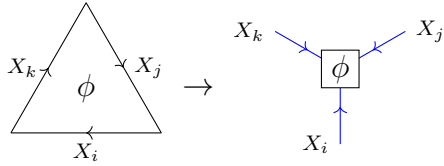
\vcenter{\hbox{\begin{tikzpicture}
\draw[->-=0.5] (-1, 0) -- (0, 1.732) node[pos=0.5, left]{\tiny$X_{k}$};
\draw[->-=0.5] (0, 1.732) -- (1, 0)node[pos=0.5, right]{\tiny$X_{j}$};
\draw[->-=0.5] (1, 0) -- (-1, 0)node[pos=0.5, below]{\tiny$X_{i}$};
\draw (0, 0.6) node{$\phi$};
\end{tikzpicture}}}
\rightarrow
\vcenter{\hbox{\begin{tikzpicture}
\draw[blue, ->-=0.5] (0, -1) node[left, black]{\tiny$X_i$} to (0, 0);
\draw[blue, ->-=0.5] (-0.866, 0.5)node[left, black]{\tiny$X_k$} to (0, 0);
\draw[blue, ->-=0.5] (0.866, 0.5)node[right, black]{\tiny$X_j$} to (0, 0);
\draw[fill=white] (-0.25, -0.25) rectangle (0.25, 0.25);
\draw (0, 0) node{$\phi$};
\end{tikzpicture}}}
$$
Hence, it further defines a collar representative in $\V(F, \emptyset_B, \emptyset)$. 
To show this is a well-defined map from $\operatorname{TV}_{\mathcal{C}}(F)$ to $\V(\mathbb{F})$, one only needs to show the kernel of the linear map defined by $(F\times I, \Delta')$ is sent to the zero vector in $\V(F, \emptyset_B, \emptyset)$.
This is obvious since the evaluation of the state-sum along the triangulation $\Delta'$ is identical by applying the moves along the handle decomposition corresponding to $\Delta'$. 
On the other hand, the linear map $\operatorname{TV}_{\mathcal{C}}(F)\rightarrow \V(F, \emptyset_B, \emptyset)$ is surjective, which follows from the identical construction in Lemma \ref{lem:surjectivity}. 
Since $\operatorname{TV}_{\mathcal{C}}$ and $\V$ are equal to each other as partition functions on closed $3$-manifold, by the trace formula (Lemma \ref{lem:traceformula}), the dimension of $\operatorname{TV}_{\mathcal{C}}(F)$ and $\V(\mathbb{F})$ are equal. 
Therefore, this linear transformation is an isomorphism.
We finished the proof of the theorem.

\end{proof}

\section{$\mathcal{E}$-Decorated Alterfold TQFT}

Replacing the graphical calculus for a spherical fusion category $\mathcal{C}$ by the one for spherical Morita context $\mathcal{E}$, we obtain $\mathcal{E}$-decorated 3-alterfolds. 
For $\mathcal{E}$-decorated $3$-alterfolds, we still have a series of moves as follows:

\begin{itemize}
\item \textbf{Planar Graphical Calculus:} The planar graphical calculus for Morita context $\mathcal{E}$, similar to the planar graphical calculus described as \eqref{eq:calculus}.

\item \textbf{Move 0:} This is similar to Move 0 described as \eqref{eq:move0}, the boundary of the $A$-colored $0$-handle $P_\epsilon$ can be colored by either $\mathfrak{A}$ or $\mathfrak{B}$. We remark that
the boundary color can be switched to the other by applying planar graphical calculus in $\mathcal{E}$, i.e. adding a closed string labelled by $J$, with its orientation property chosen.

\item \textbf{Move 1:} This is similar to Move 1 described as \eqref{eq:move1}, the $\Omega$-color on the belt is replaced by the $\Omega$-color in $\Hom(\mathfrak{A}, \mathfrak{A})$, $\Hom(\mathfrak{A}, \mathfrak{B})$, $\Hom(\mathfrak{B}, \mathfrak{A})$, $\Hom(\mathfrak{B}, \mathfrak{B})$, depending on the color on the attaching boundary of the $1$-handle. Below is an example of adding $1$-handle attached to $\mathfrak{B}$ colored regions. See Corollary \ref{cor:omega} for the relations between these $\Omega$ colors.

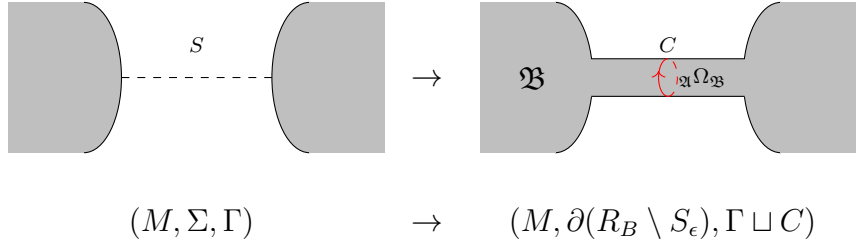
\begin{align*}
\begin{array}{ccc}
   \vcenter{\hbox{
\begin{tikzpicture}
\path[fill=gray!50!white]
(-1.5, 0) [partial ellipse=-90:90:0.5 and 1];
\draw (-1.5, 0) [partial ellipse=-90:90:0.5 and 1];
\node at (-1.8, 0) {$\mathfrak{B}$};
\node at (1.8, 0) {$\mathfrak{A}$};
\path[fill=gray!50!white]
(-2.5, -1) rectangle (-1.5, 1);
\begin{scope}[xscale=-1]
\path[fill=gray!50!white]
(-1.5, 0) [partial ellipse=-90:90:0.5 and 1];
\draw (-1.5, 0) [partial ellipse=-90:90:0.5 and 1];
\path[fill=gray!50!white]
(-2.5, -1) rectangle (-1.5, 1);
\end{scope}
\draw[dashed] (-1, 0)--(1, 0);
\draw (0, 0.2) node[above]{\tiny{$S$}};
\end{tikzpicture}}}  & \rightarrow & \vcenter{\hbox{
\begin{tikzpicture}
\path[fill=gray!50!white]
(-1.5, 0) [partial ellipse=-90:90:0.5 and 1];
\draw (-1.5, 0) [partial ellipse=-90:90:0.5 and 1];
\path[fill=gray!50!white]
(-2.5, -1) rectangle (-1.5, 1);
\node at (-1.8, 0) {$\mathfrak{B}$};
\node at (1.8, 0) {$\mathfrak{A}$};
\begin{scope}[xscale=-1]
\path[fill=gray!50!white]
(-1.5, 0) [partial ellipse=-90:90:0.5 and 1];
\draw (-1.5, 0) [partial ellipse=-90:90:0.5 and 1];
\path[fill=gray!50!white]
(-2.5, -1) rectangle (-1.5, 1);
\end{scope}
\draw[dashed] (-1, 0)--(1, 0);
\draw (0, 0.2) node[above]{\tiny{$C$}};
\path[fill=white]
(-1.1, -0.25) rectangle (1.1, 0.25);
\path[fill=gray!50!white]
(-1.1, -0.25) rectangle (1.1, 0.25);
\draw (-1.02, 0.25)--(1.02, 0.25);
\draw (-1.02, -0.25)--(1.02, -0.25);
\draw[red, dashed] (0, 0) [partial ellipse=-90:90:0.125 and 0.25];
\draw[red, -<-=0.5] (0, 0) [partial ellipse=90:270:0.125 and 0.25];
\draw (0, 0) node[right]{\tiny${}_{\mathfrak{A}}\Omega_{\mathfrak{B}}$};
\end{tikzpicture}}}\\
\vspace{2mm}\\
(M, \Sigma, \Gamma)&\rightarrow & (M, \partial(R_{B}\setminus S_{\epsilon}), \Gamma\sqcup C)\\
\end{array}
\end{align*}

\item \textbf{Move 2:}  This is identical to Move 2 described as \eqref{eq:move2}.
\item \textbf{Move 3:} This is identical to Move 3, no matter whether the boundary of the sphere is colored by $\mathfrak{A}$ or $\mathfrak{B}$.
\end{itemize}

    

By a similar argument as in Theorem 3.9 of \cite{LMWW23}, we have the following theorem extending the planar graphical calculus associated to Morita context $\mathcal{E}$ to a partition function over $\mathcal{E}$-decorated $3$-alterfolds.
\begin{theorem}\label{thm:e partition function}
Let $\mathcal{E}$ be a Morita context over $\k$, and $\zeta \in \k^\times$ nonzero scalars. 
There exists a unique partition function $Z_e$ from $\mathcal{E}$-decorated $3$-alterfolds to the ground field $\k$ satisfying the following conditions.
\begin{itemize}
    \item If $R_{B}=\emptyset$, $Z_e(M,\emptyset,\emptyset)=1$.
    
    \item \textbf{Disjoint Union:} Suppose $(M, \Sigma, \Gamma)$ and $(M', \Sigma', \Gamma')$ are two $\mathcal{E}$-decorated $3$-alterfolds. Then
    $$Z_e(M\sqcup M', \Sigma \sqcup \Sigma', \Gamma\sqcup\Gamma')=Z_e(M, \Sigma, \Gamma)Z_e(M', \Sigma', \Gamma').$$
    
     \item \textbf{Homeomorphims:} The evaluation $Z_e(M, \Sigma, \Gamma)$ only depends on the orientation preserving homeomorphism class of $(R_{B}, \Gamma)$.
     
    \item \textbf{Planar Graphical Calculus:} If $\Gamma'$ and $\Gamma$ are diagrams on $\Sigma$ that are identical outside of a contractible region $D$, and $\Gamma_{D}:=\Gamma\cap D$ equals to $\Gamma_{D}':=\Gamma'\cap D$ as morphisms in $\mathcal{E}$. Then
    $$Z_e(M, \Sigma, \Gamma)=Z_e(M, \Sigma, \Gamma').$$
    \item \textbf{Move 0:} Suppose $(M, \Sigma', \Gamma')$ is derived from $(M, \Sigma, \Gamma)$ by applying Move $0$, then
    $$Z_e(M, \Sigma', \Gamma')=\zeta\mu Z_e(M, \Sigma, \Gamma).$$
    
    \item \textbf{Move 1:} Suppose $(M, \Sigma', \Gamma')$ is derived from $(M, \Sigma, \Gamma)$ by applying Move $1$, then
    $$Z_e(M, \Sigma', \Gamma')=\frac{1}{\zeta} Z_e(M, \Sigma, \Gamma).$$
    
    \item \textbf{Move 2:} Suppose $(M, \Sigma', \Gamma')$ is derived from $(M, \Sigma, \Gamma)$ by applying Move $2$, then
    $$Z_e(M, \Sigma', \Gamma')=\zeta Z_e(M, \Sigma, \Gamma).$$
    
    \item \textbf{Move 3:} Suppose $(M, \Sigma', \Gamma')$ is derived from $(M, \Sigma, \Gamma)$ by applying Move $3$, then
    $$Z_e(M, \Sigma', \Gamma')=\frac{1}{\zeta} Z_e(M, \Sigma, \Gamma).$$
\end{itemize}
\end{theorem}

We usually assume that $\zeta=1$ for computational convenience.
Replacing the graphical calculus for a spherical fusion category $\mathcal{C}$ by the one for a 2-category, we now have $\mathcal{E}$-decorated alterfold cobordism category, $\mathcal{E}$-decorated 2-alterfolds. 
\begin{theorem}\label{thm:etqft}
The partition function $Z_e$ defined in Theorem \ref{thm:e partition function} extends to a topological quantum field theory 
$$\mathbb{V}_e: \text{ACob}_2^\mathcal{E} \rightarrow \Vec.$$
\end{theorem}
\begin{proof}
    The proof is similar to the one of Theorem \ref{thm:tqft}.
\end{proof}

\begin{theorem}\label{thm:eutqft}
Suppose $\mathcal{E}$ is unitary. 
Then the functor $\V_e$ is a unitary topological quantum field theory.
\end{theorem}
\begin{proof}
    The proof is similar to the one of Theorem \ref{thm:utqft}.
\end{proof}

The following corollary is Theorem 7.1 in \cite{Mug03a}, and also Corollary 17.9 in \cite{TurVir17}. We explain it in our settings.
\begin{corollary}\label{cor:tvmortia}
Let $\mathcal{C}$ and $\mathcal{D}$ be Morita equivalent spherical fusion categories. 
Then $\operatorname{TV}_{\mathcal{C}}(M)=\operatorname{TV}_{\mathcal{D}}(M)$ for all compact closed $3$-manifold $M$.
\end{corollary}

\begin{proof}
Choose Morita context $\mathcal{E}$ encoding the Morita equivalence of $\mathcal{D}$. Then we have
$$\operatorname{TV}_{\mathcal{C}}(M)=\V(\M)=\operatorname{TV}_{\mathcal{D}}(M)$$
where $\M$ is the $B$-colored $3$-alterfold $(M, \emptyset, \emptyset)$
\end{proof}

By a similar argument as in Theorem \ref{thm:embedding}, we have embeddings of TQFTs.

\begin{corollary}
Let $\mathcal{E}$ be a spherical Morita context encoding the Morita equivalence of spherical fusion categories $\mathcal{C}$ and $\mathcal{D}$. 
The following diagram commutes.
$$\begin{tikzcd}
{\mathrm{Cob}}\arrow[r]\arrow[dr, swap, "\operatorname{TV}_{\mathcal{C}}"]&{\Cob}\arrow[d, "{\V}"]&{\mathrm{Cob}}\arrow[l]\arrow[dl, "\operatorname{TV}_{\mathcal{D}}"]\\
{}&{\Vec}&{}
\end{tikzcd}$$
Moreover, $\operatorname{TV}_{\mathcal{C}}$ and $\operatorname{TV}_{\mathcal{D}}$ are equivalent.
\end{corollary}

\begin{remark}
Beyond the equivalence of Turaev-Viro TQFTs of Morita equivalent categories, the multi-color structure over the separating surface will allow us to do extra moves to simulate algebraic operations.
When one of $\mathcal{C}$ and $\mathcal{D}$ is a ribbon fusion category, more algebraic properties such as the modular invariant \cite{BEK99, BEK00, Xu98} can be simulated in our multi-color setting and we will discuss it in detail in a forthcoming paper.
\end{remark}

\subsection{Equalities and Inequalities}
In this section, we derive a set of equalities and inequalities by analyzing the information pertaining to the 3-alterfold topological quantum field theory (TQFT).

\begin{proposition}
Let $\mathcal{C}$ be a spherical fusion category and $Q$ be a Frobenius algebra. 
Denote the category of right $Q$-module by $\mathcal{M}$, then
$$|\Irr_{\mathcal{M}}|=\sum_{i\in \Irr_{\mathcal{C}}}d_i
\vcenter{\hbox{\begin{tikzpicture}[scale=1.2]
\draw[blue, ->-=0.5] (-1, 0)--(0, 1);
\draw[blue, ->-=0.5] (0, 1)--(1, 0);
\draw[blue, ->-=0.5] (1, 0)--(0, -1);
\draw[blue, ->-=0.5] (0, -1)--(-1, 0);
\draw[FA] (-1, 0)--(1, 0);
\draw[FA] (0, -1)--(0, 1);
\fill[FA] (0, 0) circle (0.05);
\draw[fill=white] (-0.3, 0.8) rectangle (0.3, 1.2);
\draw[fill=white] (-0.3, -0.8) rectangle (0.3, -1.2);
\begin{scope}[rotate=90]
\draw[fill=white] (-0.3, 0.8) rectangle (0.3, 1.2);
\draw[fill=white] (-0.3, -0.8) rectangle (0.3, -1.2);
\end{scope}
\draw (-1, 0) node{$\phi$};
\draw (1, 0) node{$\phi'$};
\draw (0, 1) node{$\psi$};
\draw (0, -1) node{$\psi'$};
\draw (-0.65, 0.65) node{\tiny$X_i$};
\draw (0.65, -0.65) node{\tiny$X_i$};
\draw (0.65, 0.65) node{\tiny$X_i$};
\draw (-0.65, -0.65) node{\tiny$X_i$};
\end{tikzpicture}}}, \quad\quad\quad\quad  |\Irr_{\mathcal{C}_{\mathcal{M}}^{*}}|=\sum_{i\in \Irr_{\mathcal{C}}}d_i
\vcenter{\hbox{
\begin{tikzpicture}[scale=1.2]
\draw[FA, rounded corners] (-1.3, 0) rectangle (1.3, 1.5);
\draw[FA, rounded corners] (0, 1.3) rectangle (1.5, -1.3);
\fill[FA] (1.3, 1.3) circle(0.05); 
\fill[white] (-1, -1) rectangle (1, 1);
\draw[blue, ->-=0.5] (-1, 0)--(0, 1);
\draw[blue, ->-=0.5] (0, 1)--(1, 0);
\draw[blue, ->-=0.5] (1, 0)--(0, -1);
\draw[blue, ->-=0.5] (0, -1)--(-1, 0);
\draw[FA] (-1, 0)--(1, 0);
\draw[FA] (0, -1)--(0, 1);
\fill[FA] (0, 0) circle (0.05);
\draw[fill=white] (-0.3, 0.8) rectangle (0.3, 1.2);
\draw[fill=white] (-0.3, -0.8) rectangle (0.3, -1.2);
\begin{scope}[rotate=90]
\draw[fill=white] (-0.3, 0.8) rectangle (0.3, 1.2);
\draw[fill=white] (-0.3, -0.8) rectangle (0.3, -1.2);
\end{scope}
\draw (-1, 0) node{$\phi$};
\draw (1, 0) node{$\phi'$};
\draw (0, 1) node{$\psi$};
\draw (0, -1) node{$\psi'$};
\draw (-0.65, 0.65) node{\tiny$X_i$};
\draw (0.65, -0.65) node{\tiny$X_i$};
\draw (0.65, 0.65) node{\tiny$X_i$};
\draw (-0.65, -0.65) node{\tiny$X_i$};
\end{tikzpicture}
}}$$
where the unlabelled brown strand denotes the Frobenius algebra $Q$ and the fat dot denotes the multiplication morphism of $Q$.
\end{proposition}

\begin{proof}
By Lemma \ref{lem:traceformula} (the trace formula), we have
$$|\Irr{\mathcal{M}}|=\dim \V\left(
\vcenter{\hbox{
\begin{tikzpicture}
\path[fill=gray!50!white] (0, 0) [partial ellipse=0:360:1.2 and 1.2];
\draw [fill=white]
(0, 0) [partial ellipse=0:360:0.9 and 0.9];
\draw [fill=gray!50!white] (0, 0) [partial ellipse=0:360:0.6 and 0.6];
\draw (1.05, 0) node{\tiny$\mathfrak{A}$};
\draw (0.4, 0) node{\tiny$\mathfrak{B}$};
\end{tikzpicture}}}
\right)=\V\left(\vcenter{\hbox{
\begin{tikzpicture}
\path[fill=gray!50!white] (0, 0) [partial ellipse=0:360:1.2 and 1.2];
\draw [fill=white]
(0, 0) [partial ellipse=0:360:0.9 and 0.9];
\draw [fill=gray!50!white] (0, 0) [partial ellipse=0:360:0.6 and 0.6];
\draw (1.05, 0) node{\tiny$\mathfrak{A}$};
\draw (0.4, 0) node{\tiny$\mathfrak{B}$};
\end{tikzpicture}}}\times S^1\right),$$
where the alterfold that we would like to evaluate consists $B$-colored region homeomorphic to 
a torus times an interval.
To evaluate the value of the partition function, we apply Move $1$ to connect the boundary components. 
On the belt of the attached $1$-handle, we put $\displaystyle {}_{\mathfrak{A}}\Omega_{\mathfrak{B}}=\frac{1}{d_{J}}{}_{\mathfrak{A}}\Omega_{\mathfrak{A}}\otimes J$.
Shrink the $J$-colored loop on the separating surface with $\mathfrak{B}$-color, then we further apply the Planar Graphical Calculus to write $J\otimes \overline{J}$ as the Frobenius algebra $Q$ to get a tensor diagram over the separating surface colored by $\mathfrak{A}$.

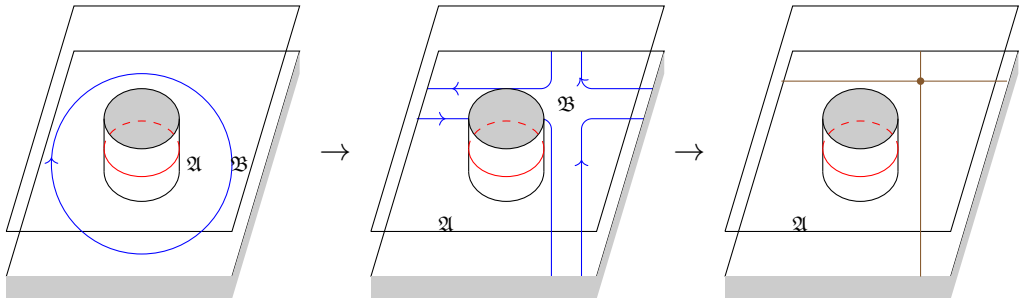
$$\vcenter{\hbox{
\begin{tikzpicture}
\draw (-1.8, -0.9) --(1.2, -0.9)-- (2.1, 2.1)--(-0.9, 2.1) --(-1.8, -0.9);
\draw (-1.8, -1.5) --(1.2, -1.5)-- (2.1, 1.5)--(-0.9, 1.5) --(-1.8, -1.5);
\path [fill=black!20!white] (-1.8, -1.5) --(1.2, -1.5)--(1.2, -1.8) --(-1.8, -1.8);
\path [fill=black!20!white] (1.2, -1.5)--(2.1, 1.5)--(2.1, 1.2)--(1.2, -1.8);
\begin{scope}[yshift=0.1cm]
\draw[fill=black!20!white] (0, 0.5) [partial ellipse=0:360:0.5 and 0.4];
\draw (-0.5, 0.5)--(-0.5, -0.2);
\draw (0.5, 0.5)--(0.5, -0.2);
\draw (0, -0.2) [partial ellipse=0:-180:0.5 and 0.4];
\end{scope}
\draw[red] (0, 0.2) [partial ellipse=0:-180:0.5 and 0.37];
\draw[red, dashed] (0, 0.2) [partial ellipse=0:180:0.5 and 0.37];
\draw [-<-=0.5, blue](0, 0) [partial ellipse=0:360:1.2 and 1.2];
\draw (0.7, 0) node{\tiny$\mathfrak{A}$};
\draw (1.3, 0) node{\tiny$\mathfrak{B}$};
\end{tikzpicture}
}}\rightarrow 
\vcenter{\hbox{
\begin{tikzpicture}
\draw (-1.8, -0.9) --(1.2, -0.9)-- (2.1, 2.1)--(-0.9, 2.1) --(-1.8, -0.9);
\draw[blue, ->-=0.1, rounded corners](-1.19, 0.6)--(0.6, 0.6)--(0.6, -1.5);
\draw[blue, -<-=0.2, rounded corners](-1.05, 1)--(0.6, 1)--(0.6, 1.5);
\draw[blue, -<-=0.3, rounded corners](1, 1.5)--(1, 1)--(2, 1);
\draw[blue, -<-=0.5, rounded corners](1.9, 0.6)--(1, 0.6)--(1, -1.5);
\draw (0.8, 0.8) node{\tiny$\mathfrak{B}$};
\draw (-0.8, -0.8) node{\tiny$\mathfrak{A}$};
\draw (-1.8, -1.5) --(1.2, -1.5)-- (2.1, 1.5)--(-0.9, 1.5) --(-1.8, -1.5);
\path [fill=black!20!white] (-1.8, -1.5) --(1.2, -1.5)--(1.2, -1.8) --(-1.8, -1.8);
\path [fill=black!20!white] (1.2, -1.5)--(2.1, 1.5)--(2.1, 1.2)--(1.2, -1.8);
\begin{scope}[yshift=0.1cm]
\draw[fill=black!20!white] (0, 0.5) [partial ellipse=0:360:0.5 and 0.4];
\draw (-0.5, 0.5)--(-0.5, -0.2);
\draw (0.5, 0.5)--(0.5, -0.2);
\draw (0, -0.2) [partial ellipse=0:-180:0.5 and 0.4];
\end{scope}
\draw[red] (0, 0.2) [partial ellipse=0:-180:0.5 and 0.37];
\draw[red, dashed] (0, 0.2) [partial ellipse=0:180:0.5 and 0.37];
\end{tikzpicture}}}
\rightarrow 
\vcenter{\hbox{
\begin{tikzpicture}
\draw (-1.8, -0.9) --(1.2, -0.9)-- (2.1, 2.1)--(-0.9, 2.1) --(-1.8, -0.9);
\draw (-1.8, -1.5) --(1.2, -1.5)-- (2.1, 1.5)--(-0.9, 1.5) --(-1.8, -1.5);
\path [fill=black!20!white] (-1.8, -1.5) --(1.2, -1.5)--(1.2, -1.8) --(-1.8, -1.8);
\path [fill=black!20!white] (1.2, -1.5)--(2.1, 1.5)--(2.1, 1.2)--(1.2, -1.8);
\begin{scope}[yshift=0.1cm]
\draw[fill=black!20!white] (0, 0.5) [partial ellipse=0:360:0.5 and 0.4];
\draw (-0.5, 0.5)--(-0.5, -0.2);
\draw (0.5, 0.5)--(0.5, -0.2);
\draw (0, -0.2) [partial ellipse=0:-180:0.5 and 0.4];
\end{scope}
\draw[red] (0, 0.2) [partial ellipse=0:-180:0.5 and 0.37];
\draw[red, dashed] (0, 0.2) [partial ellipse=0:180:0.5 and 0.37];
\draw (-0.8, -0.8) node{\tiny$\mathfrak{A}$};
\draw[FA] (-1.05, 1.1)--(1.95, 1.1);
\draw[FA] (0.8, -1.5)--(0.8, 1.5);
\fill[FA] (0.8, 1.1) circle(0.05);
\end{tikzpicture}}}
$$

One gets the first equality immediately by applying Move $2$ and $3$.
The second equality comes from the following equality.
$$|\Irr_{\mathcal{C}_{\mathcal{M}}^{*}}|=\dim \V\left(
\vcenter{\hbox{
\begin{tikzpicture}
\path[fill=gray!50!white] (0, 0) [partial ellipse=0:360:1.2 and 1.2];
\draw [fill=white]
(0, 0) [partial ellipse=0:360:0.9 and 0.9];
\draw [fill=gray!50!white] (0, 0) [partial ellipse=0:360:0.6 and 0.6];
\draw (1.05, 0) node{\tiny$\mathfrak{B}$};
\draw (0.4, 0) node{\tiny$\mathfrak{B}$};
\end{tikzpicture}}}
\right)=\V\left(\vcenter{\hbox{
\begin{tikzpicture}
\path[fill=gray!50!white] (0, 0) [partial ellipse=0:360:1.2 and 1.2];
\draw [fill=white]
(0, 0) [partial ellipse=0:360:0.9 and 0.9];
\draw [fill=gray!50!white] (0, 0) [partial ellipse=0:360:0.6 and 0.6];
\draw (1.05, 0) node{\tiny$\mathfrak{B}$};
\draw (0.4, 0) node{\tiny$\mathfrak{B}$};
\end{tikzpicture}}}\times S^1\right).$$
\end{proof}

The $n$-positivity for unitary fusion category is also true for unitary module category.
\begin{corollary}
Let $\mathcal{C}$ be a unitary fusion category and $\mathcal{M}$ be a unitary module category of $\mathcal{C}$. 
Then for all $n\in \mathbb{N}$, we have
$$\sum_{j=0}^{r}d_{j}^{-n+2}N_j^{\otimes n}\ge 0,$$
where $N_j$ is the fusion matrix of $\mathcal{C}\times \mathcal{M}\rightarrow \mathcal{M}$ for the simple object $X_j\in\mathcal{C}$ acting on $\mathcal{K}_{0}\mathcal{(M)}$ and $d_j=d(X_j)$.
\end{corollary}
\begin{proof}
It follows from Proposition \ref{prop:positivitycriterion}. 
\end{proof}

In \cite{LMWW23}, for any given $(m,r)\in\bZ^2$, $V \in \mathcal{C}$ and $(X, e_X) \in \mathcal{Z(C)}$,  we implemented the associated generalized Frobenius-Schur indicator as the value of the partition function of the following alterfold
\begin{equation}\label{eq:gamma-L}
(\bS^3, \Sigma_L, \Gamma_{(m,r)}(V) \sqcup \Phi(X, e_X)) 
:= 
\vcenter{\hbox{\scalebox{0.8}{
\begin{tikzpicture}[scale=0.8]
\draw [double distance=0.8cm] (0,0) [partial ellipse=-0.1:180.1:2 and 1.5];
\draw[blue] (0,0) [partial ellipse=0:180:1.8 and 1.3];
\draw[blue] (0,0) [partial ellipse=0:180:2.2 and 1.7];
\draw [blue] (0,0) [partial ellipse=0:180:2.1 and 1.6];
\begin{scope}[shift={(2, 0)}]
\draw [double distance=0.8cm] (0,0) [partial ellipse=-0.1:180.1:2 and 1.5];
\draw [red] (0,0) [partial ellipse=0:180:2 and 1.5];
\end{scope} 
\begin{scope}[shift={(2, 0)}]
\draw [line width=0.83cm] (0,0) [partial ellipse=180:360:2 and 1.5];
\draw [white, line width=0.8cm] (0,0) [partial ellipse=178:362:2 and 1.5];
\draw [red] (0,0) [partial ellipse=178:362:2 and 1.5];
\end{scope}
\draw [line width=0.83cm] (0,0) [partial ellipse=180:360:2 and 1.5];
\draw [white, line width=0.8cm] (0,0) [partial ellipse=178:362:2 and 1.5];
\draw[blue] (0,0) [partial ellipse=178:362:1.8 and 1.3];
\draw[blue] (0,0) [partial ellipse=178:362:2.2 and 1.7];
\draw [blue] (0,0) [partial ellipse=178:362:2.1 and 1.6];
\begin{scope}[shift={(-1.98, 0.2)}]
\draw [blue, dashed](0,0) [partial ellipse=0:180:0.5 and 0.3];
\draw [blue] (0,0) [partial ellipse=180:360:0.5 and 0.3];
\end{scope} 
\begin{scope}[shift={(-1.99, -0.1)}]
\draw [blue, dashed](0,0) [partial ellipse=180:145:0.5 and 0.3];
\draw [blue, dashed](0,0) [partial ellipse=0:45:0.5 and 0.3];
\draw [blue] (0,0) [partial ellipse=180:360:0.5 and 0.3];
\end{scope} 
\begin{scope}[shift={(-1.98, -0.2)}]
\draw [fill=white] (-0.28,-0.3)--(-0.31, 0.3)--(0.36,0.3)--(0.39,-0.3)--(-0.28,-0.3);
\end{scope}
\begin{scope}[shift={(4, 0)}]
\draw [blue, dashed](0,0) [partial ellipse=0:180:0.5 and 0.3]; 
\node at (-0.75,0) {\tiny $X$};
\draw [white, line width=4pt] (0,0) [partial ellipse=270:290:0.5 and 0.3];
\fill[white] (-0.05,-0.3) circle[radius=2pt];
\draw [blue, ->-=0.3] (0,0) [partial ellipse=180:360:0.5 and 0.3];
\fill [blue] (0.2,-0.275) circle[radius=2pt];
\end{scope}
\node at (-1.8,-2.2) {\tiny $\tilde{V}$}; \node at (-0.4,1.4) {\tiny $|m|$};
\end{tikzpicture}}}}\,,
\end{equation}
where $\Sigma_L$ is the union the of the two linked tori, $\Gamma_{(m,r)}(V)$ is the graph of an $(m,r)$-(multi)curve on the left torus colored by $\tilde{V} = V^{\operatorname{sgn}(m)}$, and $\Phi(X,e_X)$ is the graph on the right torus representing $(X,e_X) \in \mathcal{Z(C)}$. 
By \cite[Corollary 5.5]{NgSch10}, 
$$\nu_{n,k}^{(X,e_X)}(V) = \frac{1}{\mu}Z(\bS^3, \Sigma_L, \Gamma_{(m,r)}(V) \sqcup \Phi(X, e_X))\,.$$
In alterfold theory, if a closed curve is colored by a linear combination of objects in $\mathcal{C}$, then its partition function is the linear combination of the partition functions of the corresponding $\mathcal{C}$-decorated alterfolds. 
Thus, $Z(\bS^3, \Sigma_L, \Gamma_{(m,r)}(v) \sqcup \Phi(z))$ is well-defined for all $v \in K_0(\mathcal{C})$ and $z \in K_0(\mathcal{Z(C)})$. 
Note that we may need to change $v$ by $v^*$ according to the sign of $m$, and in $\Gamma_{(m,r)}(v)$, we color each curve component by $v$. 

Inspired by the above alterfold interpretation, we give the following definition. 

\begin{definition}
Let $\mathcal{C}$ be a spherical fusion category. 
For all $(m,r) \in \mathbb{Z}^2$, $v \in K_0(\mathcal{C})$ and $z \in K_0(\mathcal{Z(C)})$, we define its genus-1 {topological indicator} by
\[\mathcal{I}_{v}^{1}((m,r),z) := \frac{1}{\mu}Z(\bS^3, \Sigma_L, \Gamma_{(m,r)}(v) \sqcup \Phi(z))\,.\]
\end{definition}

When we take $v \in \operatorname{Ob}(\cC) \subset K_0(\cC)$ (where we identify  $\operatorname{Ob}(\cC)$ with $\bZ_{\ge 0}[\Irr(\cC)]\subset K_0(\cC)$), then we recover the equivariant indicator in \cite[Definition 5.1]{NgSch10}, which is a linearization of the generalized Frobenius-Schur indicator in the $z$-component. 
Note that $\mathcal{I}^{1}$ is linear in $z \in K_0(\mathcal{Z(C)})$, but it is not necessarily linear in $v \in K_0(\mathcal{C})$. 
For example, if $\lambda \in \bC$, then 
\begin{equation}\label{eq:I1-lambda}
\mathcal{I}^{1}_{\lambda v}((m,r), z) = \lambda^{d}\mathcal{I}^{1}_{v}((m,r),z)\,,
\end{equation}
where $d = \gcd(m,r)$ is equal to the number of curve components of the $(m,r)$-multicurve. 
In fact, $\mathcal{I}^1$ is in general not even additive, see \cite[Remark 2.4, Proposition 2.8]{NgSch10}. 
In particular, despite $\displaystyle \Omega_\cC = \sum_{Y \in \Irr_\cC} d_Y \cdot Y \in K_0(\cC)$, for arbitrary $(m,r) \in\bZ^2$ and $z \in K_0(\mathcal{Z(C)})$, $\mathcal{I}^{1}_{\Omega_{\cC}}((m,r), z)$ may not be equal to $\displaystyle \sum_{Y \in \Irr_\cC} d_Y \cdot \mathcal{I}^{1}_{Y}((m,r), z)$.

\begin{proposition}\label{lem:I1-1}
Let $\mathcal{E}$ be a 2-category. 
Then for all $(m,r) \in \mathbb{Z}^2$ with $d = \gcd(m,r)$, $v \in K_0(\mathcal{C})$ and $z \in K_0(\mathcal{Z(C)})$, we have 
\begin{itemize}
\item[(1)] $\mathcal{I}_{\Omega_{\mathcal{C}}}^1((m,r), z) = \mathcal{I}_{\Omega_{\mathcal{D}}}^1((m,r), z)$.
\item[(2)] For all $\ell \in \mathbb{Z}$, $\lambda \in \bC$ and $V \in \operatorname{Ob}(\cC)$, $\mathcal{I}_{(\lambda V)^{\ell}}^1((m,r),z) = \lambda^{d\ell}\mathcal{I}_{V}^1((m\ell,r\ell),z)$. In particular, if $m$ and $r$ are coprime, then $\mathcal{I}^1_{(-)}((m,r),z): K_0(\cC) \to \bC$ is a linear functional.
\item[(3)] If $m$ and $r$ are coprime, then for all $(X, e_X) \in \mathcal{Z(C)}$,
$$\sum_{Y\in \Irr_{\mathcal{C}}}\nu_{m,r}^{(X, e_X)}(Y) d_{Y}
=\sum_{V\in \Irr_{\mathcal{D}}}\nu_{m,r}^{(X, e_X)}(V) d_{V}\,.$$
\end{itemize}
\end{proposition}
\begin{proof}
(1) follows immediately from the definition and Corollary \ref{cor:omega}. 

(2) By Equation \eqref{eq:I1-lambda}, 
$$\mathcal{I}^{1}_{(\lambda V)^\ell}((m,r), z) = \mathcal{I}^{1}_{\lambda^{\ell} V^\ell}((m,r), z) = \lambda^{d\ell}\mathcal{I}^{1}_{V^\ell}((m,r), z)\,.$$
Moreover, the $(m\ell, r\ell)$-curve on a torus is nothing but $\ell$-parallel copies of the $(m, r)$-curve, so by the graphical calculus of alterfolds, 
$$
Z(\bS^3, \Sigma_L, \Gamma_{(m,r)}(V^{\ell}) \sqcup \Phi(z)) = Z(\bS^3, \Sigma_L, \Gamma_{(m\ell,r\ell)}(V) \sqcup \Phi(z))
$$ 
and so $\mathcal{I}^{1}_{(\lambda V)^\ell}((m,r), z)= \lambda^{d\ell}\mathcal{I}^{1}_{V^\ell}((m,r), z) = \lambda^{d\ell}\mathcal{I}^{1}_{V}((m\ell,r\ell), z)\,.$ If $d = 1$, then $I^1_{(-)}((m,r),z)$ preserves scalar multiplication. Moreover, by \cite[Corollary 5.5]{NgSch10}, in this case, $I^1_{(-)}((m,r),z)$ is additive, so it is a linear functional.

(3) follows immediately from (1), (2), and the definition of the $\Omega$-color.
\end{proof}

\section{Topological Indicators}

We start by setting up terminologies for curves on $\Sigma_g$, the surface of genus-$g$, which is more or less standard. The reader is referred to \cite{Mart16} for details. 

A simple closed curve in $\Sigma_g$ is a proper embedding $S^1 \to \Sigma_g$, and it is called {nontrivial} if it does not bound a disk in $\Sigma_g$. A multicurve is a finite disjoint union of nontrivial simple closed curves, and it is called {essential} if it contains no parallel components. In this paper, we consider oriented nontrivial multicurves with an ordering on its components, and we write $\vec{\gamma} = \{\gamma_1, ..., \gamma_n\}$ for such a multicurve consisting of $n$ oriented simple closed components on $\Sigma_g$.

One can generalize the genus-1 topological indicators by considering the following generalization of $\Gamma_{(m,r)}$ and $\Phi$ from the last section. Let $\vec{\gamma} \in \Sigma_g$ be a multicurve consisting of $n$ components. For all $\vec{V} = (V_1, ..., V_{n}) \in \Irr_{\cC}^{n}$, define $\Gamma_{\vec{V}}(\vec{\gamma})$ to be a $\mathcal{C}$-diagram on $\Sigma_g$ depicted below
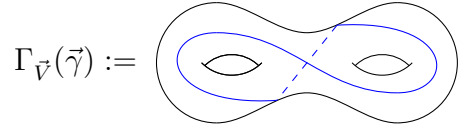
\begin{equation*}
\Gamma_{\vec{V}}(\vec{\gamma}) :=\vcenter{\hbox{
\begin{tikzpicture}[scale=0.4]
\draw(0, -1) ..controls (-1, -1) and (-2, -2).. (-3, -2);
\draw (-3, 2) arc (90: 270: 2);
\draw(0, 1) ..controls (-1, 1) and (-2, 2).. (-3, 2);
\draw(0, -1) ..controls (1, -1) and (2, -2).. (3, -2);
\draw (3, -2) arc (-90: 90: 2);
\draw(0, 1) ..controls (1, 1) and (2, 2).. (3, 2);
\draw(-3.5, 0) to[bend right=45] (-1.5, 0);
\draw(-3.4, -0.1) to[bend left=45] (-1.6, -0.1);
\draw(-3.5, 0) to[bend right=45] (-1.5, 0);
\draw(-3.4, -0.1) to[bend left=45] (-1.6, -0.1);
\draw(3.5, 0) to[bend left=45] (1.5, 0);
\draw(3.4, -0.1) to[bend right=45] (1.6, -0.1);
\draw[blue](-3, 1)..controls (-0.5, 1) and (0.5, -1)..(3, -1);
\draw[blue](3, -1)..controls (5, -1) and (5, 1.5)..(1, 1.25);
\draw[blue](-3, 1)..controls (-5, 1) and (-5, -1.5)..(-1, -1.25);
\draw[blue, dashed] (1, 1.25) to (-1, -1.25);
\end{tikzpicture}}}   
\end{equation*}
Next, let $z$ be a basis vector of 
$$\hom_{\mathcal{Z(C)}}\left(\1_{\mathcal{Z(C)}}, \bigoplus_{X_i \in \Irr_{\mathcal{Z(C)}}}(X_i \otimes X_i^*)^{\otimes g}\right)\,.$$ 
By Theorem \ref{thm:rtbasis}, the 3-alterfold $\Phi(z)$ (or rather, its equivalence class) is a basis for $\V(\Sigma_g, \emptyset,\emptyset)$, which is also a basis for $\operatorname{RT}_{\mathcal{Z(C)}}(\Sigma_g)$ (see Theorem \ref{thm:embedding}). 
We denote the tensor diagram on $\Phi(z)$ by $\Theta(z)$, so that $\Phi(z) = (\bS^3, \Sigma_g, \Theta(z))$.

Clearly, the definition of the alterfolds $\Gamma_{\vec{\gamma}}(\vec{V})$ and $\Phi(z)$ can be extended so that $\vec{V} \in \operatorname{Ob}(\mathcal{C})^n$ and $z \in \operatorname{RT}_{\mathcal{Z(C)}}(\Sigma_g)$. Now we consider a 3-alterfold $(\bS^3, \Sigma_{g,1}\sqcup\Sigma_{g,2}, \Gamma_{\vec{V}}(\vec{\gamma}) \sqcup \Theta(z))$ whose separating surfaces are decorated by diagrams described above, and are ``linked handle by handle'', so that the $B$-colored region in $S^3$ is homeomorphic to $\Sigma_g \times [0,1]$. An example of such a 3-alterfold with genus 2 separating surfaces is given below, where the decoration on the top surface is understood as a morphism in $\Hom_{\mathcal{Z(C)}}(\1_{\mathcal{Z(C)}}, X_i \otimes X_i^* \otimes X_j \otimes X_j^*)$, where $X_i$, $X_j$, $X_k$ are elements of $\mathcal{Z(C)}$.
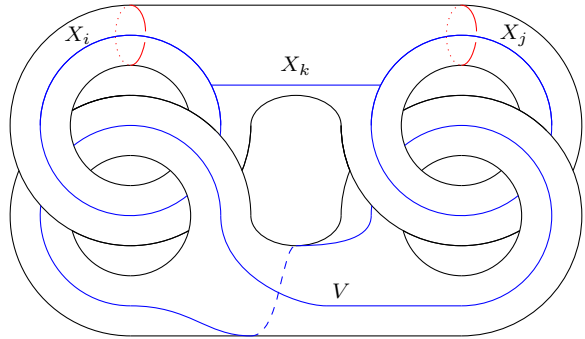
\[
\begin{tikzpicture}
\begin{scope}[scale=0.05]
  \draw
    (56.0001, 24)
     { [rotate=90] arc[start angle=180, end angle=360, x radius=32, y radius=-32] };
  \draw
    (56.0001, 40)
     { [rotate=90] arc[start angle=180, end angle=360, x radius=16, y radius=-16] };
  \draw
    (56, 0)
     { [rotate=90] arc[start angle=180, end angle=292.0242, x radius=32, y radius=-32] };
  \draw
    (56.0001, 16)
     { [rotate=90] arc[start angle=180, end angle=255.5225, x radius=16, y radius=-16] };
  \draw
    (56.0001, 48)
     arc[start angle=90, end angle=131.4096, radius=16];
  \draw
    (56.0001, 64)
     arc[start angle=90, end angle=118.9551, radius=32];
  \draw
    (56.0001, 48)
     { [rotate=90] arc[start angle=0, end angle=180, x radius=16, y radius=-16] };
  \draw
    (56.0001, 40)
     arc[start angle=270, end angle=311.4096, radius=16];
  \draw
    (56.0001, 72)
     { [rotate=90] arc[start angle=0, end angle=75.5225, x radius=16, y radius=-16] };
  \draw
    (56.0001, 24)
     arc[start angle=270, end angle=298.9551, radius=32];
  \draw
    (40.5081, 60)
     { [rotate=90] arc[start angle=-28.9551, end angle=45, x radius=32, y radius=-32] };
  \draw
    (40.5081, 28.0001)
     arc[start angle=241.0449, end angle=298.9551, radius=32];
  \draw
    (144.0001, 40)
     { [rotate=90] arc[start angle=180, end angle=360, x radius=16, y radius=-16] };
  \draw
    (144.0001, 16)
     { [rotate=90] arc[start angle=180, end angle=255.5225, x radius=16, y radius=-16] };
  \draw
    (144.0001, 48)
     arc[start angle=90, end angle=131.4096, radius=16];
  \draw
    (144.0001, 64)
     arc[start angle=90, end angle=118.9551, radius=32];
  \draw
    (144.0001, 48)
     { [rotate=90] arc[start angle=0, end angle=180, x radius=16, y radius=-16] };
  \draw
    (144.0001, 40)
     arc[start angle=270, end angle=311.4096, radius=16];
  \draw
    (144.0001, 72)
     { [rotate=90] arc[start angle=0, end angle=75.5225, x radius=16, y radius=-16] };
  \draw
    (144.0001, 24)
     arc[start angle=270, end angle=298.9551, radius=32];
  \draw
    (128.5081, 60)
     { [rotate=90] arc[start angle=-28.9551, end angle=45, x radius=32, y radius=-32] };
  \draw
    (128.5081, 28.0001)
     arc[start angle=241.0449, end angle=298.9551, radius=32];
  \draw
    (144.0001, 64)
     { [rotate=90] arc[start angle=0, end angle=180, x radius=32, y radius=-32] };
  \draw[shift={(56, 0)}, rotate=90]
    (0, 0)
     -- (0, -88);
  \draw
    (40.5081, 60)
     { [rotate=90] arc[start angle=-28.9551, end angle=90, x radius=32, y radius=-32] };
  \draw
    (114.3352, 43.9999)
     arc[start angle=157.9758, end angle=180, radius=32];
  \draw
    (85.6649, 44.0001)
     arc[start angle=-22.0242, end angle=0.0036, radius=32];
  \draw
    (144.0001, 88)
     { [rotate=90] arc[start angle=0, end angle=112.0242, x radius=32, y radius=-32] };
  \draw[shift={(56, 88)}, rotate=90]
    (0, 0)
     -- (0, -88);
  \draw
    (88.0001, 56.002)
     { [rotate=90] arc[start angle=89.9964, end angle=112.0242, x radius=32, y radius=-32] };
  \draw
    (112.0005, 56)
     arc[start angle=180, end angle=202.0154, radius=32.0116];
  \draw
    (112.0001, 56)
     arc[start angle=180, end angle=202.0242, radius=32];
  \draw
    (112.0001, 56)
     arc[start angle=180, end angle=270, radius=32];
  \draw
    (88, 55.9998)
     { [rotate=90] arc[start angle=-90, end angle=90, x radius=8, y radius=-12] };
  \draw
    (88, 31.9997)
     { [rotate=90] arc[start angle=90, end angle=270, x radius=8, y radius=12] };
  \draw[blue]
    (79.8514, 53.333)
     arc[start angle=-6.3803, end angle=308.942, radius=24];
  \draw[shift={(77.466, 66.733)}, rotate=90, blue]
    (0, 0)
     -- (0, -45.068);
  \draw[blue]
    (167.8514, 53.333)
     arc[start angle=-6.3803, end angle=308.942, radius=24];
  \draw[red, dotted]
    (56.0001, 87.9999)
     { [rotate=90] arc[start angle=0, end angle=180, x radius=8, y radius=4] };
  \draw[red, dotted]
    (144.0001, 87.9999)
     { [rotate=90] arc[start angle=0, end angle=180, x radius=8, y radius=4] };
  \draw[blue]
    (32.0001, 56)
     { [rotate=90] arc[start angle=-90, end angle=90, x radius=24, y radius=-24] };
  \draw[blue]
    (120.0001, 56)
     { [rotate=90] arc[start angle=-90, end angle=90, x radius=24, y radius=-24] };
  \draw[blue]
    (80.0001, 32)
     arc[start angle=0, end angle=128.942, radius=24];
  \draw[blue]
    (32.1487, 34.6671)
     arc[start angle=173.6197, end angle=270, radius=24];
  \draw[shift={(56, 8)}, rotate=90, blue]
    (0, 0)
     .. controls (0, -16) and (-8, -20) .. (-8, -32);
  \draw[shift={(87.999, 0)}, rotate=90, blue, dashed]
    (0, 0)
     .. controls (0, -8.001) and (24, -8.001) .. (24, -12.001);
  \draw[shift={(80, 32)}, rotate=90, blue]
    (0, 0)
     .. controls (-16, 0) and (-24, -20) .. (-24, -28);
  \draw[shift={(120, 32)}, rotate=90, blue]
    (0, 0)
     .. controls (-5.3333, 0) and (-8, 6.6667) .. (-8, 20);
  \draw[shift={(108, 8)}, rotate=90, blue]
    (0, 0)
     -- (0, -36);
  \draw[blue]
    (120.1487, 34.667)
     arc[start angle=173.6197, end angle=180, radius=24];
  \draw[blue]
    (144.0001, 8)
     arc[start angle=270, end angle=488.942, radius=24];
  \draw[red]
    (144.0003, 87.9995)
     { [rotate=90] arc[start angle=0, end angle=90, x radius=6.9758, y radius=-3.9992] };
  \draw[red]
    (147.8243, 78.1542)
     { [rotate=90] arc[start angle=116.5683, end angle=180, x radius=11.1343, y radius=-4.2758] };
  \draw[red]
    (56.0003, 87.9995)
     { [rotate=90] arc[start angle=0, end angle=90, x radius=6.9758, y radius=-3.9992] };
  \draw[red]
    (59.8243, 78.1542)
     { [rotate=90] arc[start angle=116.5683, end angle=180, x radius=11.1343, y radius=-4.2758] };
  \end{scope}
\node at (5.6,0.6) {\tiny $V$};
\node at (2.1,4) {\tiny $X_{i}$};
\node at (5,3.6) {\tiny $X_{k}$};
\node at (7.9,4) {\tiny $X_{j}$};
\end{tikzpicture}
\]

\begin{definition}
Let $\mathcal{C}$ be a spherical fusion category and $g \ge 1$ an integer. For any oriented nontrivial multicurve $\vec{\gamma} = \{\gamma_1, .., \gamma_n\} \subset \Sigma_g$ of $n$ components, $\vec{V} \in \operatorname{Ob}(\mathcal{C})^{n}$ and $z \in \operatorname{RT}_{\mathcal{Z(C)}}(\Sigma_g)$, we define the corresponding genus-$g$ {topological indicator} by
\[\mathcal{I}_{\vec{V}}^{g}(\vec{\gamma},z) := \frac{1}{\mu^g} Z(\bS^3, \Sigma_{g,1}\sqcup\Sigma_{g,2}, \Gamma_{\vec{V}}(\vec{\gamma}) \sqcup \Theta(z))\,.\]
\end{definition}

When $g = 1$, nontrivial multicurves on $\Sigma_1$ are parameterized by pairs of integers, and $K_0(\mathcal{Z(C)})$ is naturally isomorphic to $\operatorname{RT}_{\mathcal{Z(C)}}(\Sigma_1)$, so the above definition generalizes the genus-1 topological indicator. 

For simplicity, we will use the abused notation to write $\mathfrak{f}: \Sigma_g \to \Sigma_g$ for both an orientation-preserving homeomorphism and its mapping class in $\operatorname{MCG}(\Sigma_g)$. Note that, on one hand, $\operatorname{MCG}(\Sigma_g)$ naturally acts on multicurves on $\Sigma_g$; on the other hand, the TQFT axioms imply that $\operatorname{MCG}(\Sigma_g)$ acts on $\operatorname{RT}_{\mathcal{Z(C)}}(\Sigma_g)$. Therefore, it is desirable to investigate the relationship between the two mapping class group actions. The main theorem in this section is the following generalization of the equivariance property of the genus-1 topological indicators (cf.~\cite[Theorem 5.14]{LMWW23}). For all $\mathfrak{f} \in \operatorname{MCG}(\Sigma_g)$, and a multicurve $\vec{\gamma}\subset\Sigma_g$, we denote the image of $\vec{\gamma}$ under $\mathfrak{f}$ by $\mathfrak{f}(\vec{\gamma})$, which is understood as $(\mathfrak{f}(\gamma_1), ..., \mathfrak{f}(\gamma_n))$ if $\vec{\gamma}$ has $n$ components. Similarly, to simplify notations, we denote the TQFT-action of $\mathfrak{f}$ on an element $z \in \operatorname{RT}_{\mathcal{Z(C)}}(\Sigma_g)$ by $\mathfrak{f}(z)$, which is a linear combination of a chosen basis of $\operatorname{RT}_{\mathcal{Z(C)}}(\Sigma_g)$. Let $\mathfrak{j}: \Sigma_g \to \Sigma_g$ be an orientation-reversing homeomorphism implemented by the cylinder $\Sigma_g \times [0,1]$.

\begin{theorem}[$\operatorname{MCG}(\Sigma_g)$-Equivariance]
Let $\vec{\gamma} \subset \Sigma_{g}$ be any oriented nontrivial multicurve containing $n$ components, $\vec{V} \in \operatorname{Ob}(\mathcal{C})^{n}$, and $z \in \operatorname{RT}_{\mathcal{Z(C)}}(\Sigma_g)$. Then for all $\mathfrak{f} \in \operatorname{MCG}(\Sigma_g)$, we have 
\begin{align*}
\mathcal{I}_{\vec{V}}^{g}(\mathfrak{f}(\vec{\gamma}),z) 
= \mathcal{I}_{\vec{V}}^{g}(\vec{\gamma},\tilde{\mathfrak{f}}(z))\,,
\end{align*}
where $\tilde{\mathfrak{f}} = \mathfrak{j}\circ\mathfrak{f}\circ\mathfrak{j}$ is an orientation-preserving homeomorphism of $\Sigma_g$.
\end{theorem}
\begin{proof}
We start with the following equality of $B$-colored 3-alterfolds with boundary:
\begin{equation*}\begin{split}
& R_B\left(\bS^3, \mathfrak{f}(\Sigma_{g,1}) \sqcup \Sigma_{g,2}, \Gamma_{\vec{V}}(\mathfrak{f}(\vec{\gamma})) \sqcup \Theta(z)\right) 
= \left(C_{\mathfrak{f},1}, \emptyset, \Gamma_{\vec{V}}(\vec{\gamma}) \sqcup \Theta(z)\right)
= \left(C_{\tilde{\mathfrak{f}}, 2}, \emptyset, \Gamma_{\vec{V}}(\vec{\gamma}) \sqcup \Theta(z)\right)\,
\end{split}\end{equation*}
where $C_{\mathfrak{f},1}$ is the mapping cylinder of $\mathfrak{f}$ on $\Sigma_{g,1}$, and $C_{\tilde{\mathfrak{f}}, 2}$ is the mapping cylinder of $\tilde{\mathfrak{f}}$ on $\Sigma_{g,2}$, and the second equality in the above equation means we can push the action of $\mathfrak{f}$ from $\Sigma_{g,1}$ to $\Sigma_{g,2}$ with a conjugation of $\mathfrak{j}$, as theses two surfaces have opposite orientations.

Write $\left(C_{\tilde{\mathfrak{f}}}, \emptyset, \Gamma_{\vec{V}}(\vec{\gamma}) \sqcup \Theta(z)\right)_B$ to indicate that $C_{\tilde{\mathfrak{f}}}$ is $B$-colored. Then by the Homeomorphism property in Theorem \ref{thm:partition function}, 
we have 
\[Z(\bS^3, \mathfrak{f}(\Sigma_{g,1}) \sqcup \Sigma_{g,2}, \Gamma_{\vec{V}}(\mathfrak{f}(\vec{\gamma})) \sqcup \Theta(z)) 
= Z\left(\left(C_{\tilde{\mathfrak{f}},2}, \emptyset, \Gamma_{\vec{V}}(\vec{\gamma}) \sqcup \Theta(z)\right)_B \sqcup \left({H}_{g,1} \sqcup {H}_{g,2}, \emptyset, \emptyset\right)_A \right) \,,\]
where $\left({H}_{g,1} \sqcup {H}_{g,2}, \emptyset, \emptyset\right)_A$ is a disjoint union of $A$-colored handlebodies of genus $g$.
On the other hand, by the definition of the TQFT action of the mapping class group,
\[Z\left(\left(C_{\tilde{\mathfrak{f}}}, \emptyset, \Gamma_{\vec{V}}(\vec{\gamma}) \sqcup \Theta(z)\right)_B \sqcup \left({H}_{g,1} \sqcup {H}_{g,2}, \emptyset, \emptyset\right)_A \right) 
= 
Z\left(\bS^3, \Sigma_{g,1}\sqcup\Sigma_{g,2}, \Gamma_{\vec{V}}(\vec{\gamma}) \sqcup \Theta(\tilde{\mathfrak{f}}(z))\right)\,,\]
and this completes the proof.
\end{proof}

\begin{proposition}
Suppose $\mathcal{C}, \mathcal{D}$ are Morita equivalent spherical fusion categories.
Then for any simple closed curve $\gamma$ on $\Sigma_g$, and $z \in \operatorname{RT}_{\mathcal{Z(C)}}(\Sigma_g)$, we have
    \begin{align*}
    \sum_{V\in \Irr_{\mathcal{C}}}\mathcal{I}_{V}^{g}(\gamma,z) d_V
    = \sum_{Y\in \Irr_{\mathcal{D}}}\mathcal{I}_{Y}^{g}(\gamma,z) d_Y\,.
    \end{align*}
\end{proposition}
\begin{proof}
The argument is similar to the one of Proposition \ref{lem:I1-1}.
\end{proof}

Suppose $\Sigma_g$ is a surface of genus $g$. Recall that $\operatorname{TV}_{\mathcal{C}}(\Sigma_g)$ is generated by equivalence classes of 3-alterfolds with time boundary $\Sigma_g$. We denote $\operatorname{TV}^{\operatorname{cv}}_{\mathcal{C}}(\Sigma_g)$ to be the subspace of $\operatorname{TV}_{\mathcal{C}}(\Sigma_g)$ generated by (equivalence classes) of 3-alterfolds with time boundary of the form $(H_g, \Sigma_g, \Gamma_{\vec{\gamma}}(\vec{V}))$, where $H_g$ is a genus-$g$ handlebody,  $\Sigma_g$ is an embedded surface parallel to the time boundary of $H_g$ (homeomorphic to $\Sigma_g$), $\vec{\gamma}$ is a nontrivial multicurve of $n$ components and $\vec{V} \in \Irr_{\mathcal{C}}^n$. 
In other words, we consider those 3-alterfolds with $\mathcal{C}$-decorated curves, and for completeness, we also allow empty collection of curves. Note that $\MCG(\Sigma_g)$ maps curves to curves, so $\operatorname{TV}_{\mathcal{C}}^{\operatorname{cv}}(\Sigma_g)$ is a nonzero invariant subspace of $\operatorname{TV}_{\mathcal{C}}(\Sigma_g)$. Using skein theory \cite{KauLin94}, it is easy to see that when $\mathcal{C}$ is the quantum group modular tensor category $\mathcal{C}(\mathfrak{sl}_2, q)$ for some root of unity $q$ \cite{Row06}, then  $\operatorname{TV}_{\mathcal{C}}^{\operatorname{cv}}(\Sigma_g) = \operatorname{TV}_{\mathcal{C}}(\Sigma_g)$.

\begin{question}
For which $\mathcal{C}$ and $g$, is it true that  $\operatorname{TV}_{\mathcal{C}}^{\operatorname{cv}}(\Sigma_g)= \operatorname{TV}_{\mathcal{C}}(\Sigma_g)$?
\end{question}

\begin{remark}
In forthcoming papers, we will present a detailed study and explicit computations of high genus topological indicators. The topological nature of these numerical invariants indicates the possibility of revealing categorical information of modular tensor categories that is beyond modular data. Moreover, through the equivariance, we can translate the mapping class group action on the TQFT spaces to its natural actions on curves on surfaces, which can potentially help us understand the properties of quantum representations of mapping class groups better.
\end{remark}

 \bibliographystyle{abbrv}
 \bibliography{TQFT}

 \end{document}